\documentclass[10pt,aps,prx,twocolumn,superscriptaddress,floatfix,citeautoscript,preprintnumbers]{revtex4-2}
\usepackage{makecell}
\usepackage{array}
\usepackage{amsmath}
\usepackage{amsthm}
\usepackage{amssymb}
\usepackage{algorithm}
\usepackage{algpseudocode}
\usepackage{braket}
\usepackage{bbm}
\usepackage{comment}
\usepackage{color}
\usepackage{dsfont}
\usepackage{epstopdf}
\usepackage{enumitem}
\usepackage[T1]{fontenc}
\usepackage{hyperref}
\hypersetup{
    colorlinks=true,
    citecolor=blue,
    linkcolor=blue,
    filecolor=blue,
    urlcolor=blue,
}
\usepackage[nameinlink,capitalize]{cleveref}
\usepackage{mathtools}
\usepackage{multirow}
\usepackage{qcircuit}
\usepackage{textcomp}
\usepackage{tabularx}
\usepackage{url}
\usepackage{wrapfig,epsfig}

\let\C\relax
\usepackage{tikz}
\usepackage{circuitikz}
\usetikzlibrary{arrows}

\usepackage[caption=false]{subfig}

\makeatletter
\newcommand{\multiline}[1]{%
  \parbox{\dimexpr\linewidth-4em-\ALG@thistlm}{#1}
}
\makeatother

\newtheorem{theorem}{Theorem}[section]
\newtheorem{lemma}[theorem]{Lemma}
\newtheorem{definition}[theorem]{Definition}

\newtheorem{proposition}[theorem]{Proposition}

\newtheorem{assumption}{Assumption}
\newtheorem{problem}[theorem]{Problem}
\theoremstyle{remark}

\newtheorem{remark}[theorem]{Remark}

\newcommand{\wt}{\widetilde}
\renewcommand{\tilde}{\wt}

\newcommand{\R}{\mathbb{R}}
\newcommand{\C}{\mathbb{C}}

\renewcommand{\d}{\mathrm{d}}
\renewcommand{\i}{\mathbf{i}}

\newcommand{\diag}{\mathrm{diag}}
\renewcommand{\d}{\mathrm{d}}

\newcommand{\norm}[1]{\| #1 \|}
\DeclareMathOperator*{\E}{{\mathbb{E}}}

\DeclareMathOperator{\tr}{tr}
\DeclareMathOperator*{\argmin}{arg\,min}
\DeclareMathOperator*{\argmax}{arg\,max}

\newsavebox{\mstrut}
\newcommand{\bbra}[1]{%
    \sbox{\mstrut}{\(#1\)}%
    \mathinner{\left\langle\kern-0.2\ht\mstrut\left\langle{#1}\right|\right.}%
}
\newcommand{\kett}[1]{%
    \sbox{\mstrut}{\(#1\)}%
    \mathinner{\left.\left|{#1}\right\rangle\kern-0.2\ht\mstrut\right\rangle}%
}
\newcommand{\bbrakett}[1]{%
    \sbox{\mstrut}{\(#1\)}%
    \mathinner{\left\langle\kern-0.2\ht\mstrut\left\langle{#1}\right\rangle\kern-0.2\ht\mstrut\right\rangle}%
}

\newcommand{\RZ}[1]{{\color{cyan}[Ruizhe: #1]}}
\newcommand{\zd}[1]{{\color{purple}[ZD: #1]}}

\newcommand{\LL}[1]{{\color{brown}[LL: #1]}}

\newcommand{\nc}{\newcommand}
\newcommand{\ben}{\begin{enumerate}}
\newcommand{\een}{\end{enumerate}}
\nc{\nnl}{\nn \\ &}  %new new line
\nc{\fot}{\frac{1}{2}} %frac one two
\nc{\oo}[1]{\frac{1}{#1}} % one over
\nc{\mc}{\mathcal}
\nc{\onenorm}[1]{\L\| #1 \R\|_1} %one norm
\nc{\Ra}{\Rightarrow}
\nc{\zo}{\{0,1\}}

\makeatletter
\newenvironment{breakablealgorithm}
  {% \begin{breakablealgorithm}
   \begin{center}
     \refstepcounter{algorithm}% New algorithm
     \hrule height.8pt depth0pt \kern2pt% \@fs@pre for \@fs@ruled
     \renewcommand{\caption}[2][\relax]{% Make a new \caption
       {\raggedright\textbf{\fname@algorithm~\thealgorithm} ##2\par}%
       \ifx\relax##1\relax % #1 is \relax
         \addcontentsline{loa}{algorithm}{\protect\numberline{\thealgorithm}##2}%
       \else % #1 is not \relax
         \addcontentsline{loa}{algorithm}{\protect\numberline{\thealgorithm}##1}%
       \fi
       \kern2pt\hrule\kern2pt
     }
  }{% \end{breakablealgorithm}
     \kern2pt\hrule\relax% \@fs@post for \@fs@ruled
   \end{center}
  }
\makeatother

\usepackage{nomencl}
\makenomenclature
\usepackage{etoolbox}
\renewcommand\nomgroup[1]{%
  \item[\bfseries
  \ifstrequal{#1}{B}{Hamiltonian notations}{%
  \ifstrequal{#1}{C}{Initial states notations}{%
  \ifstrequal{#1}{D}{Hamiltonian evolution times notations}{%
  \ifstrequal{#1}{E}{Tensor notations}{%
  \ifstrequal{#1}{F}{Algorithm notations}{%
  }}}}}%
]}
\begin{document}

\title{Quantum Filtering and Analysis of Multiplicities in Eigenvalue Spectra}

\begin{abstract}

Fine-grained spectral properties of quantum Hamiltonians, including both eigenvalues and their multiplicities, provide useful information for characterizing many-body quantum systems as well as for understanding phenomena such as topological order. Extracting such information with small additive error is \#\textsf{BQP}-complete in the worst case. In this work, we introduce QFAMES (Quantum Filtering and Analysis of Multiplicities in Eigenvalue Spectra), a quantum algorithm that efficiently identifies clusters of closely spaced dominant eigenvalues and determines their multiplicities under physically motivated assumptions, which allows us to bypass worst-case complexity barriers. QFAMES also enables the estimation of observable expectation values within targeted energy clusters, providing a powerful tool for studying quantum phase transitions and other physical properties. 
We validate the effectiveness of QFAMES through numerical demonstrations, including its applications to characterizing quantum phases in the transverse-field Ising model and estimating the ground-state degeneracy of a topologically ordered phase in the two-dimensional toric code model. {We also generalize QFAMES to the setting of mixed initial states.}
Our approach offers rigorous theoretical guarantees and significant advantages over existing subspace-based quantum spectral analysis methods, particularly in terms of the sample complexity and the ability to resolve degeneracies.

\end{abstract}

\author{Zhiyan Ding}
\affiliation{Department of Mathematics, University of Michigan, Ann Arbor, Michigan 48109, USA}

\author{Lin Lin}
\email{linlin@math.berkeley.edu}
\affiliation{Department of Mathematics, University of California, Berkeley, California 94720, USA}
\affiliation{Applied Mathematics and Computational Research Division, Lawrence Berkeley National Laboratory, California 94720, USA}

\author{Yilun Yang}
\affiliation{Department of Mathematics, University of California, Berkeley, California 94720, USA}

\author{Ruizhe Zhang}
\affiliation{Department of Computer Science, Purdue University, West Lafayette, Indiana 47907, USA}
\date{\today}
\maketitle

\section{Introduction}

Understanding the energy spectrum of a Hamiltonian is a central problem in quantum computing, quantum complexity theory, quantum chemistry, and many-body physics.  While estimating the ground state energy or a few low-lying excitations is a common task, many physical properties require more fine-grained spectral information.
For instance, the ground-state degeneracy (\textsc{gsd}) can indicate the presence of topological order---a hallmark of certain condensed matter phases~\cite{wen2017colloquium} and a key feature for topological quantum computation~\cite{nayak2008non}. More generally, we are interested in estimating both locations and the multiplicities of eigenvalues.

Despite their fundamental physical relevance, extracting such spectral information is computationally challenging: the Hilbert space dimension grows exponentially with system size, making exact diagonalization infeasible even for modestly sized systems.
Indeed, the \textsc{gsd} or eigenstate degeneracy problem is shown to be in the worst case \#\textsf{BQP}-complete---the counting version of \textsf{QMA}~\cite{bfs11}~\footnote{Ref.~\cite{bfs11} showed that exactly computing the density of states (\textsc{dos}) for a local Hamiltonian is \#\textsf{BQP}-complete, under the promise that no eigenvalue lies near the boundary of the target interval $[a,b]$. On the other hand, if one only requires a small multiplicative approximation error but allows an exponentially large additive error, the problem becomes significantly easier. For instance,  Refs. \cite{brandao2008entanglement,gyurik2022towards}  apply QPE to the maximally mixed state to sample the spectrum, yielding an efficient estimator for the \textsc{gsd} or \textsc{dos} when the target eigenspace degeneracy is exponentially large. In contrast, we focus on the regime where the target degeneracy is $\mathcal{O}(1)$.}.
As a result, we can only hope to solve this problem by leveraging additional assumptions for physically motivated Hamiltonians.

\begin{figure*}
    \centering
    \includegraphics[width=0.8\linewidth]{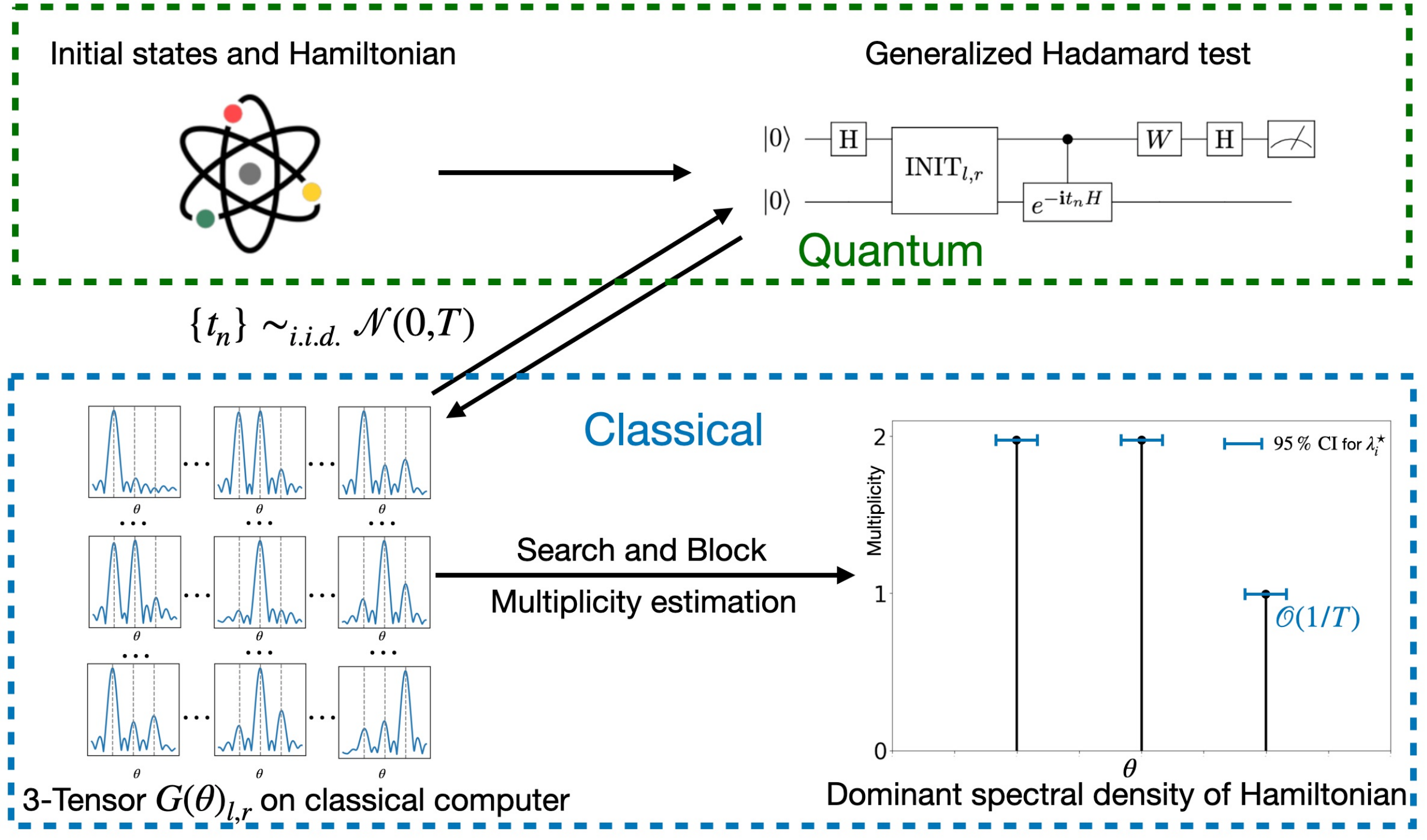}
    \caption{Illustration of the QFAMES algorithm. Given the Hamiltonian and two sets of initial states prepared by circuits $\{U_l\}_{l\in [L]}$ and $\{V_r\}_{r\in [R]}$, we measure quantities of the form $\mc{Z}_{l,r}(t_n) = \braket{0 | U^{\dagger}_l e^{-\i H t_n} V_{r} | 0}$ using the generalized Hadamard test circuit. The collected quantum data are then processed classically to (1) estimate  the locations of energy eigenstate clusters through searching and blocking, and (2) compute the multiplicities of the probed clusters. }
    \label{fig:flow_chart}
\end{figure*}

In this work, we tackle this problem by proposing a new quantum algorithm, dubbed  Quantum Filtering and Analysis of Multiplicities in Eigenvalue Spectra (QFAMES). Our approach prepares a collection of initial states and leverages their correlated information to estimate the \emph{density of dominant eigenstates} (\textsc{dods}) defined as
\begin{align}\label{eqn:def_dods}
    \mu_D(E)\propto \sum_{i\in \mc{D}} \delta(E - \lambda_i)\,,
\end{align}
where $\{\lambda_i\}_{i\in \mc{D}}$ denotes the multiset of dominant energy eigenvalues (counting multiplicity), whose corresponding eigenvectors have sufficiently large overlaps with the subspace spanned by the initial states.

To our knowledge, QFAMES is the first quantum algorithm that can provably estimate the \textsc{dods} with rigorous theoretical guarantees. Specifically, QFAMES identifies clusters of closely spaced dominant eigenvalues and determines exactly the number of (possibly degenerate) eigenstates within each cluster {with a proper choice of initial states}. The algorithm requires only a single ancilla qubit and employs short-depth circuits, which makes it suitable for the early fault-tolerant regime (see~\cref{thm:main}). Furthermore, recent advances~\cite{Yang2024,Clinton2024,Wang2025,Schiffer2025,Yi2025} suggest that even this single ancilla qubit may be eliminated under additional assumptions (see \cref{sec:ancilla_free}).

We then extend the algorithm to efficiently estimate expectation values of observables for the dominant eigenstates within a specified cluster. This task may be viewed as a natural generalization of the \emph{eigenstate property estimation} problem~\cite{zeng2021universal,zwp22,sun2024high,sun2025randomised} which focuses solely on a single, isolated eigenstate.
Our approach enables probing physical properties within targeted clusters, thereby providing a useful tool for characterizing phenomena such as quantum phase transitions. {The proposed framework is applicable to both pure and mixed initial states.}

{The effectiveness of QFAMES relies on the choice of initial states. Specifically, we require that the overlap between the set of initial states and \emph{any} normalized vector in a dominant eigenspace be uniformly bounded from below. This \emph{uniform overlap condition} generalizes the standard overlap assumption in quantum phase estimation. Moreover, we show that an information-theoretic barrier prevents degeneracy resolution when this condition is not met (see \cref{prop:no-go}).}

{That said, there may not be an operationally viable way to directly check the uniform overlap condition in practice, as it requires access to the target eigenspace. In practice, users should choose from a set of physically motivated initial states that are expected to have good coverage of the target eigenspaces. The QFAMES algorithm does not require the initial states to be linearly independent. This allows users to choose a larger, potentially redundant set of initial states (pure or mixed) to satisfy the uniform overlap condition, while avoiding numerical instabilities associated with near-linear dependence. We outline three scenarios where such initial states may be prepared in practice.}

{First, for molecular quantum chemistry problems, there already exist various practical proposals that prepare low-energy trial state sets, including sums of Slater determinants~\cite{Babbush2015, tubman2018, Fomichev2024}, spin flip~\cite{casanova2020spin,maskara2025programmable}, coupled cluster~\cite{Lee2019,anand2022quantum}, and more general variational ansatze~\cite{Cerezo2021,  Grimsley2019, Baek2023, zhang2025}. These protocols can be applied in conjunction with subspace-based quantum eigenvalue estimation methods~\cite{Parrish2019QuantumFD,Huggins2020,Stair2020,Seki2021,Epperly2022,Baek2023,Kirby2024,Yoshioka2025} to determine ground and excited state energies. Starting from these physically motivated initial states, QFAMES (a) provides a theoretical foundation and necessary conditions for subspace-based methods, and (b) improves performance compared to previous subspace diagonalization approaches. In particular, it is not plagued by the problem of inverting ill-conditioned overlap matrices. The analysis presented here also sheds light on the conditions required for initial states in existing subspace diagonalization methods.}

{Second, tensor network methods, such as matrix product states (or projected entangled-pair states in two or higher spatial dimensions) with low bond dimensions, can efficiently prepare initial states that are useful in quantum chemistry and condensed matter systems~\cite{Wei2023,Wei2024,Smith2024,scheer2025,jaderberg2025,mansuroglu2025,Berry2025Rapid,Fomichev2024,Leimkuhler2025,Wu2025}. While the bond dimension required to represent the exact ground state of a quantum system may be large, QFAMES only requires that the set of initial states collectively has non-negligible overlap with the ground-state manifold. This relaxes the bond dimension requirements. In this paper, we study the transverse field Ising model, the toric code model, and the XXZ model, which demonstrate the efficacy of these matrix product initial states when combined with QFAMES.}

{Third, generalizing QFAMES to mixed initial states (discussed in \cref{sec:QFAMES_mixed}) opens up new possibilities. For instance, recently developed dissipative algorithms for preparing low-temperature thermal or ground states~\cite{VerstraeteWolfIgnacioCirac2009,RoyChalkerGornyiEtAl2020,Cubitt2023,RallWangWocjan2023,ChenKastoryanoBrandaoEtAl2025,DingLiLin2025,DingChenLin2024,LiZhanLin2025,Hahn2025Efficient,GilyenChenDoriguelloKastoryano2024,JiangIrani2024,Lloyd2025Quasiparticle,ding2025endtoendefficientquantumthermal,li2025dissipative,ZhanDingHuhnEtAl2026} typically yield mixed states, especially when the ground state manifold is degenerate. Also, the controlled state-preparation oracle is inaccessible due to the irreversibility of the dissipative process. Our construction addresses this by utilizing an additional system register, thereby removing the need for controlled state-preparation circuits. We demonstrate, via a proof-of-principle example, that solving the generalized eigenvalue problem with such mixed states extracts significantly more information about observables within the ground state manifold than individual measurements alone.}

The paper is structured as follows. In~\cref{sec:related_work}, we review the related works. In \cref{sec:setup}, we introduce the problem setup and summarize the main results. \cref{sec:algo} presents the QFAMES algorithm in detail, while \cref{sec:obs} discusses how to extend the method to extract information of physical observables. ~\cref{sec:ancilla_free} discusses how to implement the QFAMES algorithm in a control-free, ancilla-free setting. {The extension of QFAMES to mixed state case is discussed in~\cref{sec:QFAMES_mixed}.} In \cref{sec:numerical_tests}, we present numerical simulation results with an illustrative example, the transverse field Ising model, the toric code model{, and the XXZ model}. Finally, \cref{sec:discussion} concludes with a discussion of possible future directions.

\section{Related work}\label{sec:related_work}
The estimation of Hamiltonian eigenvalues is a fundamental problem in the literature of quantum computation, which can be addressed by a variety of methods collectively known as quantum phase estimation (QPE) algorithms~\cite{kitaev1995quantum,KitaevShenVyalyi2002,Dobsicek2007,Wiebe2016,Somma2019,OBrien2019,DingLin2023,Ding2023simultaneous,ding2024quantum,ding2023robust,yi2024quantum,castaldo2025heisenberg,ni2023lowdepth,wfz22,ni2023lowdepth_2,yi2024quantum,castaldo2025heisenberg}. These methods primarily aim at determining the locations of eigenvalues but are inherently incapable of identifying their multiplicities, as they typically assume access to only a single initial state. A more detailed review of these algorithms is provided in~\cref{sec:single_state_limitation}.

In the literature, several recent works~\cite{scali2025spectral,Guo_2025} have also explored problems related to estimating~\cref{eqn:def_dods}, albeit under different oracle models and assumptions. Some of these studies establish a connection between DOS estimation and quantum topological data analysis (TDA)~\cite{Lloyd2016QuantumAF,Berry2024,Scali2024}, where QPE can be employed to sample energies and estimate degeneracies within small energy intervals~\cite{scali2025spectral}. In contrast,~\cite{Guo_2025} reformulates the estimation of ground state degeneracy as a ground state preparation task for a specially constructed auxiliary Hamiltonian, which in turn requires efficient ground state preparation of the auxiliary Hamiltonian.

As will be shown later, our observable eigenstate property estimation problem reduces to a generalized eigenvalue problem, which shares structural similarities with subspace-based quantum eigenvalue estimation methods~\cite{Parrish2019QuantumFD,Huggins2020,Stair2020,Seki2021,Epperly2022,Baek2023,Kirby2024,Yoshioka2025,Klymko2022,Oliveira2025}.
These subspace methods can be generalized to estimate eigenvalue multiplicities when multiple initial states are used.~{
QFAMES, however, distinguishes itself with the following features:
\begin{itemize}
    \item \emph{Heisenberg-limited scaling.} Krylov subspace-based methods often require uniform time sampling, which prevents the attainment of Heisenberg-limited scaling in energy estimation, as established in~\cite[Theorem 1.7]{ding2024esprit}. QFAMES circumvents this limitation by utilizing Gaussian sampling of the evolution time.
    \item \emph{Separation of dominant eigenvalues.} Subspace-based methods usually require high-precision measurements of data matrices due to coherence between different dominant eigenvalues. Although noise can be suppressed via singular value truncation~\cite{Klymko2022,Oliveira2025} as in QFAMES, the resulting complexity remains suboptimal according to state-of-the-art analyses of subspace methods~\cite{Epperly2022,Baek2023,Kirby2024}.
    QFAMES addresses this by applying a Gaussian energy filter to distinguish distinct dominant eigenvalues and separate them into smaller, well-conditioned SVD sub-problems for multiplicity determination.
    \item \emph{Small data matrix.} The dimension of the data matrix in subspace-based methods often scales with the number of uniformly chosen time points, which is usually proportional to the inverse target precision in energy. In contrast, thanks to the filtering technique, the dimension of the QFAMES data matrix is only determined by the number of initial states, which is independent of the accuracy requirement.
\end{itemize}
These features allow us to derive stronger theoretical guarantees than those available for standard subspace-based algorithms.
To our knowledge, this is also the first work capable of provably estimating eigenvalue multiplicities.
}

\section{Setup and main results}\label{sec:setup}
Let us first formally state the \textsc{dods} estimation problem:
\begin{problem}[Density of dominant eigenstates (\textsc{dods}) estimation]\label{prob:dods_estimation}
Let $H$ be a Hamiltonian with eigen-decomposition $\{(\lambda_i, \ket{E_i})\}$, where, after appropriate normalization, the spectrum satisfies $\lambda_i \in [-\pi, \pi]$ for all $i$.
Assume black-box access to controlled time evolution $e^{-\mathrm{i}Ht}$ for any $t\in\mathbb{R}$, as well as two families of (controlled) state-preparation unitaries
$\{U_\ell\}_{\ell\in[L]}$ and $\{V_r\}_{r\in[R]}$ satisfying $U_\ell|0\rangle=|\phi_\ell\rangle$ and $V_r|0\rangle=|\psi_r\rangle$.

We define two overlap matrices $\Phi \in \mathbb{C}^{L \times M}$ and $\Psi \in \mathbb{C}^{R \times M}$:
\begin{equation}
    \begin{aligned}
        \Phi_{l, m} &:= \langle \phi_l | E_m \rangle \quad \forall\, l \in [L]\,, m \in [M]\,, \\
        \Psi_{r, m} &:= \langle \psi_r | E_m \rangle \quad \forall\, r \in [R]\,, m \in [M]\,.
    \end{aligned}
    \label{eq:def_Phi_Psi}
\end{equation}
For each eigenstate $\ket{E_m}$, define its overlap with the initial states as:
\begin{align}\label{eqn:def_pm}
    p_m:=\|\Phi_{:,m}\|_2\cdot \|\Psi_{:, m}\|_2\,.
\end{align}
The multiset of \emph{dominant eigenvalues}, indexed by $\mathcal{D}$, consists of those eigenvalues that satisfy the {\emph{dominance condition}}:
\begin{align}
        {\left(p_{\min} := \min_{m\in\mc{D}} p_m \right)\gtrsim\left( p_{\mathrm{tail}} := \sum_{m \notin \mc{D}} p_{m} \right)},
        \label{eq:def_domination}
\end{align}
The distinct dominant eigenvalues are denoted by $\{\lambda_i^{\star}\}_{i \in [I]}$, each separated by at least $\Delta$, with corresponding eigenvectors indexed by $\mc{D}_i$. The degeneracies $|\mc{D}_i|$ satisfy $|\mc{D}| = \sum_{i \in [I]} |\mc{D}_i|$.

\textbf{Goal:} With high probability, for each $i\in [I]$, estimate the dominant eigenvalue $\lambda_i^\star$ within $\epsilon$-accuracy and exactly determine its dominant degeneracy $m_i$. These outputs together suffice to approximate the density of dominant eigenstates (\textsc{dods}) $\mu_D$, as defined in~\cref{eqn:def_dods}.
\end{problem}

The left and right initial states can be either the same or different, depending on the problem setup. When they are different, the left initial states can be viewed as bases for projective measurements. {In this work, to ensure that our algorithm can theoretically solve the \textsc{dods} estimation problem efficiently, we require the dominant eigenvalues to satisfy a \emph{sufficient dominance condition}, namely, that $p_{\min}/p_{\rm tail}$ be greater than a constant determined by the structure of $\Phi$ and $\Psi$. The precise condition, together with a detailed discussion of the requirement on $p_{\min}/p_{\rm tail}$, is deferred to~\cref{sec:rigorous_version},~\cref{thm:non_orthogonal_location}, and~\cref{rem:thm_condition}.}

If $\lambda_i$ is a non-degenerate dominant eigenvalue, and there is only one pair of initial states $\ket{\phi}$ and $\ket{\psi}$, then both $|\langle \phi | E_m \rangle|$ and $|\langle \psi | E_m \rangle|$ need to be large enough. Indeed, when $\ket{\phi}=\ket{\psi}$, $p_m=|\langle \phi | E_m \rangle|^2$ is simply the squared overlap between the initial state and the eigenstate, which is the standard assumption in quantum phase estimation and its variants (see \cref{sec:single_state_limitation}).  More generally, if $\lambda_i$ is a non-degenerate dominant eigenvalue but there are multiple pairs of initial states, then both $\sum_{l\in [L]} |\langle \phi_l | E_m \rangle|^2=\|\Phi_{:,m}\|_2^2$ and $\sum_{r\in [R]} |\langle \psi_r | E_m \rangle|^2=\|\Psi_{:,m}\|_2^2$ need to be sufficiently large, which states that the eigenstate $\ket{E_m}$ has non-vanishing overlap with at least one initial state in each set.
This can be concisely captured by the condition for $p_m$ in \cref{eqn:def_pm} and \cref{eq:def_domination}.

In order to compute the multiplicity of a degenerate dominant eigenvalue $\lambda_i^{\star}$, we need a stronger condition to ensure that \emph{all} normalized eigenvectors corresponding to the degenerate eigenvalue $\lambda_i^{\star}$ in the eigenspace $\mathcal{E}_i:=\operatorname{span} \{\ket{E_m}\}_{m\in\mc{D}_i}$ have non-vanishing overlap with the set of initial states. In other words, both of the following quantities
\begin{equation}\label{eqn:sufficient_coverage_informal}
    \begin{aligned}
        &\min_{\ket{E}\in \mathcal{E}_i, \braket{E|E} = 1 }  \sum_{l\in [L]} \left| \braket{\phi_l | E} \right|^2\\
        &\min_{\ket{E}\in \mathcal{E}_i, \braket{E|E} = 1 }  \sum_{r\in [R]} \left| \braket{\psi_r | E} \right|^2
    \end{aligned}
\end{equation}
need to be sufficiently large.  This will be formulated more precisely as a \emph{uniform overlap condition} in \cref{asp:linear_dep} of \cref{sec:assump}. Indeed, if one of the quantities in \cref{eqn:sufficient_coverage_informal} vanishes, we prove that there is an \emph{information-theoretic barrier} for resolving the degeneracy (\cref{prop:no-go}). As a special case, this proves that existing quantum phase estimation algorithms, which typically use only a single initial state, cannot resolve eigenvalue degeneracies.

The uniform overlap condition states that the overlap of all normalized vectors in the eigen-subspace with the initial states should be uniformly bounded from below. However, this does not require the initial states to be well-conditioned. In fact, we do not even require the vectors $\{\ket{\phi_l}\}$ (or $\{\ket{\psi_r}\}$) to be \emph{linearly independent}.  This provides significantly enhanced flexibility in choosing the initial states compared to existing subspace methods, which typically require the initial states to be a well-conditioned set of vectors. That said, the initial states should not be chosen randomly.
In practical calculations, the initial states should be chosen carefully and in a physically motivated manner, such as low-energy states prepared by variational quantum algorithms~\cite{peruzzo2014variational,farhi2014quantum,Kandala2017}, quantum imaginary time evolution~\cite{Motta2020,Lin2021,Yuan2019}, and the non-orthogonal ansatz states used in the Non-Orthogonal Quantum Eigensolver (NOQE) \cite{Huggins2020,Baek2023} approach.

The main result of this paper is an efficient quantum algorithm to solve the \textsc{dods} estimation problem with \textbf{rigorous theoretical guarantees}.

\begin{theorem}[Main result, Informal version of \cref{thm:main}]\label{thm:informal}
There exists a quantum algorithm that repeatedly runs the quantum circuit in~\cref{fig:flow_chart} to collect data with Hamiltonian evolution times $t_1,\dots,t_N\in \R$ such that
\begin{itemize}
    \item maximal evolution time:
    $T_{\max}:=\max \{|t_n|\}=\wt{\mc{O}}\left(
    p_{\rm tail} \epsilon^{-1}/{p_{\min}} \right)$;
    \item total evolution time:
    $T_{\rm total}:=2\sum^N_{n=1}\sum_{l\in[L]}\sum_{r\in[R]} |t_n|=\wt{\mc{O}}( (p_{\rm tail}{p_{\min}})^{-1}\epsilon^{-1})$.
\end{itemize}
\noindent By classically post-processing the data, it guarantees that for each distinct dominant eigenvalue,
\begin{itemize}
    \item its location can be estimated to within $\epsilon$ accuracy; and
    \item its multiplicity can be determined exactly.
\end{itemize}
\end{theorem}
{According to~\eqref{eq:def_domination} and~\cref{thm:informal}, we have $p_{\rm tail}/p_{\min}\lesssim 1$, so it suffices to take $T_{\max}=\mathcal{O}(1/\epsilon)$. Moreover, the formal result in~\cref{thm:main} allows greater flexibility in the choice of $T_{\max}$ and $T_{\rm total}$ than~\cref{thm:informal}. In particular, it suffices to choose $T_{\max}=\wt{\mc{O}}\left(
\delta \epsilon^{-1}\right)$ and $T_{\rm total}=\wt{\mc{O}}( (\delta p^2_{\min})^{-1}\epsilon^{-1})$ for any $\delta>p_{\rm tail}/p_{\min}$. When $p_{\rm tail}\ll p_{\min}$, the constant prefactor of $1/\epsilon$ in $T_{\max}$ can be much smaller than $1$, which corresponds to the short-depth regime observed in the literature~\cite{DingLin2023,Ding2023simultaneous,ding2024quantum,ni2023lowdepth,ni2023lowdepth_2}.
}

Note that our algorithm can also resolve dominant eigenvalues that are \textbf{approximately degenerate}, i.e., when there is a cluster of dominant eigenvalues lying in a narrow energy interval around $\lambda^\star_i$ with width $\delta \ll \Delta$. See more discussion in~\cref{sec:complex}.

As an extension of the algorithm, we also propose an efficient method (\cref{alg:QFAMES_obs}) for estimating degenerate eigenstate properties in~\cref{sec:obs}. The problem is defined as follows.
\begin{problem}[Degenerate eigenstate property estimation]
Let $O$ be an observable. For each dominant eigenvalue $\lambda_i^{\star}$, estimate the eigenvalues of the projected observable $O_{\mc{D}_i}$ in the corresponding subspace:
\begin{align}
    O_{\mc{D}_i} := (\braket{E_{k}|O|E_{k'}})_{k,k'\in \mc{D}_i}, \quad \mc{D}_i = \left\{ k\in \mc{D}: \lambda_k = \lambda_i^{\star} \right\}\,.
\end{align}
\end{problem}

Let us denote the eigenvalues by $\lambda^{O}_k$. We remark that
any value within the range $[\min_k \lambda^{O}_k , \max_k \lambda^{O}_k]$ can be realized by some linear combination of the degenerate eigenstates. It is thus less useful to estimate $\langle E_m|O|E_m\rangle$ for each individual dominant eigenstate.
For example, in the ferromagnetic phase of a spin chain model, there can be two degenerate ground states  $\ket{E_0} = \ket{\uparrow\uparrow\cdots \uparrow}$ and $\ket{E_1} = \ket{\downarrow\downarrow\cdots \downarrow}$. The measured ground states can be however $(\ket{E_0} +\ket{E_1}) / \sqrt{2}$ and $(\ket{E_0} - \ket{E_1}) / \sqrt{2}$, both of which yield zero average magnetization.

For simplicity, we provide the resource estimation in the case of no tail eigenvalues (i.e., $p_{\rm tail}=0$)
and $p_{\min}$ being bounded below by a constant {(i.e., $p_{\min}=\Omega(1)$)}. To achieve $\epsilon_O$-accuracy in estimating each eigenvalue of $O_{\mathcal{D}_i}$, it suffices to choose
\begin{align}
    T_{\max}=\wt{\mc{O}}(\Delta^{-1}),
    \qquad
    T_{\rm total}= \wt{\mc{O}}(\Delta^{-1}\epsilon^{-2}_O)\,.
\end{align}
{Here, for simplicity, we suppress the dependence on the constant $p_{\min}$, since $p_{\min}=\Omega(1)$.}
The proposition showing this result is~\cref{prop:observable} in~\cref{sec:obs}.

\section{QFAMES algorithm and complexity analysis}~\label{sec:algo}

In this section, we describe the QFAMES algorithm in \cref{sec:QFAMES}, present its pseudo-code, and provide an illustrative example. A detailed complexity analysis follows in \cref{sec:complex}.

\subsection{Algorithmic description}\label{sec:QFAMES}
Our goal is to estimate the \textsc{dods},
including the dominant eigenvalues $\lambda_i^{\star}$ and corresponding multiplicities $m_i$.
If we assume the eigenvalues are separated by at least $\Delta$,
then we can define energy filters of the form $\exp\left[-(\theta-H)^2 T^2 \right]$ to filter out the energy eigenstates away from the desired dominant eigenvalue, where the filter width is chosen to be
$ T \gtrsim 1 / \Delta $. The filter can be implemented with time evolution operators through sampling from its Fourier transform
\begin{align}\label{eq:filter_Gauss}
    e^{-(\theta-H)^2 T^2} = \int e^{\i (\theta-H) t} a_T(t) \d t,
\end{align}
where
\begin{equation}\label{eqn:a_T_untrunc}
\begin{aligned}
 a_T(t)= \frac{1}{2T\sqrt{\pi}}\exp\left(-\frac{t^2}{4T^2}\right)
\end{aligned}
\end{equation}
is the probability density function for evolution time $t$ (more rigorously, we will sample $t$ from a truncated Gaussian distribution as defined in \cref{eqn:a_T} to avoid extremely large $t_n$).

For the given sets of left and right initial states $\ket{\phi_l}$ and $\ket{\psi_r}$, we want to compute the matrix $\mathcal{G}(\theta)$ with matrix elements
\begin{align}
    \mathcal{G}_{l,r}(\theta) =\braket{\phi_l | e^{-(\theta-H)^2 T^2} |\psi_r} \,.
    \label{eq:matrx_g}
\end{align}
At $\theta=\lambda_i^{\star}$, we have
\begin{equation}
    \begin{aligned}
    \mathcal{G}(\lambda_i^{\star}) & =\Phi \cdot \diag\left(\left\{e^{-(\lambda_i^{\star}-\lambda_m)^2 T^2} \right\}_{m\in [M]}\right)\cdot \Psi^\dagger \\
    & \approx \Phi_{:,\mc{D}_i}\cdot \left(\Psi_{:,\mc{D}_i}\right)^\dagger\,,
\end{aligned}
\end{equation}
where the matrices $\Phi$ and $\Psi$ are defined in \cref{eq:def_Phi_Psi}, and $\Phi_{:,\mc{D}_i}$ and $\Psi_{:,\mc{D}_i}$ are the corresponding submatrices that contain the columns of
dominant eigenstates with eigenvalue $\lambda_i^{\star}$.
$\mc{D}_i$ denotes the index set of these eigenvalues. Due to the concentration of $\exp(-x^2 T^2)$ around zero, this expression implies that only those eigenvalues that are close to $\lambda^{\star}_i$ contribute significantly to the matrix $\mathcal{G}(\lambda^{\star}_i)$, and the remaining eigenvalues are filtered out.
Under the uniform overlap assumption (\cref{asp:linear_dep}), we can show that the rank of the matrix $\mathcal{G}(\lambda_i^\star)$ equals the degeneracy of dominant eigenstates at $\lambda_i^{\star}$.
This is the key insight that enables QFAMES to solve the \textsc{dods} estimation problem. Specifically:
\begin{itemize}
    \item \textit{Location estimation:} We evaluate the \textbf{Frobenius norm} of $\mathcal{G}(\theta)$  on a uniform grid of $\theta\in [-\pi, \pi]$, and use a ``searching and blocking'' strategy to find the peaks $\theta_i^{\star}$ of $\|\mathcal{G}(\theta)\|_F$. Each peak corresponds to a distinct dominant eigenvalue $\lambda_i^{\star}$.

    \item \textit{Multiplicity estimation:} For each candidate dominant eigenvalue $\theta_i^\star$, we compute the \textbf{numerical rank} of $\mathcal{G}(\theta_i^\star)$ (the number of singular values above some threshold) as the estimated multiplicity.
\end{itemize}
The detailed description of the algorithm is as follows. The first stage of the algorithm, consisting of Steps 1{–}3, focuses on estimating the locations of the dominant eigenvalues. The second stage, Step 4, focuses on determining the multiplicity of each dominant eigenvalue.

\paragraph*{Step 1: Generate data.}
To obtain \cref{eq:matrx_g}, we need to measure quantities of the form
\begin{align}
    \label{eq:zlrt}
    \mc{Z}_{l,r}(t_n):=\bra{0}U_l^\dagger e^{-\i H t_n} V_r \ket{0} = \langle \phi_l|e^{-\i H t_n} |\psi_r\rangle\,.
\end{align}
{Their real and imaginary parts} can be obtained with the generalized Hadamard test as drawn below:
\begin{align}
    \label{eq:generalized_hadamard}
    \Qcircuit @C=1em @R=1em {
    \lstick{|0\rangle} & \gate{\rm H} & \multigate{1}{{\rm INIT}_{l,r}} & \ctrl{1} &  \gate{W} & \gate{\rm H} & \meter \\
    \lstick{\ket{0}}   & \qw          & \ghost{{\rm INIT}_{l,r}}         & \gate{e^{-\mathbf{i}t_n H}} & \qw & \qw & \qw
}
\end{align}
\begin{align}
    \Qcircuit @C=1em @R=1em {
    & \multigate{1}{{\rm INIT}_{l,r}} & \qw \\
    & \ghost{{\rm INIT}_{l,r}}        & \qw
}
\quad \raisebox{-1.0em}{=} \quad
\Qcircuit @C=1em @R=1em {
    & \ctrlo{1} & \ctrl{1} & \qw \\
    & \gate{U_l} & \gate{V_r} & \qw
}
\end{align}
{where $W$ is either the identity or the phase gate $S^\dagger = \begin{bmatrix}
    1 & 0 \\ 0 & -\i
\end{bmatrix}$.}
Given $U_l,V_r$, and $t_n$, the measurement outcomes $Z_{l,r,n}$ of the circuit \cref{eq:generalized_hadamard}
provide an unbiased estimator for $\mathbb{E}[Z_{l,r,n}]=\mc{Z}_{l,r}(t_n)$.
Using the generalized Hadamard test, we can efficiently gather these measurement outcomes for every pair of $(U_l, V_r)$ at the Hamiltonian evolution times $t_1,\dots, t_N$ and formulate them into a third-order tensor $Z\in \C^{L\times R\times N}$:
\begin{align}
    \mathbb{E}[Z_{l,r,n}] = \sum_{m\in [M]} \Phi_{l,m}\cdot e^{-\i \lambda_m t_n} \cdot (\Psi_{r,m})^*\,.
\end{align}
where $\Phi$ and $\Psi$ are the overlap matrices defined in Eq.~\eqref{eq:def_Phi_Psi}.

The times $t_1,\dots,t_N$ are sampled independently from the Gaussian probabilistic density function $a_T(t)$. In practice, we may truncate the density function to ensure that the maximal runtime is always bounded by $\sigma T$:
\begin{equation}\label{eqn:a_T}
\begin{aligned}
    & a_T^{\rm trunc}(t)\\
    =&
    \left(1-\int^{\sigma T}_{-\sigma T}\frac{1}{2T\sqrt{\pi}}\exp\left(-\frac{s^2}{4T^2}\right)\mathrm{d} s\right)\delta_{0}(t)\\
    &+
    \frac{1}{2T\sqrt{\pi}}\exp\left(-\frac{t^2}{4T^2}\right)\textbf{1}_{[-\sigma T,\sigma T]}(t)\,,
\end{aligned}
\end{equation}
where $\sigma$ is an adjustable parameter and $\delta_0(t)$ is the delta function at point $0$. Its Fourier transform approximates a Gaussian function and is concentrated around zero.

\paragraph*{Step 2: Compute the filtered {matrix estimator}.}
We can now \textbf{contract} the time dimension of the data tensor with the following complex exponential vector:
\begin{align}
    \frac{1}{N}\begin{bmatrix}
        e^{\i \theta t_1} & e^{\i \theta t_2} & \cdots & e^{\i \theta t_N}
    \end{bmatrix}\,,
\end{align}
for a location parameter $\theta \in \R$ and chosen time samples $t_1,\dots,t_N$. This yields an $L$-by-$R$ matrix $G(\theta)$, where
\begin{align}\label{eqn:G}
    G(\theta)_{l,r} = \frac{1}{N}\sum_{n=1}^N Z_{l,r,n} e^{\i \theta t_n}, \quad l\in [L], r\in [R]\,.
\end{align}
Since $t_1,\dots,t_N$ are sampled independently from the truncated Gaussian probabilistic density function $a_T^{\rm trunc}(t)$,
it holds that
\begin{equation}
    \begin{aligned}
        &\E[G(\theta)] = \mathcal{G}(\theta) \\
        = &\Phi \cdot \diag\left(\left\{\int e^{\i (\theta-\lambda_m) t} a_T^{\rm trunc}(t) \d t\right\}_{m\in [M]}\right)\cdot \Psi^\dagger\\
        \approx &\Phi \cdot \diag\left(\left\{\exp(-(\theta-\lambda_m)^2 T^2)\right\}_{m\in [M]}\right)\cdot \Psi^\dagger\,.
    \end{aligned}
\end{equation}
{Thus, $G(\theta)$ serves as an approximation of $\mathcal{G}(\theta)$.
}

{In our algorithm, we assume the dominant eigenvalues belong to $[-\pi,\pi]$ for simplicity of presentation. It can be achieved by rescaling and shifting the Hamiltonian. We choose the grid points $\theta$ uniformly from the interval $[-\pi,\pi]$ (See~\cref{alg:QFAMES} Lines 14,15).
}

\paragraph*{Step 3: Search and block.} After computing $G(\theta)$, the dominant eigenvalues are identified as the peaks of $\|G(\theta)\|_F$ and can be located using the search-and-block procedure proposed in~\cite{ding2024quantum}. This strategy consists of two iterative components—the search step and the block step—which are described as follows:
\begin{itemize}
\item \emph{Search:} In the first step, we find the maximum point $\theta^\star_1$ of $\|G(\theta)\|_F$, which approximates to one of the dominant eigenvalues.

\item \emph{Block:} To approximate the next dominant eigenvalue and avoid repetition, we define a block interval
\begin{align}
    \mathcal{I}_{B,1}=[\theta^\star_1-\alpha / T,\theta^\star_1+\alpha / T]
\end{align}
for some given constant $\alpha$. The conditions for $\alpha$ are stated in \cref{thm:main}.
In the following steps, we will not find any points from this interval.

\item \emph{Second iteration:} In the second iteration, we find the {second-largest} point $\theta^\star_2$ of ${\left\|G(\theta)\right\|_F}$ outside the block interval $\mathcal{I}_{B,1}=[\theta^\star_1-\alpha / T,\theta^\star_1+\alpha/T]$ around $\theta^\star_1$:
\begin{align}
    \theta^\star_2 =\argmax_{\theta\in\mathcal{I}^c_{B,1}}\|G(\theta)\|_F\,.
\end{align}

After obtaining $\theta^\star_2$, we enlarge the block interval by setting $\mathcal{I}_{B,2} = [\theta^\star_2 - \alpha / T, \theta^\star_2 + \alpha / T] \cup \mathcal{I}_{B,1}$, and then identify the third maximum of $\|G(\theta)\|_F$ outside $\mathcal{I}_{B,2}$.
This searching and updating process is iteratively repeated until a set of $\tilde{I}$ ``maximal'' points is discovered. Ultimately, we obtain an approximate candidate set $\left\{\theta^\star_{i}\right\}^{\tilde{I}}_{i=1}$ corresponding to the set of dominant eigenvalues $\left\{\lambda^\star_{i}\right\}_{i\in[I]}$.

\end{itemize}

\paragraph*{Step 4: Estimate multiplicity.} For each $\{\theta^\star_{i}\}_{i=1}^{\tilde{I}}$, we perform the singular value decomposition
\begin{align}
    {G}(\theta^\star_{i}) = U_i \Sigma_i V_i^\dagger,
\end{align}
and define
\begin{align}
    m_i = \#\left\{ j \;\middle|\; (\Sigma_i)_{j,j} > \tau \right\}
\end{align}
for some given threshold $\tau$, which gives the multiplicity associated with each candidate $\theta^\star_{i}$. If $m_i = 0$, then $\theta^\star_{i}$ corresponds to a spurious dominant eigenvalue and should be discarded.

\paragraph*{Pseudocode.}
\cref{fig:algo_tensor} summarizes the QFAMES workflow as a tensor network. The left/right blocks $U_l$ and $V_r$ prepare two families of initial states. Time evolution $e^{-\i Ht}$ induces cross-correlations between these families, which (via the generalized Hadamard test) yield the sampled signals $\mc{Z}_{l,r}(t_n)$. A time-domain filter $a_T(t)$ and a discrete Fourier transform produce frequency-domain matrices ${G}(\theta)$, which are then processed classically: a search-and-blocking step estimates the distinct dominant eigenvalues $\{\lambda_i^{\star}\}$,
and rank/conditioning tests on submatrices determine multiplicities $\{m_i\}$. An optional branch forms projected observable matrices and solves a generalized eigenvalue problem to obtain the spectrum of $O_{\mc{D}_i}$ for each $\lambda_i^{\star}$. The detailed algorithm for eigenvalue and multiplicity estimation is presented in \cref{alg:QFAMES}, while the observable estimation will be discussed in detail in \cref{sec:obs}.

\begin{figure*}
    \centering
    \includegraphics[width=0.9\linewidth]{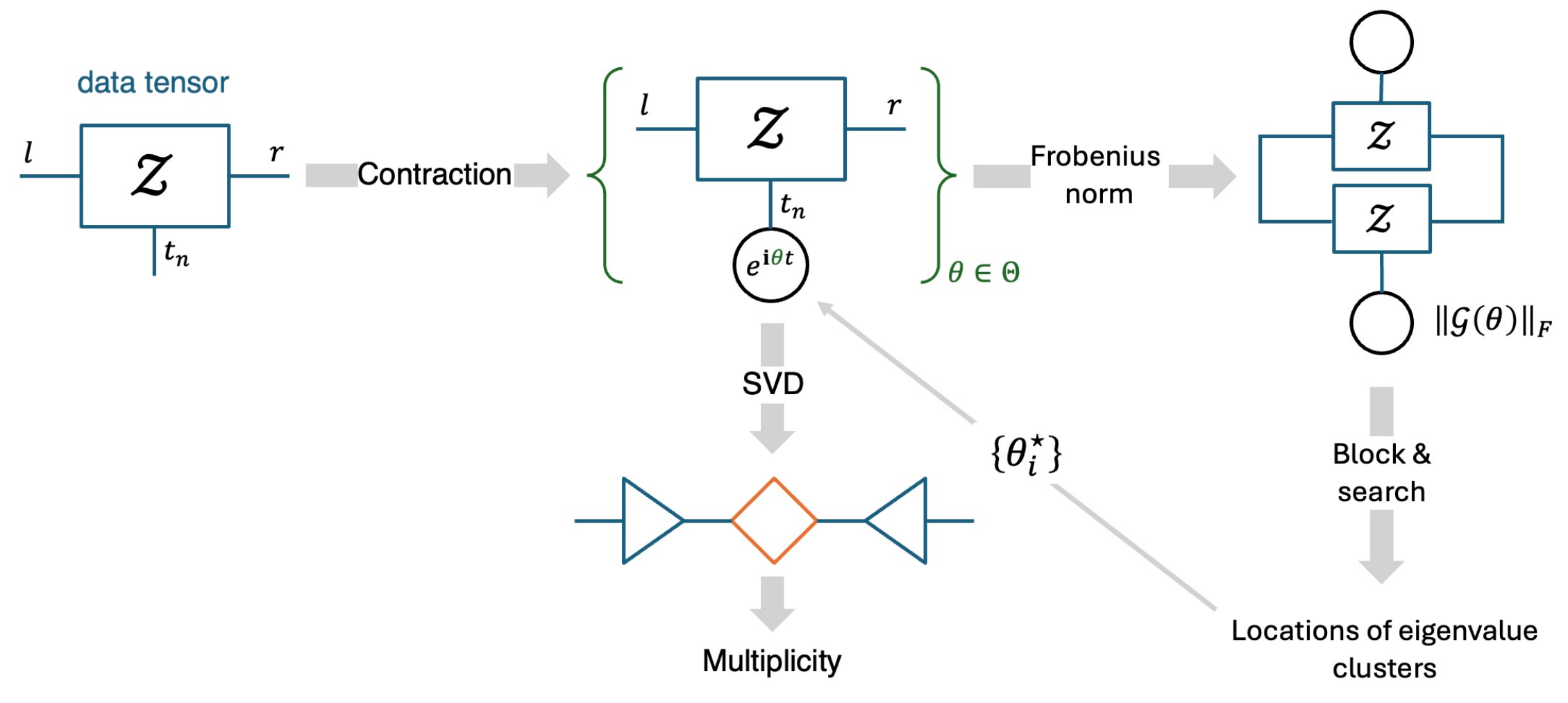}
    \caption{{Diagram of the QFAMES algorithm for post-processing the 3-tensor generated from quantum data of the generalized Hadamard test circuit. Left/right tensors $U_l$ and $V_r$ prepare initial-state families; the middle leg represents time evolution $e^{-\mathrm{i}Ht}$ and data acquisition of cross-correlators $\mc{Z}_{l,r}(t_n)$. A filter $a_T(t)$ and discrete Fourier transform yield $\mathcal{G}(\theta)$, which is post-processed by search-and-blocking to locate
    {dominant eigenvalues $\{\lambda_i^{\star}\}$}
    and by rank tests to determine multiplicities $\{m_i\}$. }}
    \label{fig:algo_tensor}
\end{figure*}

\begin{breakablealgorithm}
      \caption{Quantum Filtering and Analysis of Multiplicities in Eigenvalue Spectra (QFAMES)}
  \label{alg:QFAMES}
  \begin{algorithmic}[1]
  \State \textbf{Preparation:}

  Number of data pairs: $N$;

  Two initial state sets: $\{\phi_i\}_{i=1}^L$ and $\{\psi_j\}_{j=1}^R$;

  Filter parameter: $T$;

  Truncation parameter in $a_T^{\rm trunc}$: $\sigma$;

  Number of {dominant eigenvalues (guess)}: $\tilde{I}$;

  Singular value decomposition threshold: $\tau$;

  \multiline{
  Searching parameter: $q$ (we search with discrete energy step $q / T$);
  }

  \multiline{
  Block parameter: $\alpha$ (distinct dominant  eigenvalues are at least separated by $\alpha / T$);
  }

  \State \textbf{Running:}
  \State \textbf{Stage I: Estimate the location of dominant eigenvalues}
  \State \Comment{\textcolor{blue}{Step 1: Generate data}}
  \State Sample $\{t_n\}$ i.i.d. from the truncated Gaussian distribution $a_T^{\rm trunc}(t)$ defined in \cref{eqn:a_T}
  \For{$l=1$ to $L$}
    \For{$r$=1 to $R$}
        \For{$n=1$ to $N$}
            \State
            \multiline{
             Generate data $Z_{l,r,n}\in \{\pm 1\pm \i\}$ using the generalized Hadamard test with $(U_l, V_r, t_n)$
            }
        \EndFor
    \EndFor
  \EndFor
  \State \Comment{\textcolor{blue}{Step 2: Compute the filtered density function}}
  \State $J\gets \left\lfloor \frac{2\pi T}{q}\right\rfloor$.
    \State Generate discrete candidates: {$\theta_j\gets-\pi+\frac{jq}{T}$ for $j=0,1,\dots,J$.}
  \State Calculate
  \[
    \left\|G(\theta_j) \right\|_F\gets \left\|\frac{1}{N}\sum^N_{n=1}{Z_{:,:,n}}\exp(\i\theta_j t_n)\right\|_F\,,\quad 0\leq j\leq J.
  \]
  \State \Comment{\textcolor{blue}{Step 3: Search-and-block}}
  \State Block set: $\mathcal{B}_{1}\gets \emptyset$.
  \For{$i=1$ to $\tilde{I}$}
  \State $j_i=\mathrm{argmax}_{\theta_j\notin \mathcal{B}_{i}} \left\|G(\theta_j) \right\|_F$.
  \State $\theta^\star_i\gets\theta_{j_i}$.
  \State $\mathcal{B}_{i+1}\gets \mathcal{B}_{i}\cup \left(\theta_i^{\star}-\frac{\alpha}{T},\theta_i^{\star}+\frac{\alpha}{T}\right)$.\Comment{Block interval to avoid finding the same peak}
  \EndFor

  \State \textbf{Output:}
  \emph{Distinct dominant eigenvalues:}
  $\{\theta^\star_i\}^{\tilde{I}}_{i=1}$
  \State
  \State \textbf{Stage II: Estimate the multiplicity of each dominant eigenvalue}
  \State \Comment{\textcolor{blue}{Step 4: Estimate multiplicity}}
  \For{$i=1$ to $\tilde{I}$}
  \State Calculate SVD of the data matrix at each $\theta^\star_i$:
  \[
    U_i \Sigma_i V_i^\dagger=\text{SVD}\left(\frac{1}{N}\sum^N_{n=1}{Z_{:,:,n}}\exp(\i\theta^\star_i t_n)\right)
  \]
  \State $m_i = \#\left\{ j \;\middle|\; (\Sigma_i)_{j,j} > \tau \right\}$.
  \EndFor
  \State \textbf{Output:} \emph{Multiplicities:}  $\{m_i\}_{i=1}^{\tilde{I}}$.
    \end{algorithmic}
\end{breakablealgorithm}

In the above algorithm, $\sigma T$ denotes the maximal Hamiltonian simulation time used in the generalized Hadamard test, and $N$ is the number of repetitions for \textbf{each} entry estimation. Therefore, the total Hamiltonian simulation time is upper bounded by $LR\sigma TN$.
The importance of off-diagonal entries in the data tensor $Z$ for resolving spectral multiplicities is demonstrated through a concrete illustrative example in \cref{sec:illustrative_example}.

Besides Gaussian functions, {one} may also use the Kaiser windows and the discrete prolate spheroidal sequence (DPSS/Slepian) windows \cite{slepian1978prolate,kaiser2003use,pw93} to further suppress spectral leakage and sharpen cluster localization. We leave a detailed study of this to future work.

\subsection{Complexity analysis}\label{sec:complex}
In this section, we analyze the theoretical complexity of the QFAMES algorithm introduced in Section~\ref{sec:QFAMES}.
Recall the overlaps of the dominant eigenvalues and tails defined as \cref{eq:def_domination}:
\begin{equation}
    \begin{aligned}
        p_{\min}= &\min_{m \in \mc{D}}\sqrt{\left(\sum^L_{l=1}\left|\Phi_{l,m}\right|^2\right)\left({\sum^R_{r=1}}\left|\Psi_{{r},m}\right|^2\right)},\\
        p_{\rm tail}= &\sum_{m\notin\mc{D}} \sqrt{\left(\sum^L_{l=1}\left|\Phi_{l,m}\right|^2\right)\left({\sum^R_{r=1}}\left|\Psi_{{r},m}\right|^2\right)}\,,
    \end{aligned}
\end{equation}
The main theorem on the performance guarantees of the QFAMES algorithm is stated as follows:

\begin{theorem}\label{thm:main} Assume the uniform overlap condition (\cref{asp:linear_dep}) {holds} true,
 $p_{\rm tail}/p_{\min}$ is
sufficiently small, and the parameters in \cref{alg:QFAMES} satisfy the following conditions:
\begin{equation}
    \begin{aligned}
        & N=\widetilde{\Omega}(LRp^{-2}_{\rm tail}),
        \quad T=\widetilde{\Omega}(1/\Delta),\\
        & \sigma=\widetilde{\Omega}(1),
        \quad \wt{I}\geq I,
        \quad \tau=\Theta(p_{\rm tail}),\\
        & q=\mc{O}\left( \frac{p_{\rm tail}}{(1+\sigma){\sqrt{|\mc{D}|LR}}}\right),\\
        & \alpha=\Omega\left( \log\left( \frac{|\mc{D}|LR}{p_{\rm tail}}\right)\right), \quad \alpha=\mc{O}(\Delta T).
    \end{aligned}
\end{equation}

Given failure probability $\eta\in (0,1)$, {there exists a permutation $\pi:[I]\to[I]$ such that}, with probability at least $1-\eta$, we have
\begin{align}
    {\lambda^{\star}_{\pi(i)}} \in \left[\theta^\star_i - \epsilon, \, \theta^\star_i + \epsilon \right], \quad \forall\, 1 \leq i \leq I,
\end{align}
where the confidence interval length $\epsilon$  is given by
\begin{align}
    \label{eq:location_err}
    \epsilon = \widetilde{\mathcal{O}}\!\left( \frac{p_{\rm tail}}{p_{\min}} \cdot \frac{1}{T} \right),
\end{align}
and
\begin{align}
    m_i = {|\mathcal{D}_{\pi(i)}|}, \quad 1 \leq i \leq I,
\quad m_i = 0, \quad I < i \leq \widetilde{I}.
\end{align}

\end{theorem}
We defer a more general and rigorous version of the above theorem, along with its proof, to~\cref{sec:rigorous_version}. The above theorem shows that \cref{alg:QFAMES} can accurately estimate the locations of the dominant eigenvalues and identify their multiplicities. Note that the conditions on $T, N, \sigma, q$, and $\alpha$ are similar to those in~\cite[Theorem 3.2]{ding2024quantum} (see the detailed discussion in~\cref{rem:thm_condition}), and the choice of $\tau$ plays an important role in filtering out non-dominant eigenvalues.~{In practice, exact values for $p_{\min}$ and $p_{\rm tail}$ are often unknown. The parameters $T, N, \sigma, q$, and $\alpha$ are flexible, typically subject only to upper or lower bounds, whereas selecting $\tau$ requires care to avoid overestimating or underestimating multiplicity. When $p_{\rm tail}$ is unknown, two heuristic methods for selecting the SVD threshold $\tau$ are as follows: (1) Noise-based selection: set $\tau \gtrsim \sqrt{LR/N}$. Since each entry of $G(\theta)$ is estimated from $N$ samples, the entry-wise statistical fluctuation is of order $N^{-1/2}$, so the Frobenius norm of the corresponding noise matrix is of order $\sqrt{LR/N}$; (2) Singular value drop: in low-noise regimes, the effective rank deficiency of $G(\theta)$ typically yields a visible separation between the leading singular values and the remainder; in this case, $\tau$ can be chosen within this gap.}

It can be directly obtained from \cref{thm:main} that the maximal evolution time {$T_{\max} = \sigma T = \tilde{\mc{O}}(p_{\rm tail}p_{\min}^{-1}\epsilon^{-1})$}, and the total evolution time {$T_{\rm total} = \mc{O} (NLR T_{\max}) = \tilde{\mc{O}} (L^2 R^2 (p_{\rm tail}p_{\min})^{-1} \epsilon^{-1})$}. When $p_{\rm tail}$ is close to zero and~{$p_{\min}=\Omega(1)$}, however, the maximal evolution time will instead reach the bound from below as $T_{\max} = \tilde{\Omega}(\Delta^{-1})$, {where $\Delta$ is the spectral gap}. In this case, we can use a looser constant $c = \tilde{\Theta} (\Delta^{-1}\epsilon)$ to replace $p_{\rm tail}$ in the scalings and the total evolution time will consequently scale as $T_{\rm total}=\wt{\mc{O}}(L^2 R^2\Delta\epsilon^{-2})$.
Note that we aim to resolve the gap $\Delta$ between distinct dominant eigenvalues, which implicitly indicates that $\Delta > \epsilon$.

Furthermore, our algorithm also applies to the case of approximately degenerate dominant eigenvalues, where a cluster of such eigenvalues lies within a narrow interval of width $\delta \ll \Delta$ around $\lambda^\star_i$. In this setting, the algorithm can efficiently identify the locations of the clusters and determine the exact number of dominant eigenvalues in each. The corresponding assumptions and results are presented in detail in~\cref{sec:rigorous_version}.

\section{Extension to observable measurements}~\label{sec:obs}

Within each located {dominant eigenvalues}, other than multiplicity, more meaningful information can be encoded in expectation values of physical observables. In this section, we will show that it can also be extractable under the QFAMES framework.

The main idea is to slightly modify our circuit to include the observables. Given an observable $O$, it can always be rewritten as the weighted sum of unitary Hermitian operators such as Pauli matrices. We can thus assume $O$ to be unitary without loss of generality.
The measurement outcomes $Z_{l,r,n}^{O}$ of the modified circuit
\begin{widetext}
\begin{align}
    \label{eq:generalized_hadamard_obs}
    \Qcircuit @C=1em @R=1em {
    \lstick{|0\rangle} & \gate{\rm H} & \ctrlo{1} & \ctrl{1} & \ctrlo{1} & \ctrl{1} & \ctrl{1} &  \gate{W} & \gate{\rm H} & \meter \\
    \lstick{\ket{0}}   & \qw          & \gate{U_l} & \gate{V_r}       & \gate{e^{-\mathbf{i}t'_{n} H}} & \gate{e^{-\mathbf{i}t_n H}}  & \gate{O} &\qw & \qw & \qw
    }
\end{align}
\end{widetext}
provides an unbiased estimator of
\begin{align}
    \mathbb{E}[Z_{l,r,n}^{O}]= \mc{Z}_{l,r}^{O}(t_n, t'_{n}) :=  \langle \phi_l|e^{\mathbf{i} H t'_{n}} O e^{-\mathbf{i}Ht_n} |\psi_r\rangle\,.
    \label{eq:obs_measure}
\end{align}
Analogous to the QFAMES algorithm, the next step is to {contract the time dimension} of the 3-tensor $Z_{l,r,n}^{O}$ with the complex exponential vector
\begin{equation}
    \begin{aligned}
     &\frac{1}{N}\begin{bmatrix}
        e^{\i \theta (t_1 - t_1')} & e^{\i \theta (t_2 - t_2')} & \cdots & e^{\i \theta (t_N - t_N')}
    \end{bmatrix}\,,
\end{aligned}
\end{equation}
where both $\left\{t_n\right\}$ and ${\left\{t'_{n}\right\}}$ are sampled independently from the truncated Gaussian density distribution $a_{T/\sqrt{2}}^{\rm trunc}(t)$ defined in \cref{eqn:a_T}. This results in an $L$-by-$R$ matrix
\begin{align}\label{eqn:G_O}
    G^{O}(\theta)_{l,r} = \frac{1}{N}\sum_{n=1}^N Z^{O}_{l,r,n} e^{\i \theta (t_n - t_n')}~~~\forall l\in [L], r\in [R]\,,
\end{align}
whose expectation values are
\begin{equation}
    \begin{aligned}
    \label{eq:exp_go}
    \mathbb{E}[G^{O}(\theta)_{l,r}]=&\mathbb{E}\left[\langle \phi_l|e^{\mathbf{i} (H-\theta) t'_{n}} O e^{-\mathbf{i}(H-\theta)t_n} |\psi_r\rangle\right]\\
    =&\langle \phi_l|\mathbb{E}\left[e^{\mathbf{i} (H-\theta) t'_{n}}\right] O \mathbb{E}\left[e^{-\mathbf{i}(H-\theta)t_n} \right]|\psi_r\rangle\\
    {\approx} &\braket{\phi_l | e^{-(\theta - H)^2T^2/2} O e^{-(\theta - H)^2 T^2/2} |\psi_r  } \\
    = & \Phi \cdot \diag\left(\left\{e^{-(\theta-\lambda_k)^2 T^2/2}\right\}_{k\in [M]}\right) \cdot O_{[M]} \\
    &\cdot \diag\left(\left\{e^{-(\theta-\lambda_{k'})^2 T^2/2}\right\}_{k'\in [M]}\right) \cdot \Psi^\dagger\,,
    \end{aligned}
\end{equation}
where the matrix $O_{[M]}:=
\left(\braket{E_k | O | E_{k'}}\right)_{k,k'\in[M]}$ is written in the Hamiltonian eigenstate basis. In the second equality, we use the fact that $t_n,t'_n$ are independent random variables. If we use the Gaussian density distribution {without truncation}, in \cref{eq:exp_go} the {approximately equal sign should be instead an equal sign}.

Let us focus on degenerate dominant eigenstates located at $\theta_i$ given by \cref{alg:QFAMES} with multiplicity $m_i$. Recall that we have obtained the SVD $G(\theta_i) = U_i \Sigma_i V_i^{\dagger}$, and can thus restrict the matrices $G(\theta_i)$ and $G^{O}(\theta_i)$ to the large singular value subspace with cutoff threshold $\tau$:
\begin{align}\label{eqn:truncated_matrix}
    \tilde{G}(\theta_i) & =\left((U_i)_{:,[m_i]}\right)^{\dagger} \cdot  G(\theta_i) \cdot  \left((V_i)_{:,[m_i]}\right) =(\Sigma_i)_{[m_i]}\\
    \tilde{G}^{O}(\theta_i) & =\left((U_i)_{:,[m_i]}\right)^{\dagger} \cdot  G^{O}(\theta_i) \cdot  \left((V_i)_{:,[m_i]}\right)\,,
\end{align}
where $(\Sigma_i)_{[m_i]}$ is the diagonal matrix of $m_i$ singular values of $G(\theta_i)$ that are larger than the SVD threshold $\tau$, and $(U_i)_{:,[m_i]}$ and $(V_i)_{:,[m_i]}$ are the submatrices of $U_i$ and $V_i$ that contain the corresponding columns. The rest is to solve the generalized eigenvalue problem
\begin{align}
    \tilde{G}^O(\theta_i) w = \lambda^O \tilde{G}(\theta_i)w
    \label{eq:generalized_eig_prob}
\end{align}
with regard to a low dimension, full rank matrix $G(\theta_i)$, and the resulting eigenvalues $\lambda^O_1,\cdots,\lambda_{m_i}^O$ will be approximately those of the matrix
\begin{align}\label{eqn:O_D_i}
    O_{\mc{D}_i}:= \left(\braket{E_k | O | E_{k'}}\right)_{k,k'\in \mc{D}_i},\quad \mc{D}_i = \left\{ k\in \mc{D}: \lambda_k \in \mc{I}_i \right\},
\end{align}
which is the projection of $O_{[M]}$ onto the subspace spanned by the dominant eigenstates indexed by $\mc{D}_i$ with eigenvalue $\lambda_i^{\star}$.
{We note that when the observable $O$ and the Hamiltonian do not commute, there is not a one-to-one correspondence between the generalized eigenvalues $\lambda_k^O$ and the Hamiltonian eigenstate expectations $\braket{E_k | O |E_k}$. Instead, it provides the possible range of expectation values of observable $O$ within the corresponding dominant eigenvector subspace $\mc{E}_i$, which is $[\min_{k} \lambda^O_k, \max_k \lambda^O_k]$. }

The correctness of this result can be justified as follows. When $T$ is large enough and the random noise is negligible, we have
\begin{align}
    G(\theta_i)\approx {\Phi_{:,\mc{D}_i}} \cdot \diag\left(\left\{e^{-(\theta_i-\lambda_k)^2 T^2}\right\}_{k\in {\mc{D}_i}}\right)  \cdot ({\Psi_{:,\mc{D}_i}})^\dagger\,.
\end{align}
Thus, there must exists two invertible matries $P_U\in\mathbb{C}^{|m_i|\times |m_i|}$ and $P_V\in\mathbb{C}^{|m_i|\times |m_i|}$ such that
\begin{align}
    {\Phi_{:,\mc{D}_i}}=(U_i)_{:,[m_i]}P_U,\quad {\Psi_{:,\mc{D}_i}}=(V_i)_{:,[m_i]}P_V\,.
\end{align}
Plugging this into the formula of $\tilde{G}^{O}(\theta_i)$ and $\tilde{G}(\theta_i)$, we obtain
\begin{widetext}
    \begin{align}
    \tilde{G}^{O}(\theta_i)=P_U \cdot \diag\left(\left\{e^{-(\theta_i-\lambda_k)^2 T^2/2}\right\}_{k\in {\mc{D}_i}}\right) \cdot O_{\mc{D}_i} \cdot \diag\left(\left\{e^{-(\theta_i-\lambda_{k'})^2 T^2/2}\right\}_{k'\in {\mc{D}_i}}\right) \cdot P^\dagger_V,
\end{align}
and
\begin{align}
    \tilde{G}(\theta_i)=\underbrace{P_U \cdot \diag\left(\left\{e^{-(\theta_i-\lambda_k)^2 T^2/2}\right\}_{k\in {\mc{D}_i}}\right)}_{\text{full rank}} \cdot \underbrace{\diag\left(\left\{e^{-(\theta_i-\lambda_{k'})^2 T^2/2}\right\}_{k'\in {\mc{D}_i}}\right) \cdot P^\dagger_V}_{\text{full rank}}\,.
\end{align}
\end{widetext}
Compare the above two formulas, it is straightforward to see that the generalized eigenvalue problem in Eq.~(\ref{eq:generalized_eig_prob}) will give eigenvalues of $O_{\mc{D}_i}$.

\begin{breakablealgorithm}
\caption{Extracting observable expectation values under the QFAMES framework}
  \label{alg:QFAMES_obs}
  \begin{algorithmic}[1]
  \State \textbf{Preparation:}

  Number of data pairs: $N$;

  Two initial state sets: $\{\phi_i\}_{i=1}^L$ and $\{\psi_j\}_{j=1}^R$;

  Filter parameter: $T$;

  Truncation parameter in $a_T^{\rm trunc}$: $\sigma$;

  Unitary observable: $O$;

  \hspace*{-2em}\textbf{Output from \cref{alg:QFAMES}:}

  Location of degenerate dominant eigenvalue:
  $\theta_i$;

  Degeneracy: $m_i$;

  SVD result of $G(\theta_i)$: $G(\theta_i) = U_i \Sigma_i V_i^{\dagger}$.

  \State \textbf{Running:}
  \State \Comment{\textcolor{blue}{Step 1: Generate data}}
    \State Sample $\{t_n\}$ and $\{t'_{n}\}$ i.i.d. from the truncated Gaussian distribution {$a_{T/\sqrt{2}}^{\rm trunc}(t)$} defined as Eq.~\eqref{eqn:a_T}
  \For{$l=1$ to $L$}
    \For{$r$=1 to $R$}
        \For{$n=1$ to $N$}
            \State
            \multiline{Generate $Z^O_{l,r,n}\in \{\pm 1\pm \i\}$ using generalized Hadamard test {in \cref{eq:generalized_hadamard_obs}} with $(U_l, V_r, t_n, t'_{n}, O)$}
        \EndFor
    \EndFor
  \EndFor
  \State \Comment{\textcolor{blue}{Step 2: Solve the generalized eigenvalue problem}}
  {
  \State Calculate
  \[
    G^O(\theta_i) \gets \frac{1}{N}\sum^N_{n=1}{Z^O_{:,:,n}}\exp(\i\theta_i (t_n-t'_n)).
  \]
}
  \State Calculate matrices
    \begin{align*}
        \tilde{G}^{O}(\theta_i) & =\left((U_i)_{:,[m_i]}\right)^{\dagger} \cdot  G^{O}(\theta_i) \cdot  \left((V_i)_{:,[m_i]}\right)\\
        \tilde{G}(\theta_i) & =\left((U_i)_{:,[m_i]}\right)^{\dagger} \cdot  G(\theta_i) \cdot  \left((V_i)_{:,[m_i]}\right) =(\Sigma_i)_{[m_i]}\,,
    \end{align*}
  \State Solve the generalized eigenvalue problem
    \begin{align*}
        \tilde{G}^O(\theta_i) w = \lambda^O \tilde{G}(\theta_i)w.
    \end{align*}
  \State \textbf{Output:} \emph{Eigenvalues:}  $\{\lambda^O_k\}_{k=1}^{m_i}$ that provide estimations of eigenvalues of $O_{\mc{D}_i}$.

    \end{algorithmic}
\end{breakablealgorithm}

Below, we will provide an informal resource analysis for the observable estimation algorithm. For simplicity, we focus on the second regime we discussed at the end of \cref{sec:complex}.  The proof is given in Appendix~\ref{sec:ana_ob}.
\begin{proposition}\label{prop:observable}
    In additional to the assumption in \cref{thm:main}, assume that ${|\mc{D}_i|},L,R=\mathcal{O}(1)$, access to a Pauli observable $O$, and
    \begin{enumerate}
        \item $p_{\min} = \Omega(1)$, which is reflected by the SVD threshold $\tau$ in the algorithm;
        \item there are no tail eigenvalues, i.e., $p_{\rm tail} = 0$.
    \end{enumerate}
    To achieve $\epsilon_O$-accuracy in estimating eigenvalues of $O_{\mc{D}_i}$, it suffices to choose
    \begin{align}
        T_{\max}=\wt{\mc{O}}(\Delta^{-1}), \quad
        N=\tilde{\Omega}(\epsilon_O^{-2} ),
    \end{align}
    and the total evolution time
    \begin{align}
        T_{\rm total} = \wt{\mc{O}}(NLRT_{\max}) =
        \wt{\mc{O}}(\Delta^{-1}\epsilon_O^{-2}).
    \end{align}
\end{proposition}

{We note that the constant in~\cref{prop:observable} depends on polynomially on $p^{-1}_{\min}$, as seen from the proof in Appendix~\ref{sec:ana_ob}. Here, we assume $p_{\min} = \Omega(1)$ to avoid this amplification factor and only focus on the scaling with respect to $\epsilon_O$ and $\Delta$.} It is instructive to contrast the generalized eigenvalue problem in \cref{eq:generalized_eig_prob} with those arising in quantum subspace methods for eigenvalue estimation~\cite{Parrish2019QuantumFD,Huggins2020,Stair2020,Seki2021,Epperly2022,Baek2023,Kirby2024,Yoshioka2025}. According to~\cite[Theorem 2.7]{Epperly2022}, if each entry of the projected Hamiltonian and overlap matrices is estimated using $N_s$ samples, then the statistical error per entry scales as $\mathcal{O}(N_s^{-1/2})$. The corresponding error in the eigenvalue of interest scales as $\mathcal{O}\left(N_s^{-1/(2(1+\gamma))}\right)$, where $0 \le \gamma \le \tfrac{1}{2}$ arises from a weighted geometric mean inequality for the projected matrices. Thus, the convergence can be slower than the shot-noise-limited scaling $\mathcal{O}(N_s^{-1/2})$ whenever $\gamma>0$. Moreover, the error bound in~\cite{Epperly2022} includes additional amplification factors that depend polynomially on the subspace dimension $D$ {(the dimension of the data matrix)}; in practice, $D$ is therefore typically kept small.
{In our setting, the generalized eigenvalue problem in \cref{eq:generalized_eig_prob} involves an $L\times R$ matrix, so the relevant dimension is $D=\max(L,R)$. In the regime of \cref{prop:observable}, $L$ and $R$ are constants, and improved accuracy is obtained by increasing the maximal evolution time $T_{\max}$ and the number of samples used to estimate each matrix entry, rather than by enlarging the data matrix dimension.}

\section{Ancilla-free alternative to the generalized Hadamard test}~\label{sec:ancilla_free}

In the QFAMES algorithm, the generalized Hadamard test circuits \cref{eq:generalized_hadamard} and \cref{eq:generalized_hadamard_obs} are used to generate the quantum data  $Z_{l,r,n}$, and for observables, {$Z^O_{l,r,n}$}, which requires controlled Hamiltonian evolution and initial state preparation circuits. Such controlled operations can be challenging to implement on near-term and even early fault-tolerant quantum devices. To address this limitation, several alternatives to the generalized Hadamard test have recently been proposed that are either ancilla-free~\cite{dong2022ground,Yang2024,Clinton2024,Wang2025,Yi2025}, or substantially reduce {the} requirements of controlled operations~\cite{Schiffer2025}. The method in \cite{Yang2024} is particularly suitable for QFAMES, as it
\begin{enumerate}
    \item exactly allows to measure ${\mc Z}_{l,r}(t)=\braket{\phi_l | e^{-\mathbf{i}Ht_n} | \psi_r}$ with different left and right initial states $\ket{\phi_l}$ and $\ket{\psi_r}$;
    \item can produce the data for $t$ in the time interval $[0,\sigma T]$ in a single execution of the algorithm;
    \item can be easily generalized to the cases involving multiple time evolutions, thereby allowing access to measuring {${\mc Z}^O_{l,r}(t,t')$} in the observable version of QFAMES.
\end{enumerate}

Below we will briefly introduce this algorithm.
A key observation is that the main difficulty of measuring ${\mc{Z}}_{l,r}(t) = r(t)e^{\mathbf{i}{\varphi}(t)}$ lies in the phase  $\varphi(t)$ rather than the absolute value $r(t)$, and these two quantities are related via complex analysis. Notice that for a bounded Hamiltonian, the analytic continuation ${\mc Z}_{l,r}(z) = \braket{\phi_l | e^{-\i H z } | {\psi_r}}$ is a holomorphic function on the complex plane $z = t - \mathbf{i} \beta$, and so is $\ln {\mc Z}_{l,r}(z) = \ln r(z) + \mathbf{i} \varphi(z)$ when ${\mc Z}_{l,r}(z)\neq 0$. Therefore, the Cauchy-Riemann equation connects the absolute value and phase as
\begin{align}
    \frac{\mathrm{d} \varphi(t)}{\mathrm{d} t} =
    \frac{\partial \ln r(z)}{\partial \beta}\Bigg\vert_{\beta=0} \nonumber
    \approx & \frac{1}{2h}\left[ \ln r(t-\mathbf{i}h) - \ln r(t+\mathbf{i}h) \right].
    \label{eq:phase}
\end{align}

The last step is the finite difference approximation and $h$ is a small constant. Assuming knowledge of $\varphi(0)$, one can obtain the phases $\varphi(t)$ for $t$ in the desired interval by integrating \cref{eq:phase}. Thus it is only required to measure the quantities $r(t) = \left|\braket{\phi_l|e^{-\mathbf{i}Ht} |\psi_r}\right|$ and $r(t\pm \mathbf{i}h) = \left|\braket{\phi_l|e^{-\mathbf{i}Ht} e^{\pm hH}|\psi_r}\right|$. The corresponding circuits are
\begin{align}
    \Qcircuit @C=1em @R=1em {
    \lstick{\ket{0}}   & \gate{V_r}   & \gate{e^{-\mathbf{i}t H}}   & \gate{U_l^{\dagger}} & \meterB{\ket{0}} & \qw
    } \\
    \Qcircuit @C=1em @R=1em {
    \lstick{\ket{0}}   & \gate{V_r}   &  \gate{ e^{\pm hH} / c_{\psi_r,\pm}} & \gate{e^{-\mathbf{i}t H}}   & \gate{U_l^{\dagger}} & \meterB{\ket{0}} & \qw
    }
\end{align}
which provide unbiased estimations of $r(t)^2$ and ${r(t\pm \mathbf{i}h)^2} / {c_{\psi_r,\pm}^2}$, respectively.
In addition to the necessary state preparations and real-time evolutions, these circuits require at most \emph{a single short step} of imaginary time evolution to one side of the initial states, and the last step is a projective measurement onto computational basis state $\ket{0}$.

Note that the imaginary time evolution is non-unitary and thus a pre-computed normalization factor $c_{\psi_r,\pm} = \sqrt{\braket{\psi_{r} | e^{\pm 2h H} |\psi_r}}$ is required in the quantum circuit. When the initial state $\ket{\psi_r} = V_r\ket{0}$ has short correlation length and the Hamiltonian $H$ is local, the single step of imaginary time evolution can be approximated by local unitary evolution operators times $c_{\psi_r,\pm}$, as shown in the quantum imaginary time evolution (QITE) algorithm~\cite{Motta2020}. If it is not the case, there also exist alternatives to QITE. For example, one can introduce a single auxiliary qubit to implement the single-step imaginary-time evolution via probabilistic imaginary-time evolution (PITE) algorithms~\cite{Lin2021, Kosugi2022}, while still avoiding the requirement for controlled long real-time evolutions.

To measure ${\mc{Z}}^O_{l,r}(t,t')$, we just need to view $t$ and $t'$ as two separate variables and perform the algorithm from both sides. The modified circuits are as follows:
\begin{widetext}
    \begin{align}
    \Qcircuit @C=1em @R=1em {
    \lstick{\ket{0}}   & \gate{V_r}   & \gate{e^{-\mathbf{i}t H}}  & \gate{O} & \gate{e^{\mathbf{i}t' H}} & \gate{U_l^{\dagger}} & \meterB{\ket{0}} & \qw
    } & \quad \left|  \mc{Z}^O_{l,r}(t,t') \right|^2\\
    \Qcircuit @C=1em @R=1em {
    \lstick{\ket{0}}   & \gate{V_r}   &  \gate{ {e^{\pm hH}} / {c_{\psi_r,\pm}} }   & \gate{e^{-\mathbf{i}t H}} & \gate{O} & \gate{e^{\mathbf{i}t' H}}   & \gate{U_l^{\dagger}} & \meterB{\ket{0}} & \qw
    }  & \quad  \left| \mc{Z}^O_{l,r}(t\pm \i h,t') \right|^2 / c_{\psi_r,\pm}^2\\
    \Qcircuit @C=1em @R=1em {
    \lstick{\ket{0}}   & \gate{V_r}   & \gate{e^{-\mathbf{i}t H}}  & \gate{O} & \gate{e^{\mathbf{i}t' H}} & \gate{ {e^{\pm hH}} / {c_{\phi_l,\pm}} }  & \gate{U_l^{\dagger}} & \meterB{\ket{0}} & \qw
    }  & \quad \left| \mc{Z}^O_{l,r}(t,t'\pm \i h) \right|^2 / c_{\phi_l,\pm}^2
\end{align}
\end{widetext}

{
\section{QFAMES with Mixed-State Inputs}\label{sec:QFAMES_mixed}
}
We have presented QFAMES assuming that the initial states are pure states and we have access to the controlled unitary operators that prepare these states. In this section, we generalize the observable version to mixed states and remove the requirement for controlled state-preparation gates. Note that the {multiplicity} estimation version of QFAMES can be naturally embedded in the observable version of QFAMES. For simplicity, we only present the observable version of QFAMES with mixed-state inputs here.

We first introduce the algorithm. Assume that we have access to a set of mixed states $\left\{ \rho_l \right\}_{l \in [L]}$ and $\left\{ \sigma_r \right\}_{r\in [R]}$.

\paragraph*{Step 1: Generate data.}

Analogous to the case of pure initial states, we aim to measure quantities of the form
\begin{equation}
\begin{aligned}
    \mc{Z}_{l,r}^{\rm mixed}(t_n, t_n') :&= \tr \left( \rho_l e^{-\mathbf{i}Ht_n}  \sigma_r e^{-\mathbf{i}Ht_n'}  \right),\\
    \mc{Z}_{l,r}^{O, {\rm mixed}}(t_n, t'_{n}, t_{n}'') :&= \tr \left( \rho_l e^{-\mathbf{i} H t_{n}} O e^{-\mathbf{i}Ht_n'}  \sigma_r e^{-\mathbf{i}Ht_n''}  \right).
\end{aligned}
\end{equation}
Their unbiased estimates can be obtained through the following circuits using a swap test:

\begin{widetext}
\begin{align}
    \label{eq:swap_mixed}
    Z_{l,r,n}^{\rm mixed}: \qquad
    \raisebox{2.5em}{
    \Qcircuit @C=1em @R=1em {
    \lstick{\ket{0}} & \gate{\rm H} &  \ctrl{1} &  \ctrl{2} &  \ctrl{2} & \gate{W} & \gate{\rm H} & \meter \\
    \lstick{\rho_l}   & \qw  & \gate{e^{-\mathbf{i}t_n' H}}  & \qw & \qswap & \qw & \qw & \qw\\
    \lstick{\sigma_r}   & \qw  &\qw    & \gate{e^{-\mathbf{i}t_n H}} & \qswap & \qw & \qw & \qw
    }
    }
\end{align}
\begin{align}
    \label{eq:swap_mixed_ob}
    Z^{O,{\rm mixed}}_{l,r,n}: \qquad
    \raisebox{2.5em}{
    \Qcircuit @C=1em @R=1em {
    \lstick{\ket{0}} & \gate{\rm H} &  \ctrl{1} &  \ctrl{2} &  \ctrl{2} &  \ctrl{2} &  \ctrl{2} & \gate{W} & \gate{\rm H} & \meter \\
    \lstick{\rho_l}   & \qw       & \gate{e^{-\mathbf{i}t_n'' H}} & \qw  & \qw   & \qw & \qswap & \qw & \qw & \qw\\
    \lstick{\sigma_r}   & \qw & \qw    & \gate{e^{-\mathbf{i}t_n' H}} & \gate{O} &  \gate{e^{-\mathbf{i}t_n H}}     & \qswap & \qw & \qw & \qw
    }
    }
\end{align}
\end{widetext}

Both circuits only require the ability to prepare the mixed states $\rho_l$ and $\sigma_r$, and do not require~controlled state-preparation gates, which might be of independent interest.

\paragraph*{Step 2: Construct filtered data matrices.} The second step is to contract the time dimensions of the 3-tensors $Z_{l,r,n}^{{\rm mixed}}$ and $Z_{l,r,n}^{O,{\rm mixed}}$ with element-wise products of complex exponential vectors of the form $ e^{\i \theta t_n}$,
\begin{equation}
    \begin{aligned}
     \begin{bmatrix}
         f_n(\theta, T, \left\{ t_n \right\})
     \end{bmatrix} := \begin{bmatrix}
         e^{\i \theta t_n}
     \end{bmatrix},
\end{aligned}
\end{equation}
where the set $\left\{ t_n \right\}$ are sampled independently from the truncated Gaussian density distribution $a_{T}^{\rm trunc}(t)$. Similar to the pure state case, we define
\begin{equation}\label{eqn:G_mixed_construction}
    \begin{aligned}
        G^{\rm mixed}_{l,r}(\theta) = & \frac{1}{N}\sum_{n=1}^{N} Z_{l,r,n}^{\rm mixed} f_n(\theta, T,  \left\{ t_n \right\})f_n(\theta, T, \left\{ t_n' \right\})\\
        G^{O, {\rm mixed}}_{l,r}(\theta) = & \frac{1}{N}\sum_{n=1}^{N} Z_{l,r,n}^{O, \rm mixed} f_n\left(\theta, \frac{T}{\sqrt{2}},  \left\{ t_n \right\}\right)\\
        & f_n\left(\theta, \frac{T}{\sqrt{2}}, \left\{ t_n' \right\}\right)
        f_n\left(\theta, T, \left\{ t_n'' \right\}\right).
    \end{aligned}
\end{equation}
When $N$ is sufficiently large, we have
\begin{equation}
    \begin{aligned}
        G^{\rm mixed}_{l,r}(\theta) &\approx \mathbb{E} \left[ G^{\rm mixed}_{l,r}(\theta) \right]  = \tr \left( \rho_l(\theta, T) \sigma_r(\theta, T) \right),\\
        G^{O, {\rm mixed}}_{l,r}(\theta) &\approx \mathbb{E} \left[ G^{O, \rm mixed}_{l,r}(\theta) \right] = \tr \left( \rho_l(\theta, T) O \sigma_r(\theta, T) \right),
    \end{aligned}
\end{equation}
where $\rho(\theta, T):=e^{-(\theta - H)^2T^2/2} \rho e^{-(\theta - H)^2T^2/2}$ denotes the filtered density matrix of a given mixed state $\rho$.
If we vectorize the density matrices as $\rho \to \kett{\rho}$, we obtain
\begin{equation}\label{eqn:G_mixed_approx}
    \begin{aligned}
        & G^{\rm mixed}_{l,r}(\theta) \approx \bbrakett{ \rho_l(\theta, T)  | \sigma_r(\theta, T) } \\
        = & \sum_{i,j} (\rho_{l})_{i,j}(\sigma_{r})_{j,i}\exp(-((\theta-\lambda_j)^2+(\theta-\lambda_i)^2)T^2)\,,\\
        & G^{O, \rm mixed}_{l,r}(\theta) \approx \bbra{\rho_l(\theta, T)} \mathbb{I} \otimes O \kett{\sigma_r(\theta, T)}\,.
    \end{aligned}
\end{equation}

\paragraph*{Step 3: Estimate dominant eigenvalues and multiplicities.} According to~\cref{eqn:G_mixed_approx}, $\left\|G^{\rm mixed}(\theta)\right\|_F$ peaks when $\theta$ is close to the dominant eigenvalues of $H$. Similar to the pure state case, in the third step, we find the dominant eigenvalues $\{\theta^\star_k\}$ by locating the peaks of the function $\left\|G^{\rm mixed}(\theta)\right\|_F$ over $\theta$ following the same \emph{search and block} strategy (\cref{sec:QFAMES}) as in the pure state case.

\paragraph*{Step 4: Observable estimation.} For each identified dominant eigenvalue $\theta^\star_k$, we solve the following generalized eigenvalue problem:
\begin{equation}\label{eq:general_eig_mix}
    {G^{O, \rm mixed}(\theta^\star_k)} w= \lambda^{O} {G^{\rm mixed}(\theta^\star_k)} w.
\end{equation}
to obtain estimates of the extremal eigenvalues (after taking the maximal/minimal generalized eigenvalues $\lambda^{O}$) of the observable $O$ within the (nearly) degenerate dominant eigenspace of $H$. As in~\cref{sec:obs}, one may first restrict $G^{\rm mixed}(\theta_k^\star)$ and $G^{O,\rm mixed}(\theta_k^\star)$ to the large-singular-value subspace of $G^{\rm mixed}(\theta_k^\star)$ (using the cutoff threshold $\tau$) to obtain a low-dimensional, full-rank generalized eigenvalue problem.
We note that, because $\theta^\star_k$ is close to the dominant eigenvalue $\lambda^\star_k$ of $H$, the filtered density matrices $\rho_l(\theta^\star_k, T)$ and $\sigma_r(\theta^\star_k, T)$ are approximately supported on the subspace spanned by the (nearly) degenerate dominant eigenvectors of $H$. Let $\Pi_k$ denote the spectral projector onto the (nearly) degenerate dominant subspace associated with the cluster centered at $\theta_k^\star$. The physically meaningful range of the expectation values of $O$ within this manifold is
\begin{align}\label{eqn:max_psi_O}
    \lambda_{\min}(\Pi_k O\Pi_k),\quad \lambda_{\max}(\Pi_k O\Pi_k)\,.
\end{align}
Under the condition that the projection of $G^{\rm mixed}_{l,r}(\theta^\star_k)$ onto the degenerate eigenspace maintains sufficient rank (see~\cref{eqn:uniform_mixed}), the generalized eigenvalue problem~\eqref{eq:general_eig_mix} yields consistent estimates of these extremal eigenvalues, which is the same as in the pure state case in~\cref{sec:obs}.
For example, in the case when $L=R$ and $\rho_l=\sigma_l$ for all $l$, and the vectors $\{\kett{\rho_l(\theta^\star_k,T)}\}_{l\in [L]}$ span the degenerate dominant eigenspace in the vectorized representation, the maximal eigenvalue problem~\cref{eqn:max_psi_O} can be formulated as
\begin{equation}
\begin{aligned}
    &\max_{\ket{\psi}\in \mathrm{Ran}(\Pi_k)}\frac{\left\langle \psi|O|\psi \right\rangle}{\left\langle \psi|\psi\right\rangle}
    =\max_{\ket{\psi}\in \mathrm{Ran}(\Pi_k)}\frac{\tr\left(\ket{\psi}\bra{\psi}O\right)}{\left\langle \psi|\psi\right\rangle}\\
    =&\max_{\ket{\psi}\in \mathrm{Ran}(\Pi_k)}\frac{\mathrm{Tr}\left(\ket{\psi}\bra{\psi}O\ket{\psi}\bra{\psi}\right)}{\mathrm{Tr}\left(\ket{\psi}\bra{\psi}\ket{\psi}\bra{\psi}\right)}\\
    =& \max_{\sum_i |c_i|^2=1}\frac{\sum_{j,k}c^\star_j c_k\tr\left(\rho_j O\rho_k\right)}{\sum_{j,k}c^\star_j c_k \tr\left(\rho_j\rho_k\right)}=\max_i \lambda^O_i\,.
\end{aligned}
\end{equation}
where
\begin{equation}
    {G^{O, \rm mixed}(\theta^\star_k)} w= \lambda^{O}_i {G^{\rm mixed}(\theta^\star_k)} w.
\end{equation}

Analogous to the pure state case, the observable structure is better characterized by the pairing of mixed states $\rho_l$ and $\sigma_r$, which yields information inaccessible through direct measurements of $O$ on the individual states. Here, we provide a simple toy example to illustrate this. Assume that there is a two-fold degenerate ground state $\left\{ \ket{0}, \ket{1} \right\}$. In this ground state subspace an observable $Z$ has the form $Z = \begin{bmatrix}
    1 & 0 \\ 0 & -1
\end{bmatrix}$. Assume access to four density matrices, which are after the filtering process projected onto the ground state subspace as
\begin{equation}
\begin{aligned}
    & \rho_1 = \begin{bmatrix}
        0.6 & 0 \\ 0 & 0.4
    \end{bmatrix}\,,\\
    & \rho_2 = \begin{bmatrix}
        0.4 & 0 \\ 0 & 0.6
    \end{bmatrix}\,,\\
    & \rho_3 = \ket{+}\bra{+} = \frac{1}{2}\begin{bmatrix}
        1 & 1 \\ 1 & 1
    \end{bmatrix}\,,\\
    & \rho_4 = \ket{+\i}\bra{+\i} = \frac{1}{2}\begin{bmatrix}
        1 & -\i \\ \i & 1
    \end{bmatrix}\,.
\end{aligned}
\end{equation}
Here $\ket{+} = \frac{1}{\sqrt{2}} \left( \ket{0} + \ket{1} \right)$ and $\ket{+\i} = \frac{1}{\sqrt{2}} \left( \ket{0} + \i\ket{1} \right)$. Direct measurements of the observable $Z$ with these four density matrices yield
\begin{equation}
    \tr \left( \rho_{1,2} Z \right) = \pm 0.2, \quad \tr \left( \rho_{3,4} Z \right) = 0,
\end{equation}
which can not fully capture the ``ferromagnetic'' feature of the ground state manifold with regard to $O=Z$. However, in the mixed state version of the QFAMES algorithm, we have
\begin{equation}
    \begin{aligned}
        G^{\rm mixed}_{l,r} & = \tr \left( \rho_l \rho_r \right) = \begin{bmatrix}
        0.52 & 0.48 & 0.5 & 0.5\\
        0.48 & 0.52 & 0.5 & 0.5\\
        0.5 & 0.5 & 1 & 0.5\\
        0.5 & 0.5 & 0.5 & 1
    \end{bmatrix}, \\
    G^{O, \rm mixed}_{l,r} & = \tr \left( \rho_l O \rho_r \right) = \begin{bmatrix}
        0.2 & 0 & 0.1 & 0.1\\
        0 & -0.2 & -0.1 & -0.1\\
        0.1 & -0.1 & 0 & -0.5\i\\
        0.1 & -0.1 & 0.5\i & 0
    \end{bmatrix}.
    \end{aligned}
\end{equation}
Solve the generalized eigenvalue problem $G^{O, \rm mixed}_{l,r} w = \lambda^{O} G^{\rm mixed}_{l,r} w$, and we can correctly obtain $\min\lambda^O = -1$ and $\max\lambda^O = 1$.

The detailed algorithm is
summarized in \cref{alg:QFAMES_mixed}.

\begin{breakablealgorithm}
\caption{Extracting observable expectation values under QFAMES with mixed-state inputs}
  \label{alg:QFAMES_mixed}
  \begin{algorithmic}[1]
\State \textbf{Preparation:}

  Number of time samples: $N$;

  \multiline{
    Two sets of initial mixed states: $\{\rho_l\}^{L}_{l=1}$ and $\{\sigma_r\}^{R}_{r=1}$;
  }

  Filter parameter: $T$;

  Truncation parameter in $a_T^{\rm trunc}$: $\sigma$;

  Number of {dominant eigenvalues (guess)}: $\tilde{I}$;

  Singular value decomposition threshold: $\tau$;

  \multiline{
  Searching parameter: $q$ (we search with discrete
  energy step $q / T$);
  }

  \multiline{
  Block parameter: $\alpha$ (distinct dominant  eigenvalues are at least separated by $\alpha / T$);
  }

  \State \textbf{Running:}
  \State \Comment{\textcolor{blue}{Stage I: Estimate the location of dominant eigenvalues}}
  \State Sample $\{t_n,t'_n\}$ i.i.d. from the truncated Gaussian distribution $a_T^{\rm trunc}(t)$ defined in \cref{eqn:a_T}.
  \For{$l=1$ to $L$}
    \For{$r$=1 to $R$}
        \For{$n=1$ to $N$}
            \State
             Generate data $Z^{\rm mixed}_{l,r,n}$ using~\eqref{eq:swap_mixed}.
        \EndFor
    \EndFor
  \EndFor
  \State $J\gets \left\lfloor \frac{2\pi T}{q}\right\rfloor$.
    \State Generate discrete candidates: {$\theta_j\gets-\pi+\frac{jq}{T}$ for $j=0,1,\dots,J$.}
  \State Calculate~\eqref{eqn:G_mixed_construction} for $0\leq j\leq J$:
  \[
        \left\|G^{\rm mixed}(\theta_j) \right\|_F\gets \left\|\frac{1}{N}\sum_{n=1}^{N} {Z_{:,:,n}^{\rm mixed}} e^{\i \theta_j t_n}e^{\i \theta_j t'_n}\right\|_F\,.
  \]
  \State Running the search-and-block procedure with $\alpha$ in~\cref{alg:QFAMES} to estimate the dominant eigenvalues.

  \State \textbf{Output:}
\emph{Distinct dominant eigenvalues:}
$\{\theta^\star_i\}^{\tilde{I}}_{i=1}$.

  \State \Comment{\textcolor{blue}{Stage II: Estimate observable expectation values}}
    \State Sample $\{t_n,t'_n\}$ i.i.d. from the truncated Gaussian distribution {$a_{T/\sqrt{2}}^{\rm trunc}(t)$} defined as Eq.~\eqref{eqn:a_T} and $\{t''_n\}$ i.i.d. from $a_T^{\rm trunc}(t)$.
  \For{$l=1$ to $L$}
    \For{$r$=1 to $R$}
        \For{$n=1$ to $N$}
            \State
            \makecell[l]{Generate $Z_{l,r,n}^{O, \rm mixed}$ using~\eqref{eq:swap_mixed_ob}.}
        \EndFor
    \EndFor
  \EndFor

  \State Calculate
  \[
  \begin{aligned}
      G^{O,\rm mixed}(\theta^\star_i) \gets & \frac{1}{N}\sum^N_{n=1}{Z^{O,\rm mixed}_{:,:,n}}\exp(\i\theta^\star_i t_n)\\
      & \exp(\i\theta^\star_i t'_n)\exp(\i\theta^\star_i t''_n),\quad 0\leq i\leq \tilde{I}\,.
  \end{aligned}
  \]

  \State Solve the generalized eigenvalue problem
    \begin{align*}
        G^{O,\rm mixed}(\theta^\star_i) w = \lambda^O G^{\rm mixed}(\theta^\star_i)w.
    \end{align*}
    \State \textbf{Output:} \emph{Eigenvalues:}  $\{\lambda^O_k\}_{k=1}^{m_i}$ that provide estimates of eigenvalues of $O_{\mc{D}_i}$.

    \end{algorithmic}
\end{breakablealgorithm}

As previously noted, the success of mixed-state QFAMES requires that the initial mixed states have sufficient overlap with the (nearly) degenerate dominant eigenvectors of $H$. Under the assumption $\{\rho_l\}=\{\sigma_r\}$, this requirement is characterized by the following condition:
\begin{equation}\label{eqn:uniform_mixed}
\min_{ {H}\ket{\psi} = \lambda^\star_k \ket{\psi}} \,\,\max_{\substack{p_i\geq 0\\ \sum_i p_i=1}}\left\langle \psi\middle|\sum_i p_i\rho_i\middle|\psi\right\rangle=\omega(1)\,.
\end{equation}
When the goal is to estimate the extremal eigenvalues of $O$ within the low-energy subspace of $H$, there is extensive literature on preparing such mixed states using recently developed dissipative algorithms for low-temperature thermal state or ground state preparation~\cite{VerstraeteWolfIgnacioCirac2009,RoyChalkerGornyiEtAl2020,zhou2021symmetry,Cubitt2023,WangSnizhkoRomitoEtAl2023,LuLessaKimEtAl2022,Temme_2011,Mozgunov2020,shtanko2021preparing,RallWangWocjan2023,ChenKastoryanoBrandaoEtAl2023,ChenKastoryanoGilyen2023,DingLiLin2025,GilyenChenDoriguelloKastoryano2024,LiZhanLin2025,hahn2025provably,langbehn2025universal,lloyd2025quantumthermal,scandi2025thermal,ZhanDingHuhnEtAl2026,ding2025endtoendefficientquantumthermal,DingChenLin2024,wang2025lindbladdynamicsrigorousguarantees}. For example, certain physically interesting systems exhibit metastable states, often arising from energy or entropy bottlenecks in the dissipative dynamics~\cite{bergamaschi2025,Chen2025Local,gamarnik2024,rakovszky2024}. Characterizing the physical observables of these states is of significant interest in condensed matter physics and quantum materials science~\cite{Yin2025,Bovier2006,BinderYoung1986,DebenedettiStillinger2001,DiehlZoller2008,Turner2018a}. While reaching the true global thermal equilibrium in such systems typically requires exponential time, relaxing to a specific low-temperature metastable state is often a polynomial-time process~\cite{RezaGheissari2024,Gheissari2024Random,gheissari2025,bergamaschi2025_2,10.1145/3717823.3718236,koehler2024}. Consequently, starting from a random initial state allows one to efficiently populate the underlying metastable subspace, thereby making the condition in~\eqref{eqn:uniform_mixed} achievable.

\section{Numerical results}\label{sec:numerical_tests}

In this section, we numerically demonstrate the efficiency of QFAMES by applying it to {four} models: (1) A simple illustrative example in~\cref{sec:illustrative_example};
(2) A {Transverse-Field} Ising model (TFIM) in~\cref{sec:TFIM}; 
(3) 2D Toric {Code} example in~\cref{sec:toric_code}; {(4) Anisotropic Heisenberg (XXZ) model in~\cref{sec:xxz_model}.} For simplicity, we set the left and right initial states to be the same. 
In the first case, we provide three simple initial states that are not orthogonal to each other to estimate the multiplicity of the ground states. For the TFIM, we take five low-energy matrix product states with bond dimension 2 as the initial states. We not only compute the ground state degeneracies, but also probe ground state properties of different quantum phases by measuring observable expectation values.
In the third example, we employ an increasing number of random initial states to estimate the ground-state degeneracy of a model with topological order, thereby demonstrating that QFAMES is applicable in a broad range of settings. All codes and data needed to produce the results in this section are available online~\footnote{\href{https://github.com/lin-lin/qfames}{https://github.com/lin-lin/qfames}}.

\subsection{Illustrative example}\label{sec:illustrative_example}

\begin{figure}
    \centering
    \subfloat[\label{fig:illustrative_QFAMES_a}]{%
        \includegraphics[width=0.99\linewidth]{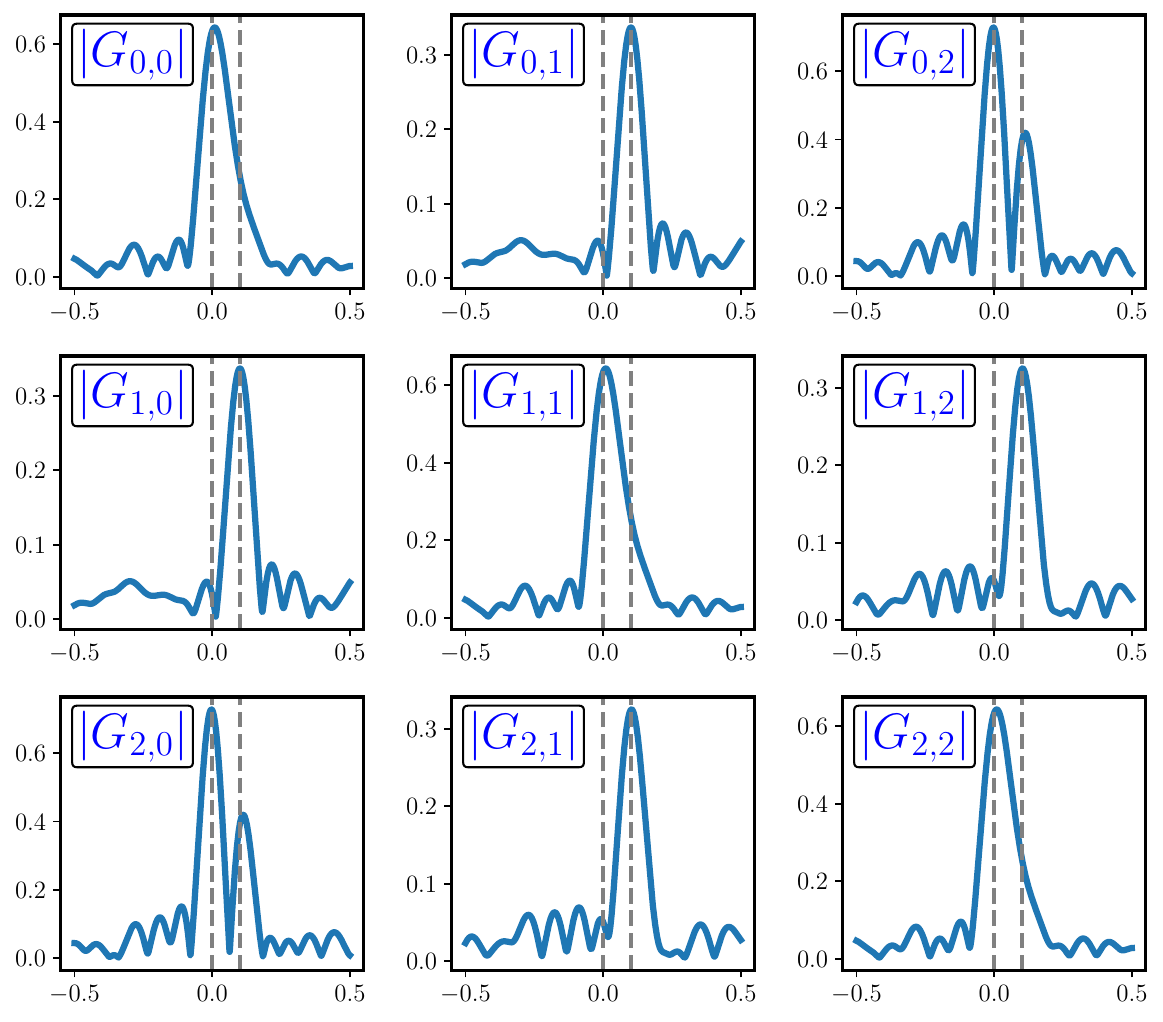}
    }

    \subfloat[\label{fig:illustrative_QFAMES_b}]{%
      \includegraphics[width=0.8\linewidth]{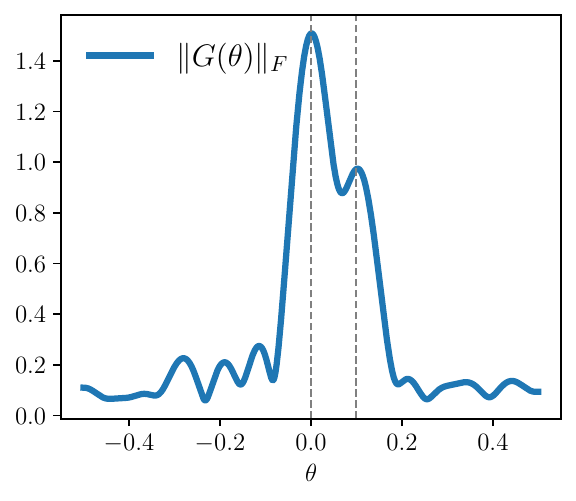}
    }

    \subfloat[\label{fig:illustrative_QFAMES_c}]{%
      \includegraphics[width=0.48\linewidth]{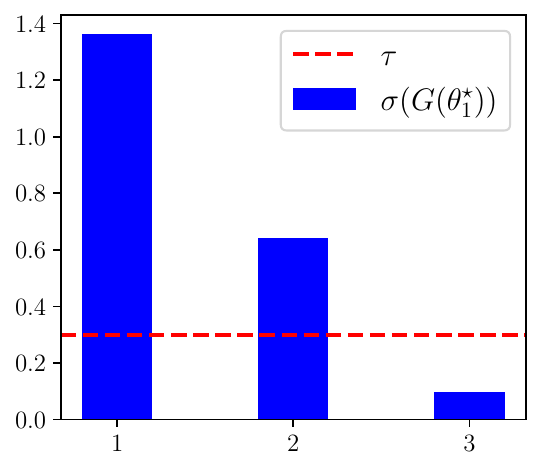}
    }\hfill
    \subfloat[\label{fig:illustrative_QFAMES_d}]{%
      \includegraphics[width=0.48\linewidth]{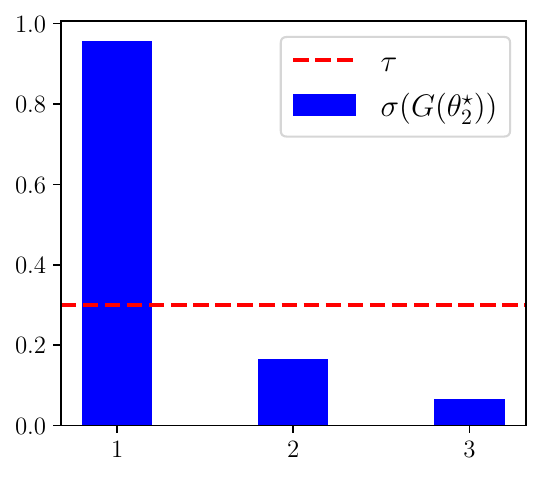}
    }
    \caption{Illustration of the QFAMES algorithm on the illustrative example. 
    (a): The magnitude of $G_{i,j}(\theta)=\frac{1}{N}\sum_{n=1}^NZ_{i,j,n}e^{\i \theta t_n}$ for each $0\leq i,j\leq 2$. {The $x$-axis stands for $\theta$.} The gray dashed lines indicate the positions of the eigenvalues. (b): The Frobenius norm $\|G(\theta)\|_F$, which clearly identifies the two distinct dominant eigenvalues. (c) and (d): The singular values of $G(\theta)$ at each $\lambda_i^{\star}$, {correctly} indicating the corresponding multiplicities when the threshold $\tau=0.1\sqrt{LR}=0.3$.}
    \label{fig:illustrative_QFAMES}
\end{figure}
\begin{figure}
    \centering
    \includegraphics[width=0.9\linewidth]{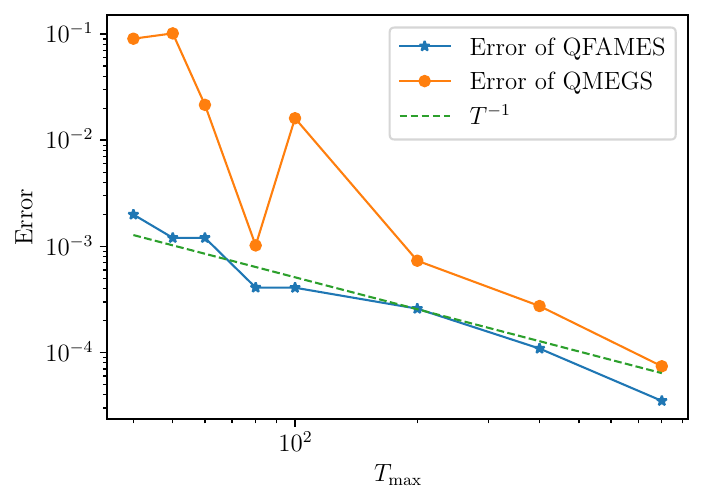}
 
     \includegraphics[width=0.9\linewidth]{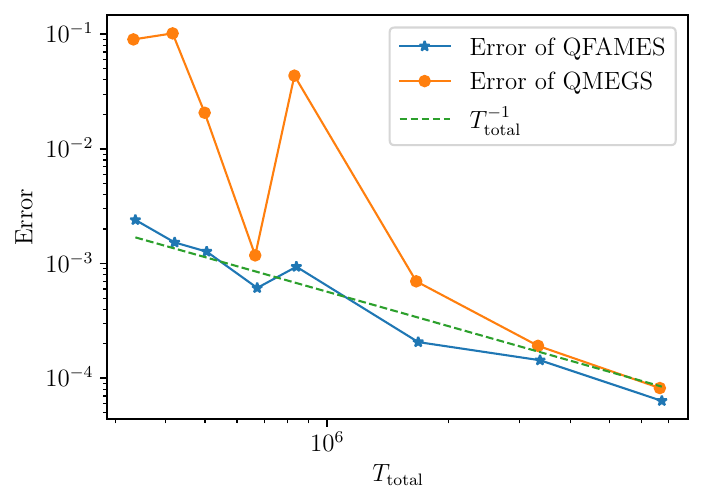}
     
     \caption{
     \label{fig:illustrative_model} QFAMES (\cref{alg:QFAMES}) vs QMEGS~\cite{ding2024quantum} in the illustrative example as a function of \textbf{(upper)}: max evolution time $T_{\max} = \sigma T$, and \textbf{(lower)}: total evolution time $T_{\rm total}$. The dashed line stands for the fitted error proportional to $1 / T$, as predicted by \cref{eq:location_err}.
     }
\end{figure}

To illustrate the effectiveness of QFAMES, we consider a simple Hamiltonian with three eigenvalues: $\lambda_0 = \lambda_1 = 0$ and $\lambda_2=0.1$. There are only two distinct eigenvalues: $\lambda_0^{\star} = 0$ and $\lambda_1^{\star} = 0.1$.
The Hamiltonian is given by
\begin{align}
    H=\begin{bmatrix}
    0 & 0 & 0\\
    0 & 0 & 0\\
    0 & 0 & 0.1
\end{bmatrix}\,,
\end{align}
and we assume access to three initial states $\ket{\phi_0}, \ket{\phi_1}, \ket{\phi_2}$ with overlap matrix:
\begin{equation}\label{eq:illustrative_example_overlap}
    \Phi=\Psi=\frac{1}{\sqrt{3}}\begin{bmatrix}
        1 & 1 & 1\\
        1 & -1 & 1\\
        1 & 1 & -1
    \end{bmatrix}.
\end{equation}

In this case, the tail overlap $p_{\rm tail}=0$ and all eigenvalues are dominant. A single initial state approach would only yield an estimate of the location of $\lambda_0^{\star}$ and $\lambda_1^{\star}$ but fail to resolve the degeneracy of $\lambda_0^{\star}$. Moreover, because the overlaps with $\lambda_2$ {($\lambda^\star_1$)} are not sufficiently dominant compared to those with $\lambda_0, \lambda_1$ {($\lambda^\star_0$)}, single-state approaches might not reliably estimate the location of $\lambda_2$ {($\lambda^\star_1$)}.

To resolve the degeneracy at $\lambda_0^{\star}$, it is essential to utilize the off-diagonal entries of the data tensor, which encode important information about the spectral structure. 
We can check that the $\mc{Z}$ matrix has the form
\begin{align}
    \mc{Z}(t)= \frac{1}{3}
    \begin{bmatrix}
        2 + e^{-\i 0.1 t} & e^{-\i 0.1 t} & 2 - e^{-\i 0.1 t} \\
        e^{-\i 0.1 t} & 2 + e^{-\i 0.1 t}  & - e^{-\i 0.1 t}\\
        2 - e^{-\i 0.1 t} & - e^{-\i 0.1 t}& 2 + e^{-\i 0.1 t} 
    \end{bmatrix}\,.
\end{align}
The off-diagonal entry $\mc{Z}_{0,1}(t) = {\frac{1}{3}}e^{-\i 0.1 t}$, for example, reveals two key pieces of information:
\begin{enumerate}
  \item ${\lambda_2} = 0.1$ must be one of the dominant eigenvalues, as it is the only effective frequency present in the signal;
  \item The degeneracy of dominant eigenvalue $\lambda_0^{\star}$ must be at least two, since the signal corresponding to the frequency $\lambda \approx 0$ vanishes in this off-diagonal entry.
\end{enumerate}
The first point allows for a more accurate estimation of the location of $\lambda_2$, while the second provides insight into the multiplicity of {$\lambda_0^{\star}$}.

Now we perform numerical simulations to demonstrate the main ideas of QFAMES. 
We independently sample ${\{t_n\}^N_{n=1}}$ with $N = 2000$ from the truncated Gaussian distribution ${a_T^{\rm trunc}}(t)$ defined in~\cref{eqn:a_T} with truncation parameter $\sigma = 1$, and generate each $Z_n$ using a single-shot generalized Hadamard test. In~\cref{fig:illustrative_QFAMES_a}, we plot the magnitude of each entry of $G(\theta)$ for $T=40$. The diagonal entries exhibit a prominent peak corresponding to $\lambda_0^{\star}$, while the second peak near $\lambda_2$ is not clearly identifiable. As a result, the standard {QPE-based} approach described in~\cref{sec:single_state_limitation} cannot locate $\lambda_2$ and fails to reveal the multiplicity of $\lambda_0^{\star}$. Although the diagonal landscapes are the same in~\cref{fig:illustrative_QFAMES_a}, if we focus on the {off-diagonal} landscapes,
it is evident that there exists one dominant eigenvalue near 
$\lambda_2=0.1$.
This can be more systematically observed in the plot of the Frobenius norm {$\|G(\theta)\|_F$} in~\cref{fig:illustrative_QFAMES_b}, which clearly reveals two distinct dominant eigenvalues, including an ``amplified'' peak near $\lambda_0 = \lambda_1 = 0$. {Note that the artificial peaks in \cref{fig:illustrative_QFAMES_a} and \cref{fig:illustrative_QFAMES_b} are due to noises that can be suppressed as we increase the number of samples $N$.} Figures~\cref{fig:illustrative_QFAMES_c} and~\cref{fig:illustrative_QFAMES_d} show the singular values of $G(\theta)$ at $\theta=0$ and $\theta=0.1$, respectively. They exactly recover the multiplicities: {two states associated with the dominant eigenvalue near zero and one state associated with the dominant eigenvalue near $0.1$.}

In \cref{fig:illustrative_model}, we show the accuracy of the QFAMES algorithm in locating dominant eigenvalues as a function of {$T_{\max}$, $T_{\rm total}$}. The parameters are chosen as follows: SVD threshold $\tau = 0.1\sqrt{LR}$, block parameter $\alpha = 5$, search parameter $q = 0.005$, and $T$ takes values $40, 50, 60, 80, 100, 200, 400,$ and $800$.
The error is defined as:
\begin{equation}\label{eqn:error_calculation}
\text{Error}=\max_i\left|\theta^\star_i-\lambda^\star_i\right|\,.
\end{equation}
In~\cref{fig:illustrative_model}, we also include a comparison with the QMEGS algorithm~\cite{ding2024quantum}, using the same parameters as QFAMES except for $N = 6000$, which ensures the total number of Hamiltonian simulations is the same. From the graph, we observe that QFAMES achieves significantly lower error than QMEGS when $T_{\max} < 100$, while the errors of both methods become comparable for larger $T_{\max}$. This demonstrates the advantage of QFAMES in identifying dominant eigenvalues when $T_{\max}$ is not sufficiently large, aided by the use of off-diagonal entries in $G(\theta)$. Furthermore, we note that QMEGS cannot identify the multiplicity of $\lambda_0^{\star}$, as it relies solely on the diagonal entries of $G(\theta)$, whereas QFAMES accurately estimates the multiplicities of both dominant eigenvalues.

\subsection{Transverse Field Ising Model (TFIM)}\label{sec:TFIM}

Here we consider the transverse field Ising model on a spin chain of length {$N_{\rm spin}$}: \begin{equation}\label{eq:TFIM-8}
H = -\sum_{i=1}^{N_{\rm spin}-1}Z_iZ_{i+1} - g\sum_{i=1}^{N_{\rm spin}}X_i\,,
\end{equation}
where $Z_i$ and $X_i$ are the Pauli operators acting on the $i$-th qubit and $g$ is the coupling coefficient. The ground state of the TFIM lies in the ferromagnetic phase for $g<1$, and undergoes a quantum phase transition at $g=1$ to the paramagnetic phase for $g>1$. Note that {in the thermodynamic limit,} the GSD is 2 in the ferromagnetic phase, 1 in the paramagnetic phase, and the system is gapless at the phase transition point.

\begin{figure}
    \centering
    \includegraphics[width=0.51\linewidth]{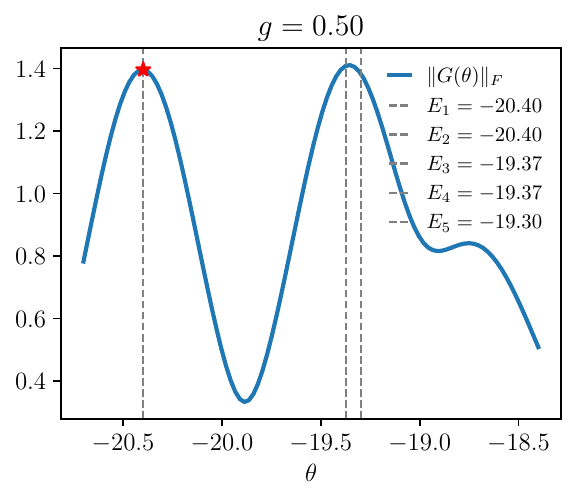}
    \includegraphics[width=0.475\linewidth]{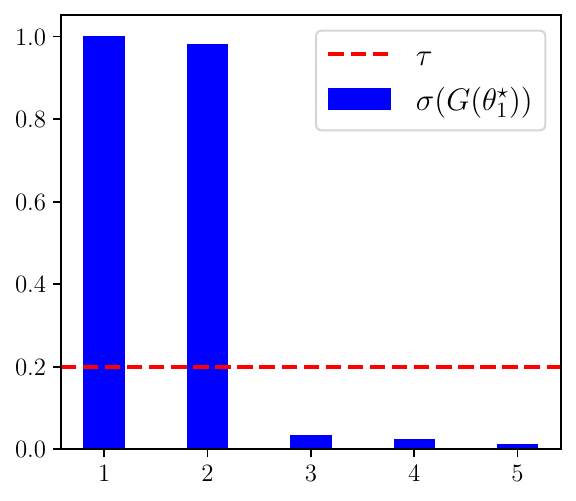}
    
    \includegraphics[width=0.51\linewidth]{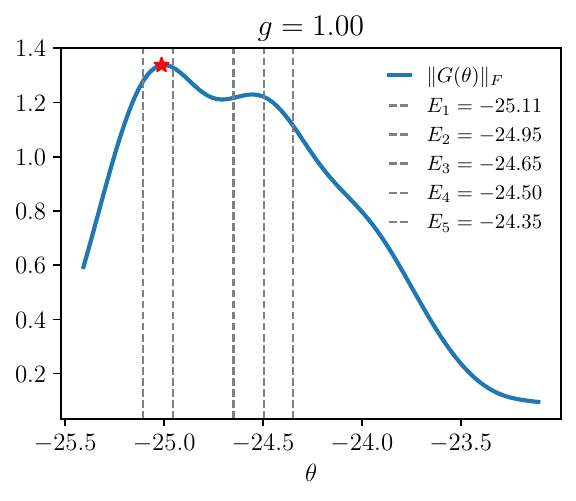}
    \includegraphics[width=0.475\linewidth]{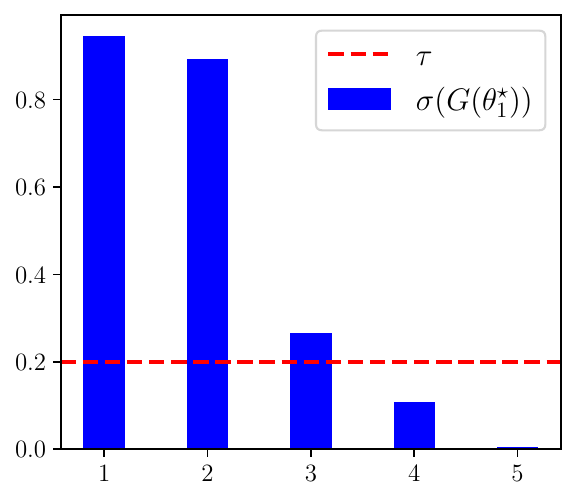}
    
    \includegraphics[width=0.51\linewidth]{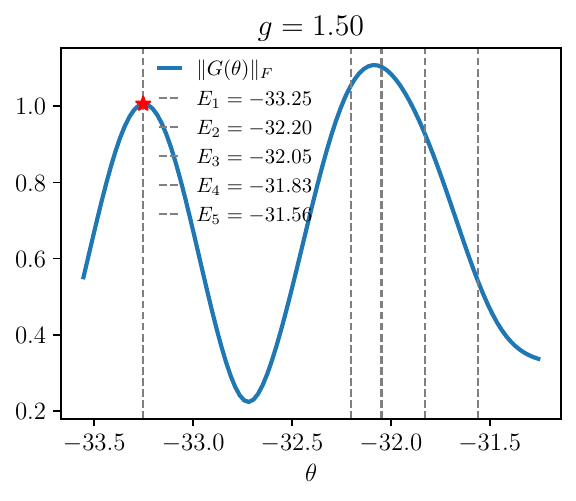}
    \includegraphics[width=0.475\linewidth]{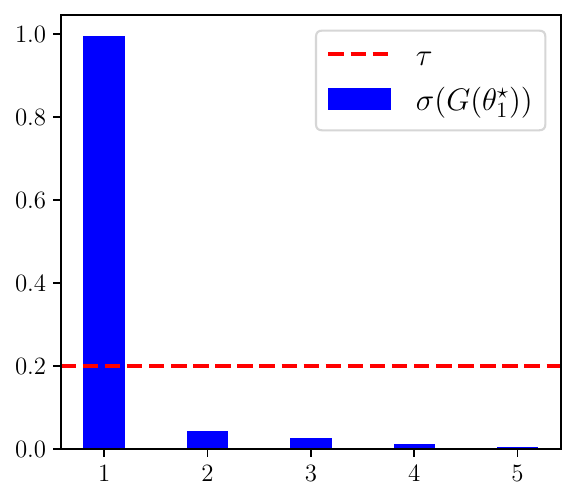}
    \caption{Illustration of the QFAMES algorithm on the TFIM model. The left column shows the Frobenius norm of $G(\theta)$ around the ground state energy for $g=0.5$, $1.0$ and $1.5$, from up to down. The {gray vertical} dashed lines stand for the 5 lowest lying energy eigenvalues. The right column shows the singular values of the $G(\theta_1^{\star})$, where $\theta_1^{\star}$ is the leftmost peak found in corresponding Frobenius norm (marked with red stars). The SVD cutoff threshold is set to be 0.2. 
    }
    \label{fig:TFIM}
\end{figure}
\begin{figure}
    \centering
    \includegraphics[width=0.7\linewidth]{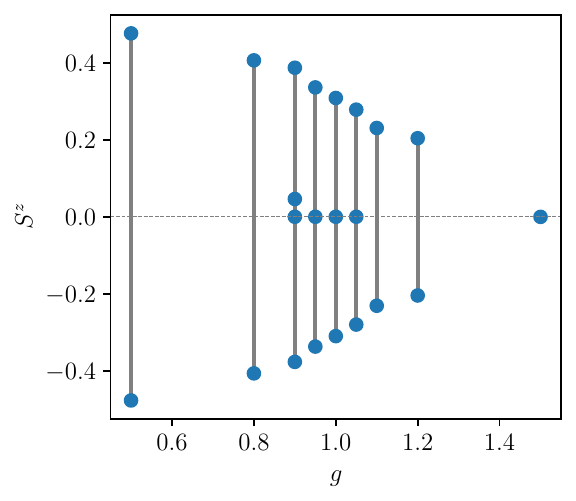}
    \caption{Illustration of \cref{alg:QFAMES_obs} on the TFIM model. The observable is $S^z = \frac{1}{2N_{\rm spin}} \sum_{i = 1}^{N_{\rm spin}} Z_i$. In the manifold of near-ground-state cluster, the possible ranges of observable expectation values are shown by the gray lines. A transition from ferromagnetic to paramagnetic phase can be observed when $g$ increases from 0.5 to 1.5.}
    \label{fig:TFIM_obs}
\end{figure}

In our test, we set system size $N_{\rm spin}=20$ and $g = 0.5, 0.8, 0.9, 0.95, 1, 1.05, 1.1, 1.2, 1.5$. For each $g$, the set of initial states is chosen to be 5 matrix product states (MPS) of bond dimension 2 that are orthogonal to each other and have approximately the lowest energies. {We first run DMRG with bond dimension 2 to find the MPS $\ket{\psi_0}$. The states $\ket{\psi_i}$ ($i > 0$) are then sequentially obtained by running DMRG with the projected Hamiltonian
\[
    H_i = \left( 1 - \sum_{j=0}^{i-1} \ket{\psi_j}\bra{\psi_j} \right) H \left( 1 - \sum_{j=0}^{i-1} \ket{\psi_j}\bra{\psi_j} \right).
\]
}
We simulate the time evolutions with the time-evolving block decimation (TEBD) algorithm~\cite{Vidal2004} with maximum MPS bond dimension $\chi=100$.  {The Gaussian filter parameter $T / \sqrt{2} = 2.5$. For simplicity, for the time parameters $t_n$ and $t'_n$ in \cref{eq:obs_measure}, we choose $(t_n, t'_n) \in \mc{T}_{\rm sample} \times \mc{T}_{\rm sample}$, where $\mc{T}_{\rm sample}$ consist of $N=100$ samples drawn from Gaussian density distribution $a_{T/\sqrt{2}}(t)$. This choice will effectively generate $N^2 = 10^4$ samples to compress measurement error.
}
In \cref{fig:TFIM} we demonstrate how QFAMES can be applied to resolve the {near-ground-state multiplicity} of TFIM. When $g = 0.5$, we can clearly see {a near-degenerate pair}. When $g=1$, we are in the gapless phase, and more states around the ground state are detected with a moderate choice of $T$. While when $g=1.5$, the ground state is found to be unique. 

We can further obtain the ground state properties of each phase by applying \cref{alg:QFAMES_obs}. The observable we choose is
\begin{align}
    S^z = \frac{1}{2N_{\rm spin}} \sum_{i = 1}^{N_{\rm spin}} Z_i
\end{align}
and the results are illustrated in \cref{fig:TFIM_obs}.
When $g=0.5$, the solutions to the generalized eigenvalue problem, which serve as an estimation of eigenvalues of $\left(\braket{E_k | S^z | E_{k'}}\right)_{k,k'\in[m_i]}$, are close to $\pm 0.5$, showing the {near-ground-state cluster} lies in the ferromagnetic phase. As $g$ increases, the possible range of expectation values of $S^z$ in the near-ground-state cluster manifold reduces and finally reaches 0 at $g=1.5$, corresponding to the paramagnetic phase. {Note that the blue dots, which represent the eigenvalues of $\left(\braket{E_k | S^z | E_{k'}}\right)_{k,k'\in[m_i]}$, do not have to be symmetric about 0 (for example at $g = 0.9$), since they are extracted from the near-ground-state cluster identified at finite $T$ rather than from an exact symmetry-resolved eigenspace.   }

\subsection{2D Toric Code}~\label{sec:toric_code}

In many-body systems, GSD is a hallmark of topological order and depends only on the global topology of the system. For stabilizer Hamiltonians, the GSD is exactly $2^k$, where $k$ is the number of logical qubits encoded by the code. 
In this section, we estimate the ground-state degeneracy of the 2D Toric code Hamiltonians~\cite{Kitaev2003} as an example. The Toric code ground state is known to exhibit $\mathbb{Z}_2$ topological order, which has been realized on a quantum computer based on the prior knowledge of its explicit form~\cite{Satzinger2021}. In contrast, our numerical experiment starts with a collection of random initial states. 
Specifically, the Toric code Hamiltonian on a 2D lattice is given by
\begin{equation}
    H = - \sum_{v} A_v - \sum_{p} B_p,
\end{equation}
where $v$ denotes the vertices and $p$ the plaquettes of the lattice. Qubits live on the edges, with $i\in v$ labeling the edges touching a vertex $v$, and $i\in p$ labeling the edges surrounding a plaquette $p$. The corresponding stabilizer operators are  $A_v = \prod_{i\in v} X_i$ and $B_p = \prod_{i \in p} Z_i$.
We study the Toric code model on a 2-by-4 lattice with torus and cylinder boundary conditions, as demonstrated in \cref{fig:toric_code}. Their GSDs are \emph{four} and \emph{two}, respectively.
The numerical experiment proceeds as follows:
\begin{enumerate}
    \item We randomly sample vectors from the Haar distribution and apply imaginary-time evolution $\exp(-\beta H)$ to enhance their overlap with the ground-state subspace.
    \item We then apply QFAMES to these boosted vectors to estimate the multiplicity of the cluster corresponding to the smallest eigenvalue.
\end{enumerate}

\begin{figure}
    \centering
    \subfloat[\label{geometry:a} torus boundary condition]{%
      \includegraphics[width=.7\linewidth]{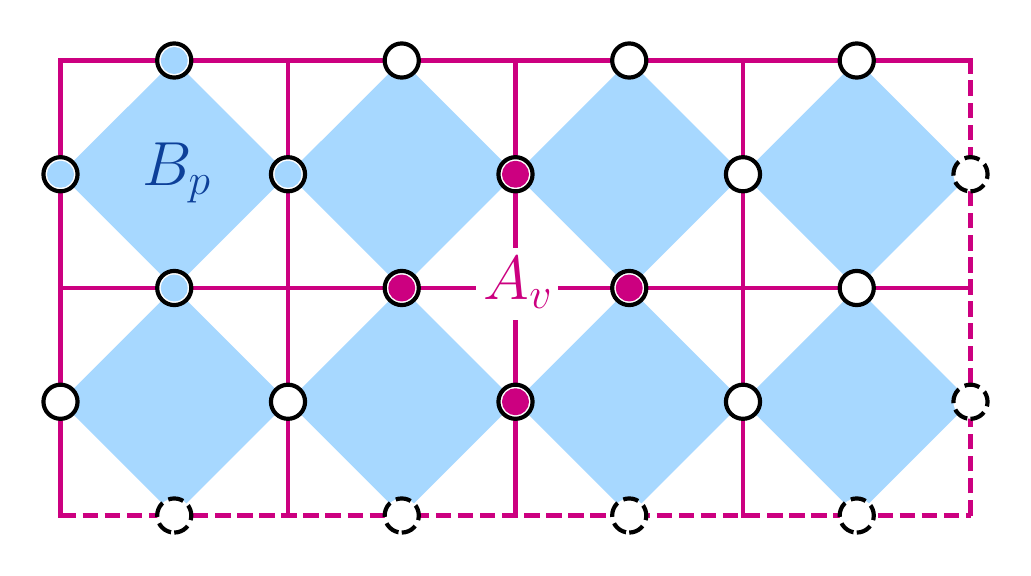}
    }
    
    \subfloat[\label{geometry:b} cylinder boundary condition]{%
      \includegraphics[width=.7\linewidth]{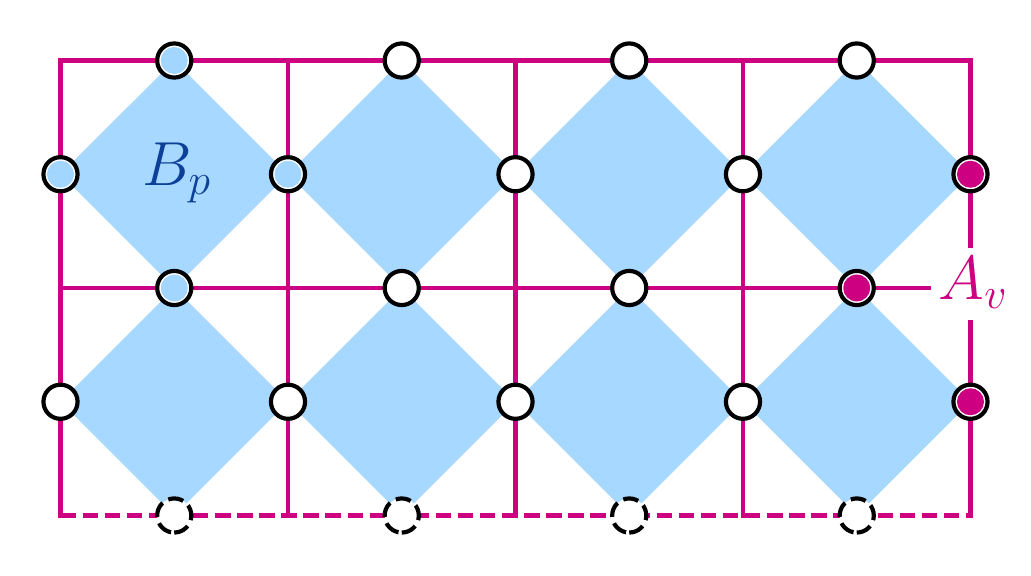}
    }
    \caption{
    2D Toric codes on a 2-by-4 lattice with \textbf{(a)} torus, and \textbf{(b)} cylinder boundary conditions. Dashed lines stand for periodic boundary conditions in the corresponding directions. The models consist of 16 and 18 qubits, respectively, represented by the small circles.}
    \label{fig:toric_code}
\end{figure}

We find that in the first step of the algorithm, it is necessary to generate a sufficient number of initial states such that the ground-state subspace is contained within their linear span. This condition is essential for correctly estimating the ground-state multiplicity. We also remark that imaginary-time evolution is generally not efficient for large-scale, general Hamiltonians on a quantum computer, and is used here primarily for illustrative purposes. In practical applications, one would need to employ more quantum-friendly methods for generating random low-energy states.
\begin{figure}
    \centering
    \subfloat[\label{fig:toric_2_4_1_a} torus boundary condition, $\beta=10, L = 15$]{%
      \includegraphics[width=0.48\linewidth]{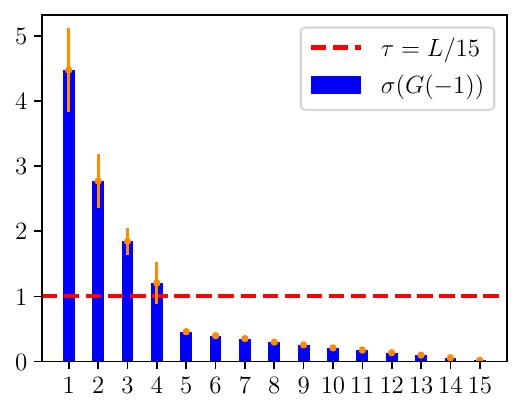}
    }\hfill
    \subfloat[\label{fig:toric_2_4_0_a} cylinder boundary condition, $\beta=10, L = 15$]{%
      \includegraphics[width=0.48\linewidth]{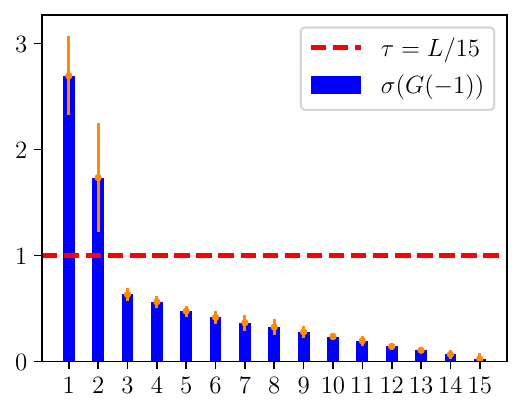}
    }

    \subfloat[\label{fig:toric_2_4_1_b} torus boundary condition, $\beta=10, L=25$]{%
      \includegraphics[width=0.48\linewidth]{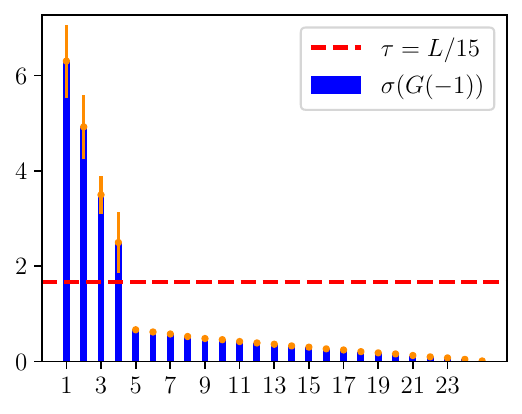}
    }\hfill
    \subfloat[\label{fig:toric_2_4_0_b} cylinder boundary condition, $\beta=10, L=25$]{%
      \includegraphics[width=0.48\linewidth]{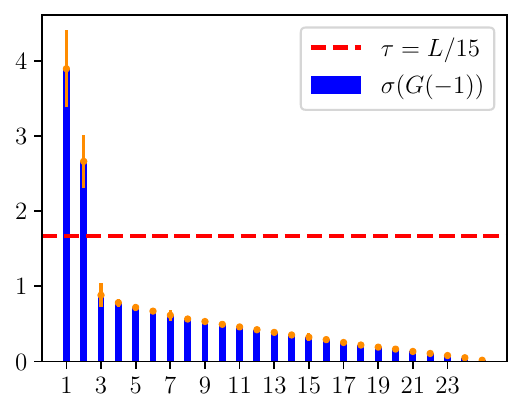}
    }

    \subfloat[\label{fig:toric_2_4_1_c} torus boundary condition, $\beta=15, L=15$]{%
      \includegraphics[width=0.48\linewidth]{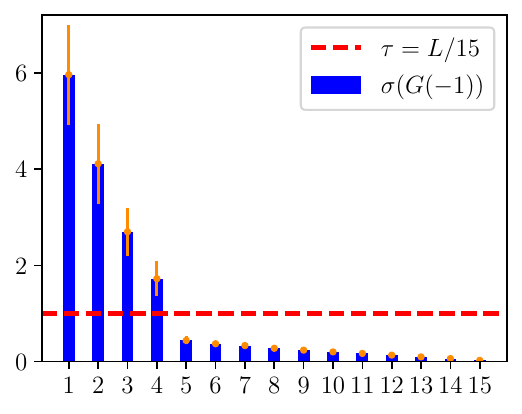}
    }\hfill
    \subfloat[\label{fig:toric_2_4_0_c} cylinder boundary condition, $\beta=15, L = 15$]{%
      \includegraphics[width=0.48\linewidth]{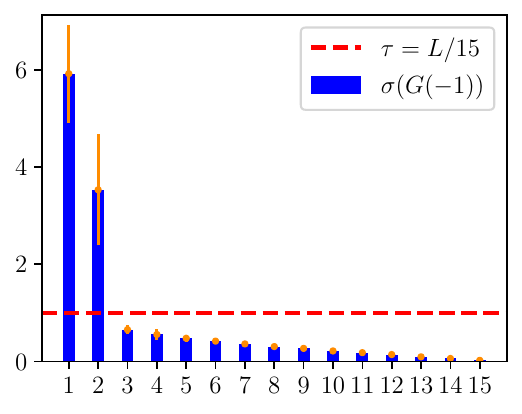} 
    }
    \caption{Illustration of GSD estimation for 2D Toric code models with \textbf{(left)} torus and \textbf{(right)} cylinder boundary conditions. The spectra are rescaled to $[-1,1]$ and we plot the averaged singular values of $G(-1)$. The orange lines are the error bars representing the standard deviation over 10 trials for each parameter configuration. Four and two large singular values can be clearly identified, respectively. In the second row, we increase the number of initial states $L$ compared with the first row. In the third row, we increase the imaginary evolution time $\beta$.}
    \label{fig:toric_2_4_result}
\end{figure}

We use $N=300$ samples, and set circuit depth parameter $T = 10$, the Gaussian truncation parameter $\sigma=1$, and SVD threshold $\tau = \sqrt{LR} / 15 = L / 15$. For inverse temperature $\beta$ and the number of states $L$, we consider different configurations: ($\beta=10$, $L=15$) and ($\beta=10$, $L=25$), and ($\beta=15$, $L=15$). Each setup is repeated 10 times, and the averaged singular values are output. Results are shown in \cref{fig:toric_2_4_result}. Four and two large singular values can respectively be identified for the two boundary conditions. The performance of the algorithm is also shown to be enhanced by 
(1) increasing the number of initial states $L$ (the second row in \cref{fig:toric_2_4_result} compared with the first one), which strengthens the uniform overlap condition,
(2) or increasing the imaginary evolution time $\beta$ (the third row compared with the first one), which reinforces the sufficient dominance condition in \cref{eq:def_domination}.

{
\subsection{Anisotropic Heisenberg (XXZ) model}\label{sec:xxz_model}
}

\begin{figure*}
    \centering
    \includegraphics[width=0.19\linewidth]{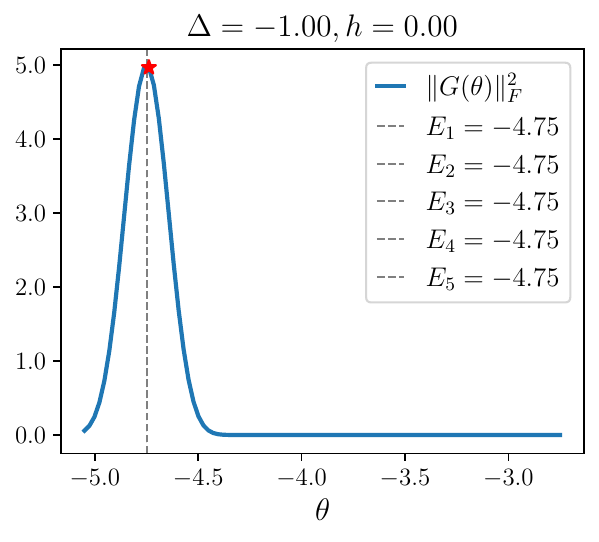}
    \includegraphics[width=0.19\linewidth]{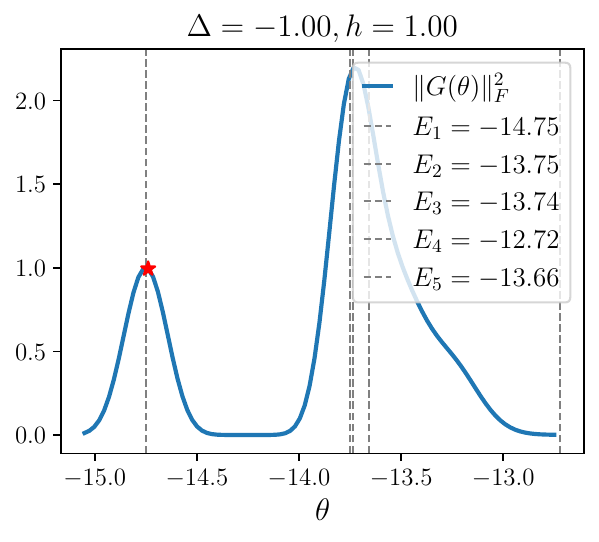}
    \includegraphics[width=0.19\linewidth]{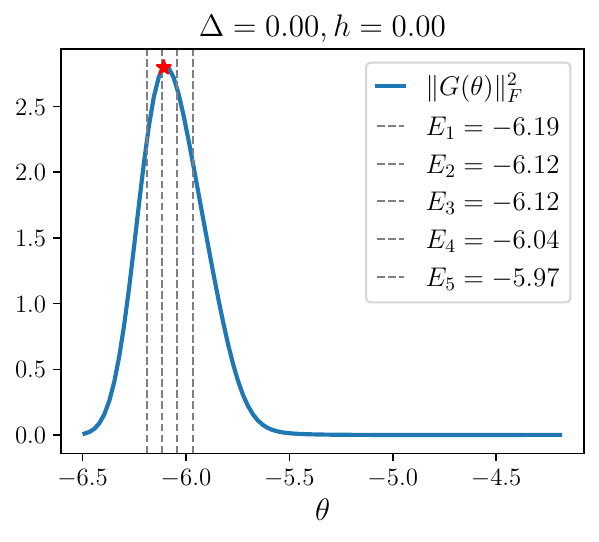}
    \includegraphics[width=0.19\linewidth]{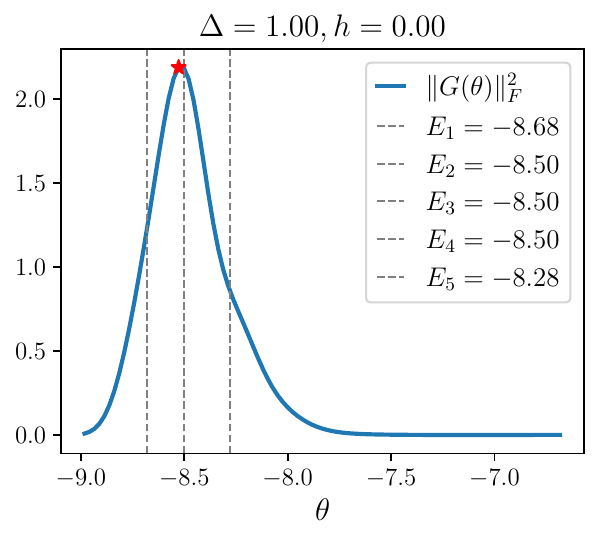}
    \includegraphics[width=0.19\linewidth]{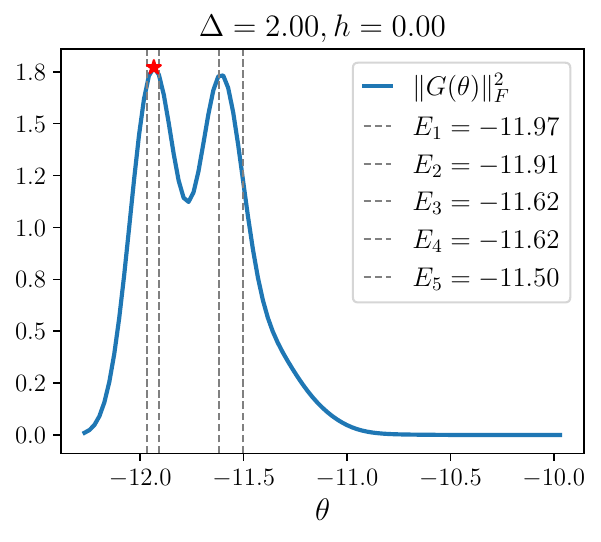}

    \includegraphics[width=0.19\linewidth]{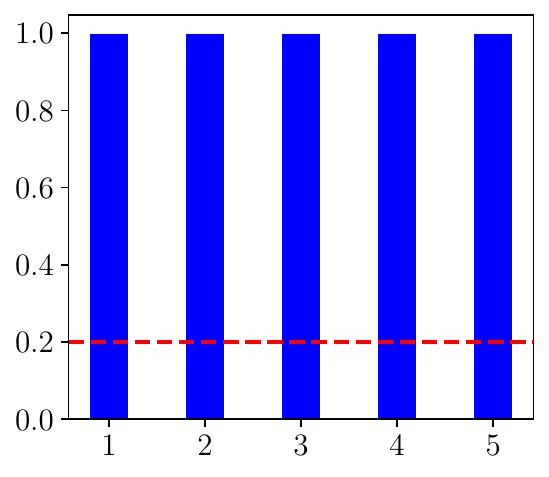}
    \includegraphics[width=0.19\linewidth]{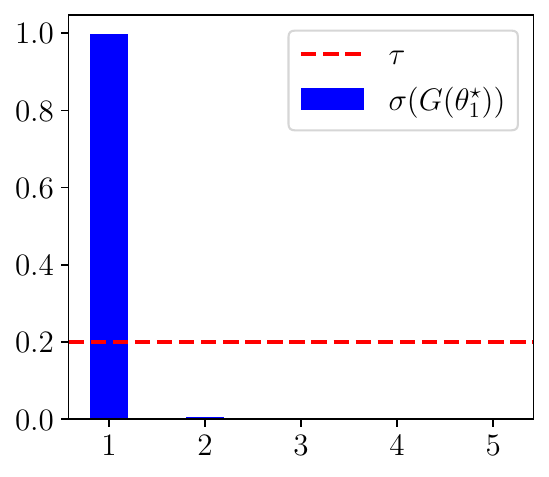}
    \includegraphics[width=0.19\linewidth]{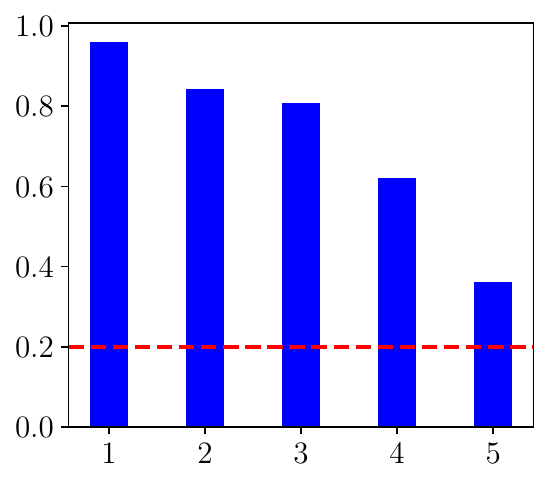}
    \includegraphics[width=0.19\linewidth]{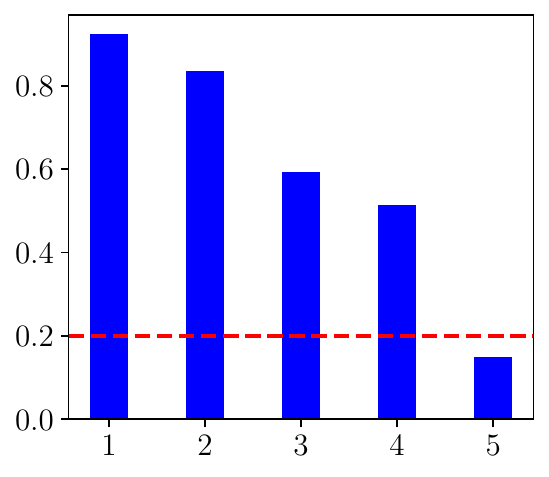}
    \includegraphics[width=0.19\linewidth]{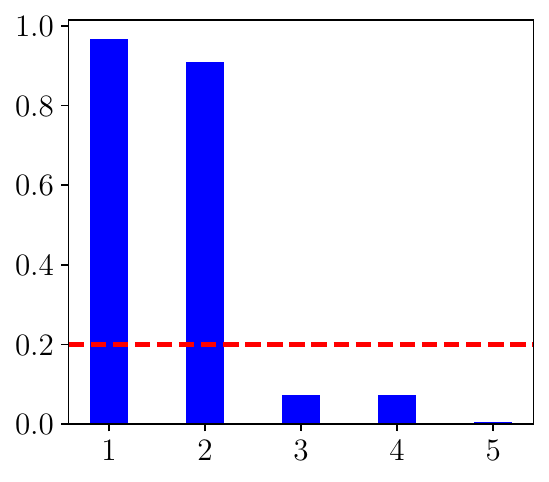}

    \subfloat[\label{XXZ:a} FM, isotropic]{
        \includegraphics[width=0.184\linewidth]{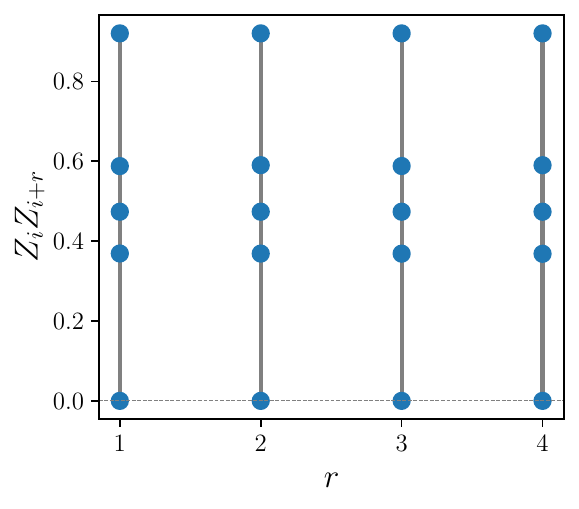}
    }
    \subfloat[\label{XXZ:b} FM, polarized]{
        \includegraphics[width=0.184\linewidth]{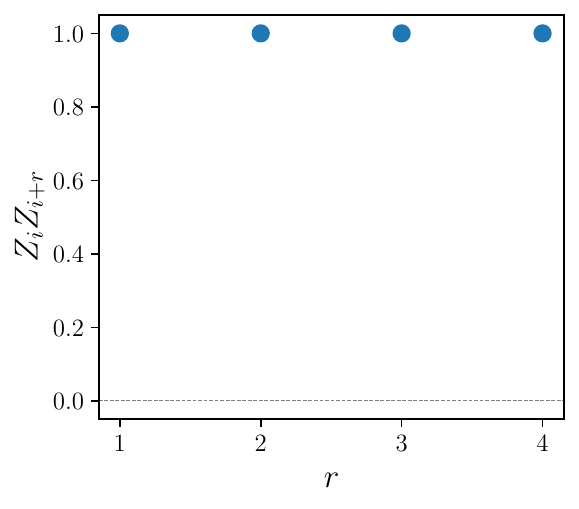}
    }
    \subfloat[\label{XXZ:c} Gapless, XY]{
        \includegraphics[width=0.184\linewidth]{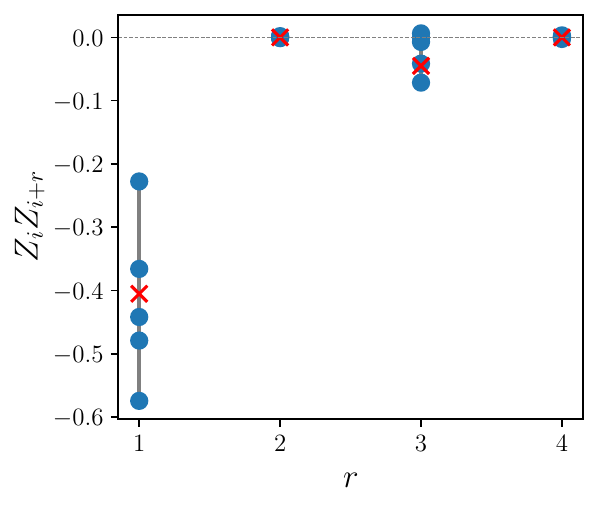}
    }
    \subfloat[\label{XXZ:d} Gapless, Heisenberg]{
        \includegraphics[width=0.184\linewidth]{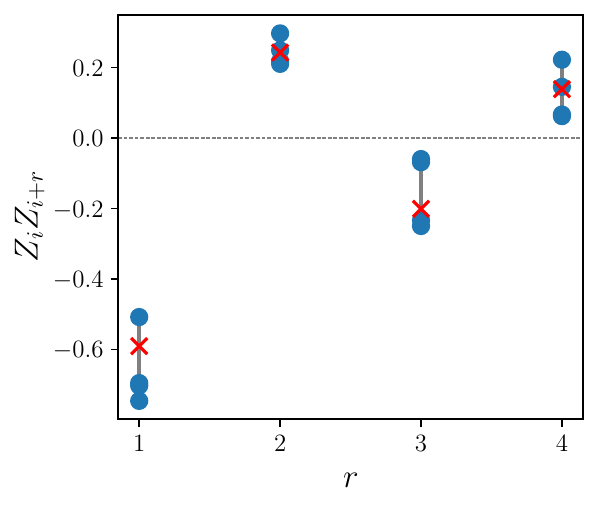}
    }
    \subfloat[\label{XXZ:e} AFM]{
        \includegraphics[width=0.184\linewidth]{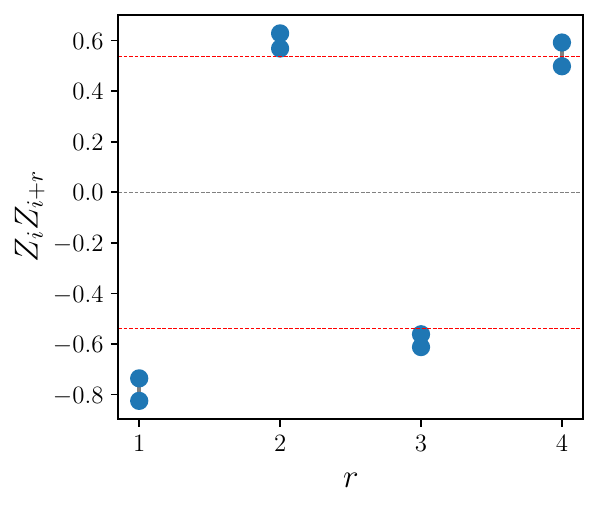}
    }
    
    \caption{Illustration of the QFAMES algorithm on the XXZ model. The top row shows the Frobenius norm of $G(\theta)$ around the ground state energy for $(\Delta,h)=(-1,0),(-1,1),(0,0),(1,0)$ and $(2,0)$, from left to right. In the top row, the gray vertical dashed lines indicate the five lowest-lying energy eigenvalues. The middle row shows the singular values of $G(\theta_1^{\star})$, where $\theta_1^{\star}$ is the leftmost peak found in the corresponding Frobenius norm plot (marked with red stars). The SVD cutoff threshold is set to be 0.2. The bottom row shows the estimated ranges of the correlation function $\braket{{Z_i} {Z_{i+r}}}$ in the center of the chain obtained with \cref{alg:QFAMES_obs}, where $1\le r \le 4$. The red crosses in (c) and (d) represent exact values for ground state correlation functions, and the red dashed lines in (e) indicate the {asymptotic staggered correlation} $\pm\lim_{r\to\infty} \lim_{N_{\rm spin}\to\infty}\left|\braket{{Z_i} {Z_{i+r}}}\right|$.
    }
    \label{fig:XXZ}
\end{figure*}

We also consider the spin-$1/2$ anisotropic Heisenberg (XXZ) chain in the presence of an external magnetic field, whose Hamiltonian is 
\begin{equation}
  H = \sum_{i=1}^{N_{\rm spin}-1} \left( X_{i}X_{i+1} + Y_{i}Y_{i+1} + \Delta Z_{i}Z_{i+1} \right) + h \sum_{i=1}^{N_{\rm spin}}Z_i.
\end{equation}
The phase diagram of this model has been extensively studied (see, for example, \cite{Langari1998, Kitanine2002, Rakov2016}). Its distinct phases can be characterized by the correlation functions, $\braket{{Z_i} {Z_{i+r}}}$, evaluated at the center of the chain.

Similar to TFIM, we set the system size $N_{\rm spin}=20$ and choose 5 initial states to be orthogonal, low-energy MPS with bond dimension $D = 3$, which is sufficient to obtain large overlaps with the ground state and first several excited states, and $T / \sqrt{2} = 5$. The rest of the parameters are the same as the TFIM experiment. We perform QFAMES with the following sets of parameters $\left( \Delta, h \right)$:
\begin{itemize}
    \item $\left( \Delta, h \right) = (-1, 0)$ (\cref{XXZ:a}). It is unitarily equivalent to the isotropic ferromagnetic (FM) Heisenberg model, which has $N_{\rm spin}+1$ degenerate ferromagnetic ground states due to $SU(2)$ symmetry. 
    Five of them are detected here, limited by the number of initial states provided. The values of correlators in $z$ direction vary depending on the specific ground state orientation, whereas they all exhibit ferromagnetic behaviors.
    \item $\left( \Delta, h \right) = (-1, 1)$ (\cref{XXZ:b}). The ground state degeneracy is lifted by the magnetic field, such that the system has a unique ground state {$\ket{\downarrow \cdots \downarrow}$} and $\braket{{Z_i} {Z_{i+r}}} = 1$ for all $r$.
    \item $\left( \Delta, h \right) = (0, 0)$ (\cref{XXZ:c}) and $(1, 0)$ (\cref{XXZ:d}). They correspond to the XY and Heisenberg models, respectively. Both of them are in a gapless, antiferromagnetic phase, where the correlation functions are quasi-long range that decay algebraically in $r$.  
    With our choice of $T$, we detect a couple of energy eigenstates around the first peak of $\left\| G(\theta) \right\|_F$, as expected for a gapless spectrum. Red crosses indicate exact values for ground state correlation functions $\braket{{Z_i} {Z_{i+r}}}$ of both the XY model~\cite{Lieb1961} and isotropic Heisenberg model~\cite{Sakai2003,Boos2005}, which lie within the ranges obtained with \cref{alg:QFAMES_obs} (solid gray lines).
    \item $\left( \Delta, h \right) = (2, 0)$ (\cref{XXZ:e}). The system is in an antiferromagnetic (AFM) gapped phase. When the system size $N_{\rm spin}$ is even, the two lowest-energy states form a quasi-degenerate pair~\cite{Grijalva2019}, with a splitting that decreases rapidly as $N_{\rm spin}$ increases. This quasi-degeneracy can be observed with QFAMES. The ground states exhibit long-range antiferromagnetic order~\cite{Vaxter1973}, and the corresponding {asymptotic staggered correlation} $\pm\lim_{r\to\infty}  \lim_{N_{\rm spin}\to\infty}\left|\braket{{Z_i} {Z_{i+r}}}\right|$ is indicated by red dashed lines.
\end{itemize}

\section{Discussion}\label{sec:discussion}

We presented QFAMES, a quantum spectral filtering framework that estimates dominant eigenvalues and their multiplicities by jointly processing cross-correlations from multiple initial states. By exploiting off-diagonal information through generalized Hadamard tests, QFAMES recovers spectral structures that are inaccessible to {phase estimation from a single initial state} and extends naturally to observable estimation within identified clusters. The algorithm admits rigorous efficiency guarantees.

The quantum data produced by QFAMES naturally form an order-$3$ tensor. One could attempt to solve the \textsc{dods} problem using general tensor-decomposition techniques \cite{kolda2009tensor,harshman1970foundations,carroll1970analysis,Leu_93,bro1997parafac,sharan2017orthogonalized} or approaches tailored to high-order time series \cite{roemer2014analytical,hangleiter2024robustly}. However, to the best of our knowledge, none of these algorithms provide provable efficiency guarantees in the high-noise or overcomplete regime, and it remains unclear how to extend the existing theoretical results to this setting. In contrast, our analysis gives a simple estimator that is tailored to the QFAMES tensor structure, yields provable efficiency and robustness to measurement noise, and provides exact multiplicity recovery under the stated assumptions. We expect these ideas to be useful beyond the present setting.

The QFAMES algorithm has broad potential applications. For instance, it can be employed to probe the ground-state manifolds of models with topological orders that are difficult to analyze with classical algorithms, such as those exhibiting the fractional quantum Hall effect~\cite{Stormer1999,Regnault2011}. In addition, {although the density of states (DOS) at finite temperatures typically scales exponentially with system size, QFAMES remains of physical interest for specialized applications, such as the study of quantum many-body scars~\cite{Turner2018a,Lin2019Scar,Chandran2023}. These scars represent a set of rare energy eigenstates in non-integrable systems that avoid thermalization. Unlike typical eigenstates at finite energy densities, which exhibit volume-law entanglement, scar states usually possess sub-volume-law entanglement. Consequently, they can be well captured by tractable ansatze such as the low-bond-dimension MPS demonstrated in this work. This enables QFAMES to be applied to characterize their properties.}

{A particularly promising direction, enabled by the mixed-state extension in \cref{sec:QFAMES_mixed}, is to interface QFAMES with state-preparation procedures that are intrinsically dissipative and therefore naturally output mixed states and/or do not admit a controlled, reversible state-preparation oracle. This setting arises in several recent proposals for dissipative Gibbs-state and ground-state preparation. In this regime, our method based on swap-test-based estimators of mixed-state cross-correlators (and their observable-decorated counterparts) provides a practical route to estimating dominant eigenvalues and their multiplicities, as well as extremal observable values within a (nearly) degenerate manifold. Importantly, pairing families of mixed states through cross-correlations can reveal information that is not accessible by measuring the observable on each individual mixed state alone (cf. the toy example in~\cref{sec:QFAMES_mixed}).}

{This perspective is especially relevant for metastability in open-system dynamics. In many physically interesting models, dissipative dynamics can exhibit long-lived metastable manifolds induced by energetic or entropic bottlenecks, with polynomial-time relaxation into a metastable subspace but potentially exponentially slow convergence to the true global equilibrium~\cite{bergamaschi2025,Chen2025Local,gamarnik2024,rakovszky2024,Yin2025,Bovier2006,BinderYoung1986,DebenedettiStillinger2001,DiehlZoller2008,Turner2018a,bergamaschi2025_2}. Mixed-state QFAMES offers a diagnostic tool in this context: by filtering around a target low-energy cluster and analyzing the resulting effective rank and generalized-eigenvalue extrema, one can characterize the dimension of the metastable subspace and bound physically meaningful observables within it, providing quantitative probes of metastable phases and long-lived prethermal or glassy regimes that are of broad interest in condensed matter physics and quantum materials science~\cite{Yin2025,Bovier2006,BinderYoung1986,DebenedettiStillinger2001,DiehlZoller2008,Turner2018a}.}

{
These examples, together with the availability of control-free, ancilla-free implementations (\cref{sec:ancilla_free}), highlight the versatility of QFAMES as a general and efficient tool for uncovering spectral properties across diverse quantum systems, in particular on early fault-tolerant quantum devices.
}

\begin{acknowledgments}
This work was supported in part by the U.S. Department of Energy, Office of Science, National Quantum Information Science Research Centers, Quantum Systems Accelerator (Z.D., L.L.), and by the U.S. Department of Energy, Office of Science, Accelerated Research in Quantum Computing Centers, Quantum Utility through Advanced Computational Quantum Algorithms, grant no. DE-SC0025572 (Y.Y, L.L.).
L.L. is a Simons Investigator in Mathematics.  We thank Wendy Billings, Dominik Hangleiter, Martin Head-Gordon and Birgitta Whaley for helpful discussions. 
\end{acknowledgments}

\appendix
\begin{widetext}
    \setlength\nomlabelwidth{14em}
\nomenclature[01]{$[n]$}{The index set $\{ 0, 1, \cdots, n-1\}$}
\nomenclature[B,01]{$H=\sum_{m\in [M]}\lambda_m \ket{E_m}\bra{E_m}$}{Hamiltonian of interest and its spectral decomposition}
\nomenclature[B,02]{$\mc{D}\subset [M]$}{Index set of dominant eigenvalues}
\nomenclature[B,03]{$I\in \mathbb{N}_+$}{The number of distinct dominant eigenvalues}
\nomenclature[B,04]{$\{\lambda_i^{\star}\}\subset \mathbb{R}$}{The distinct dominant eigenvalues}
\nomenclature[B,05]{$\mc{D}_i\subset \mc{D}$}{The index set of dominant eigenvectors with eigenvalue $\lambda_i^{\star}$}
\nomenclature[B,07]{$\Delta\in \mathbb{R}_+$}{The minimum separation between distinct dominant eigenvalues}
\nomenclature[C,01]{$L, R\in \mathbb{N}_+$}{Number of initial states}
\nomenclature[C,02]{$\{U_l\}_{l\in [L]},\ \{V_r\}_{r\in [R]}$}{Initial state preparation unitaries}
\nomenclature[C,03]{$\{\ket{\phi_l}\}_{l\in [L]},\ \{\ket{\psi_r}\}_{r\in [R]}$}{Initial states prepared by $U_l$ and $V_r$}
\nomenclature[C,04]{$\Phi\in \C^{L\times M}$}{Overlaps between $\{\ket{\phi_l}\}$ and $\{\ket{E_m}\}$}
\nomenclature[C,05]{$\Psi\in \C^{R\times M}$}{Overlaps between $\{\ket{\psi_r}\}$ and $\{\ket{E_m}\}$}
\nomenclature[C,06]{$p_{\rm tail}\in \mathbb{R}$}{Total overlaps with the non-dominant eigenvectors}
\nomenclature[C,07]{$p_{\min}\in \mathbb{R}$}{Minimum overlap with a dominant eigenvector}
\nomenclature[C,08]{$\chi\in \mathbb{R}_+$}{Uniform overlap condition parameter}
\nomenclature[D,01]{$N\in \mathbb{N}_+$}{Number of evolution times}
\nomenclature[D,02]{$\{t_n\}_{n\in [N]}\subset \mathbb{R}$}{Evolution time samples}
\nomenclature[D,03]{$a_T^{\rm trunc}(t)$}{Probability density function (PDF) of the truncated Gaussian distribution}
\nomenclature[D,04]{$F(x)$}{Fourier transform of $a(t)$;  $F(x)=\int_{-\infty}^{\infty} e^{\i x t} a(t)\,\d t$}
\nomenclature[D,05]{$T\in \mathbb{R}$}{Circuit depth parameter}
\nomenclature[D,06]{$\sigma\in \mathbb{R}$}{Gaussian truncation parameter such that $a_T(t)$ is supported over $[-\sigma T,\sigma T]$}
\nomenclature[E,01]{$Z\in \C^{L\times R\times N}$}{Data tensor collected from the measurement outcomes for \textsc{dods} estimation}
\nomenclature[E,02]{$Z_n\in \C^{L\times R}$}{A slice of $Z$ defined as $Z_{:,:,n}$}
\nomenclature[E,03]{$\mc{Z}_{l,r}(t)\in \C$}{The expectation $\langle \phi_l|e^{-\i H t} |\psi_r\rangle$}
\nomenclature[E,04]{$G(\theta)\in \C^{L\times R}$}{The contraction of $Z$ along the third coordinate with the vector $(e^{\i\theta t_n})_{n\in [N]}$}
\nomenclature[E,05]{$Z^O\in \C^{L\times R\times N}$}{Data tensor collected from the measurement outcomes for observable estimation}
\nomenclature[E,06]{$G^O(\theta)\in \C^{L\times R}$}{The contraction of $Z^O$ along the third coordinate}
\nomenclature[F,01]{$q\in \mathbb{R}$}{The grid-search parameter such that the grid width is $q/T$}
\nomenclature[F,02]{$J\in\mathbb{N}_+$}{The number of grid points}
\nomenclature[F,03]{$\{\theta_j\}\subset \mathbb{R}$}{The grid points}
\nomenclature[F,04]{$\mc{W}_j$}{$\|G(\theta_j)\|_F$ evaluated at each grid point}
\nomenclature[F,05]{$\alpha\in \mathbb{R}$}{Width of the block}
\nomenclature[F,06]{$\{\mc{B}_i\}$}{Blocked grid points around each peak}
\nomenclature[F,07]{$\eta\in [0,1]$}{The {failure} probability of the algorithm}
\nomenclature[F,08]{$\wt{I}\in \mathbb{N}_+$}{Estimated number of distinct dominant eigenvalues}
\nomenclature[F,09]{$\{\theta_j^\star\}_{{j}\in [\wt{I}]}\subset [-\pi, \pi]$}{Estimated dominant eigenvalues}
\nomenclature[F,10]{$\tau\in \mathbb{R}$}{The threshold for the singular values}
\nomenclature[F,11]{$\{m_i\}\subset \mathbb{N}_+$}{Estimated multiplicities}
\pagestyle{myheadings}
\markboth{}{}
\printnomenclature
\pagestyle{myheadings}
\markboth{}{}
    \section{Quantum phase estimation and limitation}\label{sec:single_state_limitation}
The standard quantum phase estimation (QPE) algorithm estimates eigenvalues of a Hamiltonian $H$ by applying the quantum Fourier transform (QFT) to an initial state. Recently, a class of ``post-Kitaev'' phase estimation algorithms (see e.g.,~\cite{Somma2019,OBrien2019,DingLin2023,Ding2023simultaneous,ding2024quantum,yi2024quantum,castaldo2025heisenberg}) has been proposed for early fault-tolerant quantum computers. These algorithms build upon the original single-ancilla Kitaev algorithm~\cite{KitaevShenVyalyi2002}, which utilizes the output from the Hadamard test circuit. However, they differ in their strategies for selecting Hamiltonian simulation times and in their use of advanced classical signal processing techniques to handle {noisy} quantum data.

Unlike QFT-based algorithms, post-Kitaev methods require only a single ancilla qubit while still achieving Heisenberg-limited scaling. Moreover, they are capable of simultaneously estimating \emph{multiple} eigenvalues, making them particularly advantageous for early fault-tolerant quantum devices. We first briefly review the structure of post-Kitaev phase estimation algorithms, which can be organized within the following three-step framework:
\begin{enumerate}
    \item Generate a sequence of real numbers $\{t_n\}_{n\in [N]}$ as the Hamiltonian evolution times.
    \item Apply Hadamard test circuit (\cref{fig:hadamard}) to the initial state $\ket{\phi}$ for each time $t_n$, and collect the measurement outcomes to form a dataset $\{(t_n, {\xi}_n)\}_{n\in [N]}$, where ${\xi}_n\in \{\pm1\pm\i\}$.
    \item Perform classical post-processing on the dataset $\{(t_n, {\xi}_n)\}_{n\in [N]}$ to estimate the dominant eigenvalues $\{\lambda_m\}_{m\in \mc{D}}$.
\end{enumerate}

\begin{figure}[htbp]
    \centering
        \begin{equation*}
        \Qcircuit @C=1em @R=1em {
          \lstick{|0\rangle} & \gate{\rm H}  & \qw& \ctrl{1}  &\gate{W}      & \gate{\rm H} & \meter \\
          \lstick{\ket{0}} & \qw & \gate{{\rm INIT}}  & \gate{e^{-\mathbf{i}t_n H}} &\qw & \qw & \qw&
        }
        \end{equation*}
        \caption{Hadamard test circuit. $W$ is either the identity or the phase gate $S^\dagger$. Here ${\rm INIT} \ket{0}=\ket{\phi}$.}
        \label{fig:hadamard}
\end{figure}
For each $n$, the Hadamard test circuit guarantees that
\begin{align}
    \E[{\xi}_n]=\langle \phi|e^{-\i H t_n} |\phi\rangle = \sum_{m\in [M]} \underbrace{|\langle \phi|E_m\rangle|^2}_{:=p_m}\cdot e^{-\i \lambda_m t_n}\,.
\end{align}
Thus, the dataset $\{(t_n, {\xi}_n)\}_{n\in [N]}$ corresponds to  unbiased noisy samples of the Fourier signal
\begin{align}
    x(t)=\sum_{m\in [M]}p_m \exp(-\i \lambda_m t).
\end{align}
Estimating the dominant eigenvalues $\{\lambda_m\}_{m\in \mc{D}}$ from these samples is then equivalent to the classical spectral estimation problem, which has been extensively studied in signal processing. Different post-Kitaev phase estimation algorithms primarily vary in the choice of $t_n$ in the first step and the classical spectral estimation technique used in the third step to extract the dominant eigenvalues $\{\lambda_m\}_{m\in \mc{D}}$ from these samples. Examples include complex exponential least squares \cite{DingLin2023,Ding2023simultaneous,ding2023robust}, robust phase estimation \cite{ni2023lowdepth}, Gaussian filtering and searching \cite{wfz22,ding2024quantum}, ESPRIT \cite{ni2023lowdepth_2,ding2024esprit}, compressed sensing \cite{yi2024quantum,castaldo2025heisenberg}, to name a few.

Theoretical analyses have demonstrated that state-of-the-art post-Kitaev algorithms can achieve Heisenberg-limited scaling. Specifically, for a target precision $\epsilon$, the total Hamiltonian evolution time required scales as $1/\epsilon$, ensuring that the algorithm estimates each dominant eigenvalue to within $\epsilon$-accuracy~\cite{Ding2023simultaneous,ding2024quantum,ni2023lowdepth_2}. Moreover, these algorithms achieve ``short'' circuit depths: for well-separated dominant eigenvalues, the maximal evolution time scales as $p_{\rm tail} / \epsilon$. These properties, along with other advantages (see \cite{ding2024quantum}), make post-Kitaev methods particularly suitable for early fault-tolerant quantum devices.

However, QPE and post-Kitaev phase estimation algorithms encounter fundamental limitations when it comes to resolving closely spaced eigenvalues or determining the multiplicities of degenerate eigenvalues. In the nearly degenerate case, although these algorithms guarantee that the output is close to one of the dominant eigenvalues, they cannot distinguish eigenvalues separated by extremely small gaps. For instance, if a Hamiltonian has two dominant eigenvalues with an arbitrarily small separation, the algorithm may return a single estimate that is $\epsilon$-close to both, without revealing their multiplicities. In the degenerate case, when access is restricted to a single initial state, these algorithms are incapable of producing the correct multiplicity. This limitation is formalized in~\cref{prop:no-go}.

\section{Uniform overlap condition for the QFAMES algorithm}
\label{sec:assump}

For a matrix $A\in\mathbb{C}^{m\times n}$ with $m\ge n$, denote its singular values as $\{s_i\}_{i\in [n]}$. Then we define $s_{\min}(A)=\min_{i} s_i$ to be the smallest singular value of $A$. We also define
\begin{equation}
s_{\mathrm{avg}}(A)=\sqrt{\frac{\sum_{i} s_i^2}{n}}=\frac{\norm{A}_F}{\sqrt{n}}
\end{equation}
to be the average singular value of $A$, where $\norm{A}_F$ is the Frobenius norm of $A$.

In this section, we formally present the uniform overlap condition for the QFAMES algorithm.

\begin{assumption}[Uniform overlap condition]\label{asp:linear_dep}
For each $i\in [I]$, let $\mc{D}_i:=\{m\in \mc{D}|\lambda_m = \lambda_i^{\star}\}$ denote the index set of dominant eigenvectors with the same eigenvalue $\lambda_i^{\star}$.
We assume there exists a constant $\chi>0$ such that, for every $i\in [I]$, the following inequalities hold:
\begin{equation}\label{eqn:linear_independent}
s_{\min}\left(\Phi_{:,\mc{D}_i}\right)>\frac{s_{\mathrm{avg}}\left(\Phi_{:,\mc{D}_i}\right)}{1+\chi},\quad s_{\min}\left(\Psi_{:,\mc{D}_i}\right)>\frac{s_{\mathrm{avg}}\left(\Psi_{:,\mc{D}_i}\right)}{1+\chi}.
\end{equation}
\end{assumption}

Recall from \cref{eqn:sufficient_coverage_informal} that
\begin{equation}
\min_{\ket{E}\in \mathcal{E}_i, \braket{E|E} = 1 }  \sum_{l\in [L]} \left| \braket{\phi_l | E} \right|^2=\min_{\norm{u}=1}  \left\|\Phi_{:,\mc{D}_i}u\right\|^2=s_{\min}\left(\Phi_{:,\mc{D}_i}\right)^2.
\end{equation}
Then the uniform overlap condition states that the {smallest} singular value must be $(1+\chi)^{-1}$ times larger than the average singular value, for both $\Phi_{:,\mc{D}_i}$ and $\Psi_{:,\mc{D}_i}$.

The uniform overlap condition implies that $\Phi_{:,\mc{D}_i}$ and $\Psi_{:,\mc{D}_i}$ have full column {rank}. The latter is a \emph{necessary} condition for any algorithm {in this access model} to correctly estimate multiplicities. This is summarized in the following proposition. As a corollary, when there is only a single initial state, it is impossible to resolve eigenstate degeneracy in the setting of \cref{prop:no-go}, which corresponds {to} the case of QPE and post-Kitaev phase estimation algorithms.

\begin{proposition}[Information-theoretic indistinguishability]
    \label{prop:no-go}
    Assume that we only have access to data produced by the Hadamard test with controlled Hamiltonian evolution $\exp(-\i Ht)$ and left and right initial states with overlap matrices $\Phi$ and $\Psi$ as defined in \cref{eq:def_Phi_Psi}. Let $\mc{D}_i$ denote the indices of a degenerate energy eigenvalue subspace. There exists no classical signal-processing algorithm that can reliably determine the multiplicity $|\mc{D}_i|$,
    if $ {\rm rank} \left( \Phi_{:,\mc{D}_i}  \right) < |\mc{D}_i|$ or $ {\rm rank} \left( \Psi_{:,\mc{D}_i} \right) < |\mc{D}_i|$.
\end{proposition}

\begin{proof} For simplicity, we assume $\lambda^\star$ is the only eigenvalue of $H$ with degenerate eigenvectors indexed by $\mc{D}$, meaning $H=\lambda^\star\sum_{{i}\in \mc{D}}\ket{\varphi_i}\bra{\varphi_i}$. According to the assumptions, when the noise is negligible, we have access to quantities of the form
    \begin{align}
        \mc{Z}_{l,r} (t) =  \braket{\phi_l | e^{-\i Ht} | \psi_r} = e^{-\i \lambda^\star t} \sum_{m\in\mc{D}} \Phi_{l,m} \Psi_{r,m}^{*}\,.
    \end{align}

    The corresponding data matrix is $\Phi \Psi^{\dagger}$. If $ {\rm rank} \left( \Phi  \right) = k < |\mc{D}|$, consider the singular value decomposition of $\Phi = U S V^{\dagger}$, where the singular values $s_j = S_{j,j}$ are sorted in descending order. Therefore $(\Phi V)_{:,j} = (US)_{:, j} = 0$ for any $j \ge k$ (assume the column vectors are indexed by $j \in [|\mc{D}|] = \{ 0, 1, \cdots, |\mc{D}| - 1\}$). Let us take $\tilde{\Phi}\in \mathbb{C}^{L\times (k+1)}$ and $ \tilde{\Psi}\in \mathbb{C}^{R\times (k+1)}$
    defined as
    \begin{equation}
        \begin{aligned}
             \tilde{\Phi}_{:,j} &:= \left(\Phi V\right)_{:,j},\ \tilde{\Psi}_{:,j} := \left(\Psi V\right)_{:,j},\text{ for }j\in[k]\,;\\
             \tilde{\Phi}_{:,k} &:= 0,\ \tilde{\Psi}_{r,k} := \sqrt{1 - \sum_{j \in [k]} \left|  \tilde{\Psi}_{r,j}  \right|^2
             },\text{ for } r\in [R]\,.
        \end{aligned}
    \end{equation}
    $\tilde{\Phi}, \tilde{\Psi}$ correspond to the overlap matrices generated by
    a different Hamiltonian $\tilde{H}$ that has only multiplicity $k$ at the dominant eigenvalue $\lambda^{\star}$, and an additional eigenvector $\ket{E_k}$ with a different eigenvalue. $\ket{E_k}$ is defined to ensure the normalization conditions for the overlap matrices: $\left\|\tilde{\Phi}_{l,:}\right\|= \left\|\tilde{\Psi}_{r,:}\right\|= 1$, for $l\in [L]$ and $r \in [R]$.

    Since by construction
    \begin{align}
        \tilde{\Phi}\tilde{\Psi}^{\dagger} = \tilde{\Phi}_{:,[k]}\tilde{\Psi}_{:,[k]}^{\dagger} = \Phi_{:,\mc{D}} V V^{\dagger}{\Psi}_{:,\mc{D}}^{\dagger} =  \Phi{\Psi}^{\dagger}\,,
    \end{align}
    we cannot distinguish the signals generated by $H$ and $\tilde{H}$, and hence cannot correctly resolve the multiplicity. The case for $ {\rm rank} \left( \Psi_{:,\mc{D}} \right) < |\mc{D}|$ is analogous.
\end{proof}

In our proof of the main theorem (\cref{thm:non_orthogonal_location}), the uniform overlap condition plays an important role in two aspects:

\begin{enumerate}
    \item In eigenvalue estimation, the condition~\cref{eqn:linear_independent} is applied in~\cref{eqn:Phi_Psi_lower_bound} to ensure that
    \[
        \left\|\Phi_{:,{\mc{D}_{\pi_\theta}}}\left(\Psi_{:,{\mc{D}_{\pi_\theta}}}\right)^\dagger\right\|_F^2
        \quad \text{and} \quad
        \|\mc{G}(\theta\approx \lambda^\star_{{\pi_\theta}})\|_F
    \]
    admit sufficiently large lower bounds. This guarantees that the dominant eigenvalues can be detected by identifying the peaks of $\|\mc{G}(\theta)\|_F$.

    \item In multiplicity estimation, the condition ensures that the overlap matrix
    \[
\Phi_{:,\mc{D}_{\pi_\theta}}\left(\Psi_{:,\mc{D}_{\pi_\theta}}\right)^\dagger
    \]
    has rank equal to the corresponding degeneracy. This property is then quantitatively exploited in~\cref{eqn:lower_bound_singularvalue} to establish a lower bound on the singular values of $\mc{G}(\theta\approx \theta^\star_{\pi_i})$.
\end{enumerate}

\section{Rigorous version of~\cref{thm:main} and proof}\label{sec:rigorous_version}

In this section, we introduce and prove a generalized and rigorous version of~\cref{thm:main}. In the more general setting, we consider the case where each dominant eigenvalue is not exactly degenerate but instead lies within a small interval of width $\delta$ that is centered at $\lambda^\star_i$. This motivates the following definition:
\begin{definition}[Dominant eigenvalue clusters]\label{asp:eig_cluster}
There {exist} parameters $\Delta \gg \delta > 0$ and $I\in \mathbb{N}_+$ such that:
\begin{itemize}
    \item The dominant eigenvalues $\{\lambda_m\}_{m\in \mc{D}}$ are covered by $I$ disjoint cluster intervals $\{\mc{I}_i\}_{i=1}^I$ that are centered at $\{\lambda_i^{\star}\}^I_{i=1}$.
    \item Each cluster interval has length at most $\delta$.
    \item The cluster intervals are $\Delta$-well-separated. That is, for all $i\ne j\in [I]$, $\mathrm{dis}(\mc{I}_i, \mc{I}_j)\geq \Delta$.
\end{itemize}
Here, $I$ is the number of dominant cluster intervals, $\delta$ is the maximum width of each interval, and $\Delta$ is the minimum pairwise separation between intervals.
\end{definition}
We note that \cref{thm:main} corresponds to the special case $\delta = 0$, where each cluster reduces to a single dominant eigenvalue with exact degeneracy. In the more general setting, a dominant eigenvalue cluster may be viewed as an \emph{approximately degenerate} eigenvalue, and our goal is to estimate both the center of each cluster and the number of dominant eigenvalues it contains. (We slightly abuse notation and use $\lambda_i^\star$ to denote the center of the $i$-th cluster $\mc{I}_i$.) The QFAMES algorithm (\cref{alg:QFAMES}) applies directly to this case.

We are now ready to state a generalized and rigorous version of \cref{thm:main}:
\begin{theorem}[General version of~\cref{thm:main}]\label{thm:non_orthogonal_location}
Suppose~\cref{asp:linear_dep}
hold with $\chi=\mc{O}(1)$, and the dominant eigenvalue cluster length parameter $\delta$ is sufficiently small. Given failure probability $\eta\in (0,1)$, if the parameters satisfy the following conditions: $\widetilde{I}\geq I$,
\begin{equation}\label{eqn:condition_old}
T=\Omega\left(\frac{1}{\Delta}\log\left(\frac{KLR}{p_{\rm tail}}\right)\right),\
N=\Omega\left(\frac{LR}{p_{\rm tail}^2}\log\left(\left(J+K\right)\frac{LR}{\eta}\right)\right),\
\sigma=\Omega\left(\log^{1/2}\left( \frac{\sqrt{LR}}{p_{\rm tail}} \right)\right)\,,
\end{equation}
\begin{equation}\label{eqn:condition_old_2}
q=\mc{O}\left(\frac{p_{\rm tail}}{(1+\sigma)\sqrt{KLR}}\right), \quad q=\mc{O}\left(\sqrt{\log\left(\frac{1+2\chi}{1+\chi}\right)}\right), \quad \alpha=\mc{O}(\Delta\cdot T)\,,
\end{equation}
and
\begin{equation}\label{eqn:T_alpha_extra}
T=\mc{O}\left(\frac{1}{\delta}\min\left\{\frac{p_{\rm tail}}{(1+\sigma)\sqrt{KLR}},\frac{p^2_{\rm tail}}{KLR}\right\}\right),\quad \alpha=\Omega(\delta \cdot T)\,,
\end{equation}
\begin{align}\label{eqn:suff_dom}
    {\frac{p_{\min}}{p_{\rm tail}}} > 6(1 + 2\chi).
\end{align}
Here $K = |\mc{D}|$, and $J = \left\lfloor \frac{2\pi T}{q}\right\rfloor$ is the number of candidate locations of cluster centers $\left\{\theta_i\right\}$ in the block and searching algorithm.

Then, with probability at least $1-\eta$, there exists a permutation $\pi:[I]\rightarrow [I]$
such that the output of~\cref{alg:QFAMES} $\{(\theta_j^\star,m_j)\}_{j\in [\wt{I}]}$ satisfies the following guarantees:
\begin{equation}\label{eqn:location_error}
\sup_{1\leq j\leq I}|\theta^\star_j-\lambda^\star_{\pi(j)}|=\mc{O}\left((1+\sigma)\frac{p_{\rm tail}}{p_{\min}}\frac{1}{T}\right)\,.
\end{equation}
Furthermore, if $p_{\rm tail}/p_{\min},\delta$ are sufficiently small so that
\begin{equation}\label{eqn:condition_rank}
\frac{1-\mc{O}\left(\left((1+\sigma)p_{\rm tail}/p_{\min}+\delta\cdot T\right)^2\right)}{(1+\chi)^2}=\Omega\left(\frac{p_{\rm tail}}{p_{\min}}\right)\,.
\end{equation}
we can choose $\tau=\Theta(p_{\rm tail})$, and $\alpha=\Omega(\log(KLR/p_{\rm tail}))$ to ensure that, with probability at least $1-\eta$, we have
\begin{align}
    m_{j}=|\mc{D}_{\pi(j)}|,\quad 1\leq j\leq I,\quad m_j=0,\quad j>I\,.
\end{align}
\end{theorem}
\begin{remark}\label{rem:thm_condition} Here, we give some comments about the conditions and assumptions in the above theorem:
\begin{itemize}
\item The conditions~\cref{eqn:condition_old} and~\cref{eqn:condition_old_2} are similar to those in~\cite[Theorem 3.2]{ding2024quantum}, which ensure the algorithm can identify each cluster separately.

\item In the above theorem,~\cref{eqn:T_alpha_extra} essentially requires that $\delta$ is sufficiently small so that all dominant eigenvalues within a cluster can be approximated by a single output. When $T=\Theta(\alpha/\delta)$, the analysis becomes more involved, as a single cluster may yield multiple outputs, making the second step of the algorithm difficult to track.

\item {\cref{eqn:suff_dom} can be referred to as the sufficient domination condition, defined after \cref{prob:dods_estimation}.}

\item Although the condition~\cref{eqn:condition_rank} appears technical, it essentially requires that $\chi$ and $p_{\rm tail}$ are small so that the right-hand side is close to $1$. This condition ensures a separation between the singular values of $\mc{W}_j$, which is necessary for accurately counting the number of dominant eigenvalues in each cluster.
\end{itemize}

\end{remark}

\begin{remark}
    Together, the outputs of the location and multiplicity estimation define an approximate density $\wt{\mu}$ which approximates the true density $\mu_D$ up to a small error in Wasserstein distance.
Let $\theta_i^\star$ be the estimated center of $\mc{I}_i$. Define
\begin{equation}
    \begin{aligned}
    \wt{\mu}(E):=&~ \frac{1}{|\mc{D}|}{\sum^I_{i=1}}m_i\delta(E-\theta_i^\star)\,,\\
    \mu_D(E) := &~ \frac{1}{|{\mc{D}}|} \sum_{i\in \mc{D}}\delta(E-\lambda_i)\,.
    \end{aligned}
\end{equation}
{Let $\epsilon:=\sup_{1\leq j\leq I}|\theta_j^\star-\lambda_{\pi(j)}^\star|$.}
Consider a transport map that transports each atom at {$\lambda_j$} in $\mc{I}_i$ to $\theta_i^\star$. The cost is
\begin{align}
    {|\lambda_j - \theta_i^\star| \leq |\lambda_j - \lambda_i^\star| + |\lambda_i^\star - \theta_i^\star|\leq \epsilon+\frac{\delta}{2}}\,.
\end{align}
Averaging over the {$|\mc{D}|$} atoms gives that
\begin{align}
    W_1(\mu_D, \wt{\mu})\leq \epsilon + \frac{\delta}{2}\,.
\end{align}
\end{remark}

\begin{proof}
The proof strategy of \cref{thm:non_orthogonal_location} is similar to the proof of~\cite[Theorem 3.2]{ding2024quantum}. We first show a rough estimation for the location of the clusters and then refine it to get a more accurate estimation. First of all, let us define the exact filter matrix $\mc{F}(\theta)$, the data matrix $G(\theta)$ that approximates $\mc{F}(\theta)$, the error matrix $E(\theta)$, and the computed filter value function $\mc{W}(\theta)$:
\begin{equation}
    \begin{aligned}
\mc{F}(\theta):=&~ \Phi \cdot \diag\left(\{\exp(-(\theta-\lambda_m)^2T^2)\}_{m\in [M]}\right) \cdot \Psi^\dagger\in \C^{L\times R}\,,\\
G(\theta):=&~ \frac{1}{N}\sum^N_{n=1}Z_n\exp(i\theta t_n)\in \C^{L\times R}\,,\\
\mathcal{E}(\theta):=&~ G(\theta)- \mc{F}(\theta)\in \C^{L\times R}\,,\\
\mc{W}(\theta):=&~ \left\|G(\theta)\right\|_F\in \R\,.
\end{aligned}
\end{equation}

\textbf{Step 0:} We show that the computed filter value function $\mc{W}(\theta)$ well approximates $\lVert \mc{F}(\theta) \rVert_F$.

Since $\sigma=\Omega\left(\log^{1/2}\left(\frac{\sqrt{LR}}{p_{\rm tail}}\right)\right)$ and $N=\Omega\left(\frac{LR}{p_{\rm tail}^2}\log\left(\left(J+K\right)\frac{LR}{\eta}\right)\right)$, according to~\cref{lem:truncation_error} with $\beta={p_{\rm tail}} / {8}$,
we have
\begin{equation}\label{eqn:truncation_error}
\sup_{\theta_j}\left|\mc{W}(\theta_j)-\left\| \mc{F}(\theta_j)\right\|_F\right|\leq \sup_{\theta_j}\|\mathcal{E}(\theta_j)\|_F\leq \frac{p_{\rm tail}}{8}
\end{equation}
with probability at least $1-\eta/2$. For simplicity, we only consider the case when~\cref{eqn:truncation_error} is satisfied, which is true with probability at least $1-\eta/2$.

\textbf{Step 1:} For each $\theta^\star_j$, let $\pi_j=\argmin_i\left\{\mathrm{dist}(\theta^\star_j,{\mc{I}_i})\right\}$. We show
\begin{align}
    \sup_{1\leq j\leq I}|\theta^\star_j-\lambda^\star_{\pi_j}|{\leq}\sqrt{\log\left(\frac{2(1+2\chi)}{1+\chi}\right)}\frac{1}{T}
\end{align}
when~\cref{eqn:truncation_error} holds true. To achieve this, we need to prove lower and upper bounds for $\mc{W}(\theta)$:
\begin{itemize}
  \item If $\theta\in \left(\cup_i [\lambda^\star_i-\frac{q}{2T},\lambda^\star_i+\frac{q}{2T}]\right)\cap \left\{\theta_j\right\}^J_{j=1}$, there exists $\pi_\theta$ (index of the cluster that $\theta$ belongs to) such that
  \begin{equation}\label{eqn:theta_distance}
  \begin{aligned}
      &\max_m\left\{|\theta-\lambda_m| : m\in \mc{D}_{\pi_\theta}\right\}\leq \frac{q}{2T}+\frac{\delta}{2},\\
      &\min_m\left\{|\theta-\lambda_m| : m\in \mc{D}\setminus \mc{D}_{\pi_\theta}\right\}\geq \Delta-\frac{q}{2T}\geq \Delta/2\,,
  \end{aligned}
  \end{equation}
  where we use $T\geq \frac{q}{\Delta}$ to ensure the second inequality holds. This implies that $\theta$ lies within a single cluster and is at least a distance of $\Delta/2$ from any other cluster.

Using~\cref{eqn:truncation_error} and the triangle inequality for the Frobenius norm, we observe that
\begin{equation}\label{eqn:W_F_estimation}
\begin{aligned}
\mc{W}(\theta)\geq & \left\|\Phi_{:,\mc{D}_{\pi_\theta}} \cdot \diag\left(\left\{\exp(-(\theta-\lambda_m)^2T^2)\right\}_{m\in \mc{D}_{\pi_\theta}}\right) \cdot \left(\Psi_{:,\mc{D}_{\pi_\theta}}\right)^\dagger\right\|_F\\
&-\sum_{i\neq \pi_\theta} \left\|\Phi_{:,\mc{D}_{i}} \cdot \diag\left(\left\{\exp(-(\theta-\lambda_m)^2T^2)\right\}_{m\in \mc{D}_{i}}\right) \cdot \left(\Psi_{:,\mc{D}_{i}}\right)^\dagger\right\|_F\\
&-\left\|\Phi_{:,\mc{D}^c}\cdot \diag\left(\left\{\exp(-(\theta-\lambda_m)^2T^2)\right\}_{m\in \mc{D}^c}\right) \cdot \left(\Psi_{:,\mc{D}^c}\right)^\dagger\right\|_F-\frac{p_{\rm tail}}{8}
\end{aligned}
\end{equation}
For the second term, using the second inequality of~\cref{eqn:theta_distance} and~\cref{lem:G_theta_property}~\cref{eqn:G_theta_bound}, we have
\begin{equation}\label{eqn:W_F_other_domin_upper_bound}
\begin{aligned}
&\sum_{i\neq \pi_\theta} \left\|\Phi_{:,\mc{D}_{i}} \cdot \diag\left(\left\{\exp(-(\theta-\lambda_m)^2T^2)\right\}_{m\in \mc{D}_{i}}\right) \cdot \left(\Psi_{:,\mc{D}_{i}}\right)^\dagger\right\|_F\\
\leq &~\exp(-\Delta^2T^2/4)\sum_{i\neq \pi_\theta}\sum_{m\in \mc{D}_{i}}\sqrt{\left(\sum^L_{l=1} |\Phi_{l,m}|^2\right)\left(\sum^R_{r=1} |\Psi_{r,m}|^2\right)}\\
\leq &~\exp(-\Delta^2T^2/4)\sqrt{K-1}\sqrt{\sum_{i\neq \pi_\theta}\sum_{m\in \mc{D}_{i}}\left(\sum^L_{l=1} |\Phi_{l,m}|^2\right)\left(\sum^R_{r=1} |\Psi_{r,m}|^2\right)}\\
\leq &~\exp(-\Delta^2T^2/4)\sqrt{K-1}\sqrt{ \sum_{i\neq \pi_\theta}\sum_{m\in \mc{D}_{i}}\sum^L_{l=1} |\Phi_{l,m}|^2
R }\\
\leq &~\sqrt{KLR}\exp(-\Delta^2T^2/4)\leq \frac{p_{\rm tail}}{8}\,.
\end{aligned}
\end{equation}
where we use H\"older's inequality in the second inequality and $T=\Omega\left(\frac{1}{\Delta}\log(KLR/p_{\rm tail})\right)$
in the last inequality. Similarly, for the third term, we have
\begin{equation}
    \begin{aligned}
    &\left\|\Phi_{:,\mc{D}^c}\cdot \diag\left(\left\{\exp(-(\theta-\lambda_m)^2T^2)\right\}_{m\in \mc{D}^c}\right) \cdot \left(\Psi_{:,\mc{D}^c}\right)^\dagger\right\|_F\\
    \leq &~\sum_{m\in \mc{D}^c}\sqrt{\left(\sum^L_{l=1} |\Phi_{l,m}|^2\right)\left(\sum^R_{r=1} |\Psi_{r,m}|^2\right)}=p_{\rm tail}\,.
    \end{aligned}
\end{equation}
Plugging these two inequalities into~\cref{eqn:W_F_estimation}, we have
\begin{equation}\label{W_theta_lower_bound_2}
\begin{aligned}
\mc{W}(\theta)\geq &\left\|\Phi_{:,\mc{D}_{\pi_\theta}} \cdot \diag\left(\left\{\exp(-(\theta-\lambda_m)^2T^2)\right\}_{m\in \mc{D}_{\pi_\theta}}\right) \cdot \left(\Psi_{:,\mc{D}_{\pi_\theta}}\right)^\dagger\right\|_F-\frac{5}{4}p_{\rm tail}
\end{aligned}\,.
\end{equation}
For the first term, we notice that $\left|\lambda_m-\lambda^\star_{\pi_\theta}\right|\leq \delta$ when $m\in\mathcal{D}_{\pi_\theta}$.

Combining this and~\cref{lem:G_theta_property}~\cref{eqn:G_theta_bound} we have
\begin{equation}\label{eqn:small_diff}
\begin{aligned}
&\left\|\Phi_{:,\mc{D}_{\pi_\theta}} \cdot \diag\left(\left\{\exp(-(\theta-\lambda_m)^2T^2)-\exp(-(\theta-\lambda^\star_{\pi_\theta})^2T^2)\right\}_{m\in \mc{D}_{\pi_\theta}}\right) \cdot \left(\Psi_{:,\mc{D}_{\pi_\theta}}\right)^\dagger\right\|_F\\
\leq &~\delta\cdot T\sum_{m\in [\mc{D}_{\pi_\theta}]}\sqrt{\left(\sum^L_{l=1} |\Phi_{l,m}|^2\right)\left(\sum^R_{r=1} |\Psi_{r,m}|^2\right)}\leq \delta\cdot T\sqrt{KLR}\leq \frac{p_{\rm tail}}{4}\,,
\end{aligned}
\end{equation}
where we use $\delta=\mathcal{O}(p_{\rm tail}/(\sqrt{KLR}T))$ in the last equality. This gives
\begin{equation}
    \begin{aligned}
&\left\|\Phi_{:,\mc{D}_{\pi_\theta}} \cdot \diag\left(\left\{\exp(-(\theta-\lambda_m)^2T^2)\right\}_{m\in \mc{D}_{\pi_\theta}}\right) \cdot \left(\Psi_{:,\mc{D}_{\pi_\theta}}\right)^\dagger\right\|_F\\
\geq &~\left\|\Phi_{:,\mc{D}_{\pi_\theta}} \cdot \underbrace{\diag\left(\left\{\exp(-(\theta-\lambda^\star_{\pi_\theta})^2T^2)\right\}_{m\in \mc{D}_{\pi_\theta}}\right)}_{\text{constant diagonal}} \cdot \left(\Psi_{:,\mc{D}_{\pi_\theta}}\right)^\dagger\right\|_F-\frac{p_{\rm tail}}{4}\\
\geq &~\exp\left(-\left(\frac{q}{2T}\right)^2T^2\right)\left\|\Phi_{:,\mc{D}_{\pi_\theta}}\left(\Psi_{:,\mc{D}_{\pi_\theta}}\right)^\dagger\right\|_F-\frac{p_{\rm tail}}{4}\\
\geq &~\frac{1+\chi}{1+2\chi}\left\|\Phi_{:,\mc{D}_{\pi_\theta}}\left(\Psi_{:,\mc{D}_{\pi_\theta}}\right)^\dagger\right\|_F-\frac{p_{\rm tail}}{4}
\end{aligned}
\end{equation}
where we use the fact $q=\mc{O}\left(\sqrt{\log\left(\frac{1+2\chi}{1+\chi}\right)}\right)$ in the last inequality.

Finally, plugging this into~\cref{W_theta_lower_bound_2}, we obtain that
\begin{equation}\label{W_theta_lower_bound_3}
\mc{W}(\theta)\geq \frac{1+\chi}{1+2\chi}\left\|\Phi_{:,\mc{D}_{\pi_\theta}}\left(\Psi_{:,\mc{D}_{\pi_\theta}}\right)^\dagger\right\|_F-\frac{3}{2}p_{\rm tail}\,.
\end{equation}
This gives {a} lower bound of $\mc{W}(\theta)$ if $\theta\in \left(\cup_i \left[\lambda^\star_i-\frac{q}{2T},\lambda^\star_i+\frac{q}{2T}\right]\right)\cap \left\{\theta_j\right\}^J_{j=1}$.

\item Similar to the above calculation and using~\cref{eqn:truncation_error},~\cref{eqn:small_diff} and~\cref{lem:G_theta_property}~\cref{eqn:G_theta_bound} for the first term, we have
\begin{equation}\label{eqn:W_F_upper_bound}
\begin{aligned}
\mc{W}(\theta^\star_j)\leq &~ \left\|\Phi_{:,\mc{D}_{\pi_j}} \cdot \diag\left(\left\{\exp(-(\theta^\star_j-\lambda_m)^2T^2)\right\}_{m\in \mc{D}_{\pi_j}}\right) \cdot \left(\Psi_{:,\mc{D}_{\pi_j}}\right)^\dagger\right\|_F\\
&+\underbrace{\sum_{i\neq \pi_j} \left\|\Phi_{:,\mc{D}_{i}} \cdot \diag\left(\left\{\exp(-(\theta^\star_j-\lambda_m)^2T^2)\right\}_{m\in \mc{D}_{i}}\right) \cdot \left(\Psi_{:,\mc{D}_{i}}\right)^\dagger\right\|_F}_{\leq \sqrt{KLR}\exp(-\Delta^2T^2/4)\leq \frac{p_{\rm tail}}{8}}\\
&+\underbrace{\left\|\Phi_{:,\mc{D}^c}\cdot \diag\left(\left\{\exp(-(\theta^\star_j-\lambda_m)^2T^2)\right\}_{m\in \mc{D}^c}\right) \cdot \left(\Psi_{:,\mc{D}^c}\right)^\dagger\right\|_F}_{\leq p_{\rm tail}}+\frac{p_{\rm tail}}{8}\\
\leq &~ \left\|\Phi_{:,\mc{D}_{\pi_j}} \cdot \diag\left(\left\{\exp(-(\theta^\star_j-\lambda_m)^2T^2)\right\}_{m\in \mc{D}_{\pi_j}}\right) \cdot \left(\Psi_{:,\mc{D}_{\pi_j}}\right)^\dagger\right\|_F+\frac{5}{4}p_{\rm tail}\\
\leq &~\exp\left(-(\theta^\star_j-\lambda^\star_{\pi_j})^2T^2\right)\left\|\Phi_{:,\mc{D}_{\pi_j}}\left(\Psi_{:,\mc{D}_{\pi_j}}\right)^\dagger\right\|_F+\frac{3}{2}p_{\rm tail}\,.
\end{aligned}
\end{equation}
This gives an upper bound of $\mc{W}(\theta^\star_j)$.
\end{itemize}

According to~\cref{W_theta_lower_bound_3,eqn:W_F_upper_bound}, when $j=1$, because $\mc{W}(\theta^\star_1)$ takes the maximal value, we have
\begin{equation}
\begin{aligned}
    \exp\left(-(\theta^\star_1-\lambda^\star_{\pi_1})^2T^2\right)\left\|\Phi_{:,\mc{D}_{\pi_1}}\left(\Psi_{:,\mc{D}_{\pi_1}}\right)^\dagger\right\|_F+\frac{3}{2}p_{\rm tail} \geq \mc{W}(\theta^\star_1)\geq \frac{1+\chi}{1+2\chi}\left\|\Phi_{:,\mc{D}_{\pi_1}}\left(\Psi_{:,\mc{D}_{\pi_1}}\right)^\dagger\right\|_F-\frac{3}{2}p_{\rm tail}\,,
\end{aligned}
\end{equation}
where the second inequality holds since there must exist some grid point $\theta_1$ such that $\theta_1\in [\lambda^\star_{\pi_1}-\frac{q}{2T},\lambda^\star_{\pi_1}+\frac{q}{2T}]$ and $\mc{W}(\theta_1^\star)\geq \mc{W}(\theta_1)$ by definition. Thus,
\begin{equation}
\begin{aligned}
&\exp\left(-(\theta^\star_1-\lambda^\star_{\pi_1})^2T^2\right)\left\|\Phi_{:,\mc{D}_{\pi_1}}\left(\Psi_{:,\mc{D}_{\pi_1}}\right)^\dagger\right\|_F+\frac{3}{2}p_{\rm tail}\geq \frac{1+\chi}{1+2\chi}\left\|\Phi_{:,\mc{D}_{\pi_1}}\left(\Psi_{:,\mc{D}_{\pi_1}}\right)^\dagger\right\|_F-\frac{3}{2}p_{\rm tail}\\
\Longrightarrow\quad &|\theta^\star_1-\lambda^\star_{\pi_1}|\leq \log^{1/2}\Bigg(\frac{1}{\frac{1+\chi}{2+\chi}-3p_{\rm tail}/\big\|\Phi_{:,\mc{D}_{\pi_1}}\left(\Psi_{:,\mc{D}_{\pi_1}}\right)^\dagger\big\|_F}\Bigg)\frac{1}{T}
\end{aligned}
\end{equation}

We further note that, 
\begin{equation}\label{eqn:Phi_Psi_lower_bound}
\begin{aligned}
\left\|\Phi_{:,\mc{D}_{\pi_\theta}}\left(\Psi_{:,\mc{D}_{\pi_\theta}}\right)^\dagger\right\|_F^2=&~
\mathrm{Tr}\left(\Psi_{:,\mc{D}_{\pi_\theta}}\left(\Phi_{:,\mc{D}_{\pi_\theta}}\right)^\dagger\Phi_{:,\mc{D}_{\pi_\theta}}\left(\Psi_{:,\mc{D}_{\pi_\theta}}\right)^\dagger\right)\\
\geq &~ \lambda_{\min}\left(
\left(\Phi_{:,\mc{D}_{\pi_\theta}}\right)^\dagger\Phi_{:,\mc{D}_{\pi_\theta}}\right)
\mathrm{Tr}\left(
\left(\Psi_{:,\mc{D}_{\pi_\theta}}\right)^\dagger \Psi_{:,\mc{D}_{\pi_\theta}}
\right)\\
= &~ \lambda_{\min}\left(
\left(\Phi_{:,\mc{D}_{\pi_\theta}}\right)^\dagger\Phi_{:,\mc{D}_{\pi_\theta}}\right) \left(\sum_{i\in\mc{D}_{\pi_\theta}}\sum^R_{k=1}\left|\Psi_{k,i}\right|^2\right)\\
\geq &~\left(\frac{1}{1+\chi}\right)^2\frac{1}{|\mc{D}_{\pi_\theta}|} \mathrm{Tr}\left(
\left(\Phi_{:,\mc{D}_{\pi_\theta}}\right)^\dagger \Phi_{:,\mc{D}_{\pi_\theta}}
\right)\left(\sum_{i\in\mc{D}_{\pi_\theta}}\sum^R_{k=1}\left|\Psi_{k,i}\right|^2\right)\\
=&~\left(\frac{1}{1+\chi}\right)^2\frac{1}{|\mc{D}_{\pi_\theta}|}\left(\sum_{j\in\mc{D}_{\pi_\theta}}\sum^L_{k=1}\left|\Phi_{k,j}\right|^2\right)\left(\sum_{i\in\mc{D}_{\pi_\theta}}\sum^R_{k=1}\left|\Psi_{k,i}\right|^2\right)\\
\geq &~\left(\frac{1}{1+\chi}\right)^2\frac{1}{|\mc{D}_{\pi_\theta}|}\left(\sum_{i\in\mc{D}_{\pi_\theta}}\left(\sum^L_{k=1}\left|\Phi_{k,i}\right|^2\right)\left(\sum^R_{k=1}\left|\Psi_{k,i}\right|^2\right)\right)\\
\geq &~ \left(\frac{1}{1+\chi}\right)^2 p_{\rm min}^2.
\end{aligned}
\end{equation}

The second inequality comes from \cref{asp:linear_dep}.
Plugging this bound into the above inequality, we have
\begin{equation}\label{eqn:bound_theta_star_1}
|\theta^\star_1-\lambda^\star_{\pi_1}|<\sqrt{\log\left(\frac{1}{\frac{1+\chi}{2+\chi}-3(1+\chi)p_{\rm tail}/p_{\rm min}}\right)}\frac{1}{T}\leq \sqrt{\log\left(\frac{2(1+2\chi)}{1+\chi}\right)}\frac{1}{T}\,.
\end{equation}
where we use $3{p_{\rm tail}} / {p_{\min}}\leq \frac{1}{2(1+2\chi)}$ (\cref{eqn:suff_dom}) in the last step.

Furthermore, since $\mathrm{dist}({\mc{I}_i,\mc{I}_j})>\Delta$, $T=\Omega\left(\left(1+\alpha+q\right)/\Delta\right)$ and $\delta\ll \Delta$, we have
\begin{equation}\label{eqn:sep_btw_one_and_other}
\left[\theta^\star_1-\frac{\alpha}{T},\theta^\star_1+\frac{\alpha}{T}\right]\cap \left(\cup_{i\neq \pi_1} \left[\lambda^\star_i-\frac{q}{2T},\lambda^\star_i+\frac{q}{2T}\right]\right)=\emptyset
\end{equation}
This implies that the first block interval does not overlap with $\cup_{i\neq \pi_1}\left[\lambda^\star_i-\frac{q}{2T},\lambda^\star_i+\frac{q}{2T}\right]$. In addition, $\alpha=\Omega\left(\max\left\{\sqrt{\log\left(\frac{2(1+2\chi)}{1+\chi}\right)},\delta\cdot T\right\}\right)$, we have ${\pi_2}\neq \pi_1$, meaning $\theta^\star_2$ is not close to $\lambda^\star_{\pi_1}$.
More specifically, by \cref{eqn:sep_btw_one_and_other}, we can repeat the above argument for $\theta_2^\star$ and obtain
\begin{align}
    |\theta_2^\star-\lambda_{{\pi_2}}^\star|\leq \sqrt{\log\left(\frac{2(1+2\chi)}{1+\chi}\right)}\frac{1}{T}= \mc{O}\left(\frac{\alpha}{T}\right)\,.
\end{align}
If {$\pi_2$}$=\pi_1$, then by \cref{eqn:bound_theta_star_1} and triangle inequality, we have
\begin{align}
    |\theta_2^\star - \theta_1^\star|\leq |\theta_2^\star-\lambda_{\pi_1}^\star| + |\theta_1^\star-\lambda_{\pi_1}^\star| < \frac{\alpha}{T}\,,
\end{align}
which is impossible since the grid points that are $\frac{\alpha}{T}$-close to $\theta_{1}^\star$ have been blocked.

Repeating this argument for all $1\leq j\leq I$, we can show
\begin{equation}
|\theta^\star_j-\lambda^\star_{\pi_j}|\leq\sqrt{\log\left(\frac{2(1+2\chi)}{1+\chi}\right)}\frac{1}{T}
\end{equation}
for all $1\leq j\leq I$.

\textbf{Step 2:} Using the result obtained in Step 1, we can impose stronger bounds on $\mc{W}(\theta)$ to refine the error estimation.
We will show
\begin{equation}
    \sup_{1\leq j\leq I}|\theta^\star_j-\lambda^\star_{\pi_j}|=\mc{O}\left((1+\sigma)\frac{p_{\rm tail}}{p_{\min}}\frac{1}{T}\right)
\end{equation}
when~\cref{eqn:truncation_error} holds true.

Without loss of generality, we only consider $j=1$. The other cases can be improved similarly. Let us pick $\widetilde{\theta}^\star_1\in  [\lambda^\star_{\pi_1}-\frac{q}{2T},\lambda^\star_{\pi_1}+\frac{q}{2T}]\cap \left\{\theta_j\right\}^J_{j=1}$, which is the closest energy grid point to the desired cluster center $\lambda_{\pi_1}^{\star}$. It can vary from our computed estimate $\theta_1^{\star}$ due to the errors arising from noise, contributions of non-dominant eigenvectors, and of dominant eigenvectors in distant clusters, and it holds that $\mc{W}(\theta_1^{\star}) \ge \mc{W}(\widetilde{\theta}_1^{\star})$. We will bound the distance between $\theta_1^{\star}$ and $\widetilde{\theta}_1^{\star}$, and thus between $\theta_1^{\star}$ and  $\lambda_{\pi_1}^{\star}$.

We first define
\begin{equation}
    \begin{aligned}
  G(\theta)=&~\underbrace{\Phi_{:,\mc{D}_{\pi_1}} \cdot \diag\left(\left\{\exp(-(\theta-\lambda_m)^2T^2)\right\}_{m\in \mc{D}_{\pi_1}}\right) \cdot \left(\Psi_{:,\mc{D}_{\pi_1}}\right)^\dagger}_{:=A(\theta)}\\
  +&\underbrace{\Phi_{:,\mc{D}^c} \cdot \diag\left(\left\{\exp(-(\theta-\lambda_m)^2T^2)\right\}_{m\in \mc{D}^c}\right) \cdot \left(\Psi_{:,\mc{D}^c}\right)^\dagger}_{:=B(\theta)}\\
  +&\underbrace{\sum_{i\neq \pi_1}\Phi_{:,\mc{D}_{i}} \cdot \diag\left(\left\{\exp(-(\theta-\lambda_m)^2T^2)\right\}_{m\in \mc{D}_{i}}\right) \cdot \left(\Psi_{:,\mc{D}_{i}}\right)^\dagger}_{:=C(\theta)}+\mathcal{E}(\theta)\,.
\end{aligned}
\end{equation}
According to~\cref{eqn:W_F_upper_bound},~\cref{eqn:truncation_error} and~\cref{eqn:W_F_other_domin_upper_bound}, we have
\begin{equation}\label{eqn:A_bound}
\begin{aligned}
&\left\|A(\theta^\star_1)\right\|_F\leq \exp\left(-(\theta^\star_1-{\lambda^{\star}_{\pi_1}})^2T^2\right)\left\|\Phi_{:,\mc{D}_{\pi_1}}\left(\Psi_{:,\mc{D}_{\pi_1}}\right)^\dagger\right\|_F+\frac{1}{4}p_{\rm tail}\leq \frac{5}{4}\left\|\Phi_{:,\mc{D}_{\pi_1}}\left(\Psi_{:,\mc{D}_{\pi_1}}\right)^\dagger\right\|_F,\\
&\left\|A\left(\widetilde{\theta}^\star_{1}\right)\right\|_F\leq \exp\left(-\left(\widetilde{\theta}^\star_1-{\lambda^{\star}_{\pi_1}}\right)^2T^2\right)\left\|\Phi_{:,\mc{D}_{\pi_1}}\left(\Psi_{:,\mc{D}_{\pi_1}}\right)^\dagger\right\|_F+\frac{1}{4}p_{\rm tail}\leq \frac{5}{4}\left\|\Phi_{:,\mc{D}_{\pi_1}}\left(\Psi_{:,\mc{D}_{\pi_1}}\right)^\dagger\right\|_F\,,
\end{aligned}
\end{equation}
where we use $p_{\rm tail}<\frac{p_{\min}}{1+\chi}\leq \left\|\Phi_{:,\mc{D}_{\pi_1}}\left(\Psi_{:,\mc{D}_{\pi_1}}\right)^\dagger\right\|_F$ according to~\cref{eqn:Phi_Psi_lower_bound} in the second inequalities, and
\begin{equation}\label{eqn:B_c_bound}
\left\|B(\theta^\star_1)+C(\theta^\star_1)+\mathcal{E}\left(\theta^\star_{1}\right)\right\|_F\leq \frac{5p_{\rm tail}}{4},\quad \left\|B\left(\widetilde{\theta}^\star_{1}\right)+C\left(\widetilde{\theta}^\star_{1}\right)+\mathcal{E}\left(\widetilde{\theta}^\star_{1}\right)\right\|_F\leq \frac{5p_{\rm tail}}{4}\,.
\end{equation}
Then, we get that
\begin{equation}
    \begin{aligned}
\mc{W}^2(\theta)=&~ \mathrm{Tr}(G^\dagger(\theta)G(\theta))\\
=&~ \mathrm{Tr}(A^\dagger(\theta)A(\theta))+2\mathrm{Re}\left(\mathrm{Tr}(A^\dagger(\theta)(B(\theta)+C(\theta)+\mathcal{E}(\theta)))\right)\\
&+\mathrm{Tr}((B^\dagger(\theta)+C^\dagger(\theta)+\mathcal{E}^\dagger(\theta))(B(\theta)+C(\theta)+\mathcal{E}(\theta)))\,.
\end{aligned}
\end{equation}
We first deal with the first term, the last term, and then the second term:
\begin{itemize}
\item For the first term, we have
\begin{equation}\label{eqn:A_theta}
    \begin{aligned}
A(\theta)=&~\Phi_{:,\mc{D}_{\pi_1}} \cdot \diag\left(\left\{\exp(-(\theta-\lambda^\star_{\pi_1})^2T^2)\right\}_{m\in \mc{D}_{\pi_1}}\right) \cdot \left(\Psi_{:,\mc{D}_{\pi_1}}\right)^\dagger\\
&+\Phi_{:,\mc{D}_{\pi_1}} \cdot \diag\left(\left\{\exp(-(\theta-\lambda_m)^2T^2)-\exp(-(\theta-\lambda^\star_{\pi_1})^2T^2)\right\}_{m\in \mc{D}_{\pi_1}}\right) \cdot \left(\Psi_{:,\mc{D}_{\pi_1}}\right)^\dagger\\
\end{aligned}\,.
\end{equation}
Noticing $\left|\lambda_m-\lambda^\star_{\pi_1}\right|\leq \delta$ when $m\in\mathcal{D}_{\pi_1}$, we have
\begin{equation}\label{eqn:Phi_diff}
\left\|\Phi_{:,\mc{D}_{\pi_1}} \cdot \diag\left(\left\{\exp(-(\theta-\lambda_m)^2T^2)-\exp(-(\theta-\lambda^\star_{\pi_1})^2T^2)\right\}_{m\in \mc{D}_{\pi_1}}\right) \cdot \left(\Psi_{:,\mc{D}_{\pi_1}}\right)^\dagger\right\|_F\leq \sqrt{KLR} \delta\cdot T
\end{equation}
according to~\cref{lem:G_theta_property}~\cref{eqn:G_theta_lipschitz}. Thus, we obtain that
\begin{equation}\label{eqn:A_theta_diff}
\left|\left\|A(\theta)\right\|_F-\left\|\Phi_{:,\mc{D}_{\pi_1}} \cdot \diag\left(\left\{\exp(-(\theta-\lambda^\star_{\pi_1})^2T^2)\right\}_{m\in \mc{D}_{\pi_1}}\right) \cdot \left(\Psi_{:,\mc{D}_{\pi_1}}\right)^\dagger\right\|_F\right|\leq \sqrt{KLR}\delta\cdot T\,,
\end{equation}
where we use $\delta\cdot T<1$.

\item For the last term, by~\cref{eqn:B_c_bound}, we have
\begin{equation}
\begin{aligned}
&\mathrm{Tr}((B^\dagger(\theta^\star_1)+C^\dagger(\theta^\star_1)+\mathcal{E}^\dagger(\theta^\star_1))(B(\theta^\star_1)+C(\theta^\star_1)+\mathcal{E}(\theta^\star_1))))\leq \frac{9}{4}p_{\rm tail}^2,\\
&\mathrm{Tr}((B^\dagger(\widetilde{\theta}^\star_{1})+C^\dagger(\widetilde{\theta}^\star_{1})+\mathcal{E}^\dagger(\widetilde{\theta}^\star_{1}))(B(\widetilde{\theta}^\star_{1})+C(\widetilde{\theta}^\star_{1})+\mathcal{E}(\widetilde{\theta}^\star_{1})))\leq \frac{9}{4}p_{\rm tail}^2\,.
\end{aligned}
\end{equation}
\item For the second term, we have
\begin{equation}
\begin{aligned}
&2\mathrm{Re}\left(\mathrm{Tr}(A^\dagger(\theta^\star_1)(B(\theta^\star_1)+C(\theta^\star_1)+\mc{E}(\theta^\star_1))\right)-2\mathrm{Re}\left(\mathrm{Tr}(A^\dagger(\widetilde{\theta}^\star_{1})(B(\widetilde{\theta}^\star_{1})+C(\widetilde{\theta}^\star_{1})+\mc{E}(\widetilde{\theta}^\star_1)))\right)\\
\leq &~2\left\|A(\theta^\star_1)-A(\widetilde{\theta}^\star_{1})\right\|_F\left\|(B(\theta^\star_1)+C(\theta^\star_1)+\mc{E}(\theta^\star_1))\right\|_F\\
&+2\left\|A(\widetilde{\theta}^\star_1)\right\|_F\left\|(B(\theta^\star_1)+C(\theta^\star_1)+\mc{E}(\theta^\star_1))-(B(\widetilde{\theta}^\star_1)+C(\widetilde{\theta}^\star_1)+\mc{E}(\widetilde{\theta}^\star_1))\right\|_F\,.
\end{aligned}
\end{equation}
Similar to the calculation in~\cref{eqn:small_diff}, we obtain from~\cref{eqn:A_theta_diff} that
\begin{equation}
\begin{aligned}
\left\|A(\theta^\star_1)-A(\widetilde{\theta}^\star_{1})\right\|_F\leq &~ T\left|\theta^\star_1-\widetilde{\theta}^\star_{1}\right|\left\|\Phi_{:,\mc{D}_{\pi_1}}\left(\Psi_{:,\mc{D}_{\pi_1}}\right)^\dagger\right\|_F+2\sqrt{KLR}\delta\cdot T\\
\leq &~ T\left|\theta^\star_1-\widetilde{\theta}^\star_{1}\right|\left\|\Phi_{:,\mc{D}_{\pi_1}}\left(\Psi_{:,\mc{D}_{\pi_1}}\right)^\dagger\right\|_F+\frac{p_{\rm tail}}{2}\,,
\end{aligned}
\end{equation}
where we use $\delta\cdot T=\mathcal{O}(p_{\rm tail}/\sqrt{KLR})$ in the last inequality.
According to \cref{lem:G_theta_property}~\cref{eqn:G_theta_lipschitz} and \cref{lem:truncation_error}~\cref{eqn:random_error_2}, we have
\begin{equation}
\left\|B(\theta)-B(\theta')\right\|_F\leq p_{\rm tail}\sigma T|\theta-\theta'|/2,\quad \Pr\left[\forall \theta\ne \theta'\in \Theta:\left\|\mathcal{E}(\theta)-\mathcal{E}(\theta')\right\|_F\leq p_{\rm tail}\sigma T|\theta-\theta'|/2\right]\geq 1-\frac{\eta}{2}\,.
\end{equation}
And we have
\begin{equation}
\begin{aligned}
&\left\|C({\theta^\star_1})-C({\widetilde{\theta}^\star_1})\right\|_F\\
\leq &~\sum_{i\neq \pi_1}\left\|\Phi_{:,\mathcal{D}_i}\left(\exp(-({\theta^\star_1}-\lambda^\star_{i})^2T^2)-\exp(-({\widetilde{\theta}^\star_1}-\lambda^\star_i)^2T^2)\right)\Psi^\dagger_{:,\mathcal{D}_i}\right\|_F\\
&+\sum_{i\neq \pi_1}\left\|\Phi_{:,\mathcal{D}_i}\mathrm{diag}\left(\left\{\exp(-({\theta^\star_1}-\lambda_{m})^2T^2)-\exp(-({\theta^\star_1}-\lambda^\star_i)^2T^2)\right\}_{m\in\mc{D}_i}\right)\Psi^\dagger_{:,\mathcal{D}_i}\right\|_F\\
&+\sum_{i\neq \pi_1}\left\|\Phi_{:,\mathcal{D}_i}\mathrm{diag}\left(\left\{\exp(-({\widetilde{\theta}^\star_1}-\lambda_{m})^2T^2)-\exp(-({\widetilde{\theta}^\star_1}-\lambda^\star_i)^2T^2)\right\}_{m\in\mc{D}_i}\right)\Psi^\dagger_{:,\mathcal{D}_i}\right\|_F\\
\leq &~\sum_{i\neq \pi_1}\left\|\Phi_{:,\mathcal{D}_i}\left(\exp(-({\theta^\star_1}-\lambda^\star_{i})^2T^2)-\exp(-({\widetilde{\theta}^\star_1}-\lambda^\star_i)^2T^2)\right)\Psi^\dagger_{:,\mathcal{D}_i}\right\|_F\\
&+2\delta\cdot T\sqrt{KLR}\\
\leq &~\sum_{i\neq \pi_1}\left|\exp(-({\theta^\star_1}-\lambda^\star_i)^2T^2)-\exp(-({\widetilde{\theta}^\star_1}-\lambda^\star_i)^2T^2)\right|\left\|\Phi_{:,\mathcal{D}_i}\Psi^\dagger_{:,\mathcal{D}_i}\right\|_F+\frac{p^2_{\rm tail}}{4\sqrt{KLR}}\\
\leq &~\frac{p^2_{\rm tail}}{4\sqrt{KLR}} + o(p_{\rm tail}T|\theta-\theta'|)\,,
\end{aligned}
\end{equation}
where we use a similar calculation as~\cref{eqn:small_diff} and $\delta\cdot T=\mathcal{O}(p^{2}_{\rm tail}/KLR)$ in the second step. And the third step follows from
\begin{align*}
\left(\exp(-({\theta^\star_1}-\lambda^\star_i)^2T^2)-\exp(-({\widetilde{\theta}^\star_1}-\lambda^\star_i)^2T^2)\right)\leq &~ T|{\theta^\star_1}-{\widetilde{\theta}^\star_1}|\left(\sup_{|x|\geq \min\{|({\theta^\star_1}-\lambda^\star_i)T|,|({\widetilde{\theta}^\star_1}-\lambda^\star_i)T|\}}2x\exp(-x^2)\right)\\
=&~o\left(\frac{p_{\rm tail}T|{\theta^\star_1}-{\widetilde{\theta}^\star_1}|}{K\|\Phi_{:,\mathcal{D}_i}\Psi^\dagger_{:,\mathcal{D}_i}\|_F}\right)\,,
\end{align*}
where the last inequality comes from $\min\{|({\theta^\star_1}-\lambda^\star_i)T|,|({\widetilde{\theta}^\star_1}-\lambda^\star_i)T|\}{\geq T\Delta /4}$ and $T=\Omega\left(\frac{1}{\Delta}\log\left(\frac{KLR}{p_{\rm tail}}\right)\right)$.

Plugging these three inequalities into the second term and using~\cref{eqn:A_bound} and~\cref{eqn:B_c_bound}, we have  probability at least $1-\eta$ such that
\begin{equation}\label{eqn:A_B_C_inter}
\begin{aligned}
&2\mathrm{Re}\left(\mathrm{Tr}(A^\dagger(\theta^\star_1)(B(\theta^\star_1)+C(\theta^\star_1)+\mc{E}(\theta^\star_1)))\right)-2\mathrm{Re}\left(\mathrm{Tr}(A^\dagger(\widetilde{\theta}^\star_{1})(B(\widetilde{\theta}^\star_{1})+C(\widetilde{\theta}^\star_{1})+\mc{E}(\widetilde{\theta}^\star_1))))\right)\\
\leq &~ 2\left(T\left|\theta^\star_1-\widetilde{\theta}^\star_{1}\right|\left\|\Phi_{:,\mc{D}_{\pi_1}}\left(\Psi_{:,\mc{D}_{\pi_1}}\right)^\dagger\right\|_F+\frac{p_{\rm tail}}{2}\right)\cdot \frac{5}{4}p_{\rm tail}\\
& +2\cdot \frac{5}{4}\left\|\Phi_{:,\mc{D}_{\pi_1}}\left(\Psi_{:,\mc{D}_{\pi_1}}\right)^\dagger\right\|_F\cdot  \left(2p_{\rm tail}\sigma T \left|\theta^\star_1-\widetilde{\theta}^\star_{1}\right|+\frac{p^2_{\rm tail}}{4\sqrt{KLR}}\right)\\
\leq &~ (3+5\sigma)\left\|\Phi_{:,\mc{D}_{\pi_1}}\left(\Psi_{:,\mc{D}_{\pi_1}}\right)^\dagger\right\|_FT\left|\theta^\star_1-\widetilde{\theta}^\star_{1}\right|p_{\rm tail}+\frac{15}{8}p^2_{\rm tail}\\
\leq &~(3+5\sigma)\left\|\Phi_{:,\mc{D}_{\pi_1}}\left(\Psi_{:,\mc{D}_{\pi_1}}\right)^\dagger\right\|_FT
\left(\left|\theta^\star_1-\lambda^\star_{\pi_1}\right| + \frac{q}{T} + \delta \right)
p_{\rm tail} +\frac{15}{8}p^2_{\rm tail}\\
\leq &~(3+5\sigma)\left\|\Phi_{:,\mc{D}_{\pi_1}}\left(\Psi_{:,\mc{D}_{\pi_1}}\right)^\dagger\right\|_FT\left|\theta^\star_1-\lambda^\star_{\pi_1}\right|p_{\rm tail}+\frac{23}{8}p^2_{\rm tail}\,,
\end{aligned}
\end{equation}
where we use $\left\|\Phi_{:,\mc{D}_{\pi_1}}\left(\Psi_{:,\mc{D}_{\pi_1}}\right)^\dagger\right\|_F\leq \sqrt{KLR}$ (according to similar calculation as~\cref{eqn:W_F_other_domin_upper_bound}) in the second inequality and $q=\mc{O}(p_{\rm tail}/((1+\sigma)\sqrt{KLR}))$ and $\delta=\mc{O}(p_{\rm tail}/(T(1+\sigma)\sqrt{KLR}))$ in the last inequality.
\end{itemize}

Now, plugging $\theta=\theta^\star_1$ and $\theta=\widetilde{\theta}^\star_{1}$ into $\mc{W}^2(\theta)$ and using~\cref{eqn:Phi_Psi_lower_bound,eqn:A_theta_diff}, we have
\begin{equation}\label{eqn:W_diff}
\begin{aligned}
0\leq &~\mc{W}^2(\theta^\star_1)-\mc{W}^2\left(\widetilde{\theta}^\star_{1}\right)\\
\leq &~\underbrace{\left(\exp\left(-2\left(\theta^\star_1-\lambda^\star_{\pi_1}\right)^2T^2\right)-\exp\left(-2\left(\widetilde{\theta}^\star_1-\lambda^\star_{\pi_1}\right)^2T^2\right)\right)}_{\leq 0~\text{by the definition of} ~\widetilde{\theta}^\star_1}\left\|\Phi_{:,\mc{D}_{\pi_1}}\left(\Psi_{:,\mc{D}_{\pi_1}}\right)^\dagger\right\|^2_F\\
&+\underbrace{6KLR\delta\cdot T}_{\leq p^2_{\rm tail}} +(3+5\sigma)\left\|\Phi_{:,\mc{D}_{\pi_1}}\left(\Psi_{:,\mc{D}_{\pi_1}}\right)^\dagger\right\|_FT\left|\theta^\star_1-\lambda^\star_{\pi_1}\right|p_{\rm tail}+\frac{23}{8}p_{\rm tail}^2 +\underbrace{\frac{9}{2}p_{\rm tail}^2}_{\text{last term}}\,.
\end{aligned}
\end{equation}
where we use the condition that $\delta \cdot T=\mathcal{O}\left(p^2_{\rm tail}/KLR\right)$.
Because we have already shown $\left|\theta^\star_1-\lambda^\star_{\pi_1}\right|=\mc{O}(1/T)$ (noticing $\chi=\mc{O}(1)$) in the first step, let $\delta_\theta=T\left|\theta^\star_1-\lambda^\star_{\pi_1}\right|$, the above inequality implies that
\begin{equation}\label{eqn:delta_theta_inquality}
-C\delta^2_\theta+\left((3+5\sigma)p_{\rm tail}/\left\|\Phi_{:,\mc{D}_{\pi_1}}\left(\Psi_{:,\mc{D}_{\pi_1}}\right)^\dagger\right\|_F\right)\delta_\theta+\frac{67}{8}\frac{p_{\rm tail}^2}{\left\|\Phi_{:,\mc{D}_{\pi_1}}\left(\Psi_{:,\mc{D}_{\pi_1}}\right)^\dagger\right\|_F^2}\geq 0\,,
\end{equation}
where $C$ is a uniform constant. Thus, we have that
\begin{equation}
\left|\theta^\star_1-\lambda_{\pi_1}^{\star}\right|=\mc{O}\left((1+\sigma)\frac{p_{\rm tail}}{\left\|\Phi_{:,\mc{D}_{\pi_1}}\left(\Psi_{:,\mc{D}_{\pi_1}}\right)^\dagger\right\|_F}\frac{1}{T}\right)=\mc{O}\left((1+\sigma)\frac{p_{\rm tail}}{p_{\min}}\frac{1}{T}\right)\,,
\end{equation}
{where} we use~\cref{eqn:Phi_Psi_lower_bound} and $\chi=\mathcal{O}(1)$ for the last equality.
This concludes the first part of the proof.

\textbf{Step 3:} We show that $m_{{j}}=|\mc{D}_{\pi_j}|$ when $1\leq j\leq I$.

Without loss of generality, we only consider $i=1$. The other cases can also be shown similarly. Define $s_k\left(A\right)$ as the $k$-th singular value of $A$. Recall
\begin{equation}
\begin{aligned}
  G(\theta^\star_1)=&~\underbrace{\Phi_{:,\mc{D}_{\pi_1}} \cdot \diag\left(\left\{\exp(-(\theta^\star_1-\lambda_m)^2T^2)\right\}_{m\in \mc{D}_{\pi_1}}\right) \cdot \left(\Psi_{:,\mc{D}_{\pi_1}}\right)^\dagger}_{:=A(\theta^\star_1)}\\
  +&\underbrace{\Phi_{:,\mc{D}^c} \cdot \diag\left(\left\{\exp(-(\theta^\star_1-\lambda_m)^2T^2)\right\}_{m\in \mc{D}^c}\right) \cdot \left(\Psi_{:,\mc{D}^c}\right)^\dagger}_{:=B(\theta^\star_1)}\\
  +&\underbrace{\sum_{i\neq \pi_1}\Phi_{:,\mc{D}_{i}} \cdot \diag\left(\left\{\exp(-(\theta^\star_1-\lambda_m)^2T^2)\right\}_{m\in \mc{D}_{i}}\right) \cdot \left(\Psi_{:,\mc{D}_{i}}\right)^\dagger}_{:=C(\theta^\star_1)}+\mc{E}(\theta^\star_1)\,.
\end{aligned}
\end{equation}
With probability at least $1-\eta$, we have
\begin{equation}
\left\|B(\theta^\star_1)+C(\theta^\star_1)+\mc{E}(\theta^\star_1)\right\|_2\leq \left\|B(\theta^\star_1)+C(\theta^\star_1)+\mc{E}(\theta^\star_1)\right\|_F\leq 2p_{\rm tail}\,.
\end{equation}
Furthermore, for $1\leq k\leq |\mc{D}_{\pi_1}|$,~\cref{eqn:location_error} implies that 
\begin{equation}\label{eqn:lower_bound_singularvalue}
\begin{aligned}
s_k\left(A(\theta^\star_1)\right)=&~\exp\left(-\mathcal{O}\left((1+\sigma)(p_{\rm tail}/p_{\min})+\delta\cdot T\right)^2\right)\left(\lambda_k\left(\Psi_{:,\mc{D}_{\pi_1}}\Phi^\dagger_{:,\mc{D}_{\pi_1}}\Phi_{:,\mc{D}_{\pi_1}}\left(\Psi_{:,\mc{D}_{\pi_1}}\right)^\dagger\right)\right)^{1/2}\\
\geq &~ \exp\left(-\mathcal{O}\left((1+\sigma)(p_{\rm tail}/p_{\min})+\delta\cdot T\right)^2\right)\left(\lambda_{\min}\left(\Phi^\dagger_{:,\mc{D}_{\pi_1}}\Phi_{:,\mc{D}_{\pi_1}}\right)\right)^{1/2}\left(\lambda_k\left(\Psi_{:,\mc{D}_{\pi_1}}\left(\Psi_{:,\mc{D}_{\pi_1}}\right)^\dagger\right)\right)^{1/2}\\
\geq &~ \exp\left(-\mathcal{O}\left((1+\sigma)(p_{\rm tail}/p_{\min})+\delta\cdot T\right)^2\right)\left(\lambda_{\min}\left(\Phi^\dagger_{:,\mc{D}_{\pi_1}}\Phi_{:,\mc{D}_{\pi_1}}\right)\right)^{1/2}\left(\lambda_{\min}\left(\Psi^\dagger_{:,\mc{D}_{\pi_1}}\Psi_{:,\mc{D}_{\pi_1}}\right)\right)^{1/2}
\\
\geq &~\exp\left(-\mathcal{O}\left((1+\sigma)(p_{\rm tail}/p_{\min})+\delta\cdot T\right)^2\right)\left(\frac{1}{1+\chi}\right)^2\frac{\left\|\Phi_{:,\mc{D}_i}\right\|_F\left\|\Psi_{:,\mc{D}_i}\right\|_F}{|\mc{D}_i|}\\
=&~ \exp\left(-\mathcal{O}\left((1+\sigma)(p_{\rm tail}/p_{\min})+\delta\cdot T\right)^2\right)\left(\frac{1}{1+\chi}\right)^2\frac{\sqrt{\sum_{m\in \mc{D}_i}\|\Phi_{:,m}\|^2}\sqrt{\sum_{m\in \mc{D}_i}\|\Psi_{:,m}\|^2}}{|\mc{D}_i|}\\
\geq&~ \exp\left(-\mathcal{O}\left((1+\sigma)(p_{\rm tail}/p_{\min})+\delta\cdot T\right)^2\right)\left(\frac{1}{1+\chi}\right)^2\frac{\sum_{m\in \mc{D}_i} \|\Phi_{:,m}\|\|\Psi_{:,m}\|}{|\mc{D}_i|}\\
\geq &~ \exp\left(-\mathcal{O}\left((1+\sigma)(p_{\rm tail}/p_{\min})+\delta\cdot T\right)^2\right)\left(\frac{1}{1+\chi}\right)^2p_{\min}\\
\geq &~\frac{1-\mc{O}\left(\left((1+\sigma)(p_{\rm tail}/p_{\min})+\delta\cdot T\right)^2\right)}{(1+\chi)^2}p_{\min}
\end{aligned}
\end{equation}
where we use $\delta=\mc{O}(1/T)$,~\cref{asp:linear_dep}, and the definition of $p_{\min}$.

Because
\begin{equation}
\frac{1-\mc{O}\left(\left((1+\sigma)(p_{\rm tail}/p_{\min})+\delta\cdot T\right)^2\right)}{(1+\chi)^2}=\Omega\left(\frac{p_{\rm tail}}{p_{\min}}\right)\,,
\end{equation}
we have
\begin{equation}
 \inf_{1\leq k\leq |\mc{D}_1|}s_k\left(G(\theta^\star_1)\right)\geq \inf_{1\leq k\leq |\mc{D}_1|}s_k\left(A(\theta^\star_1)\right)-2p_{\rm tail}>\tau\geq 2p_{\rm tail}\geq \sup_{k> |\mc{D}_1|}s_k\left(G(\theta^\star_1)\right)
\end{equation}
This guarantees that $m_1=|\mc{D}_{\pi_1}|$.

\textbf{Step 4:} We show that $m_i=0$ when $i>I$.

For $i>I$, because of the block intervals from $1\leq i\leq I$, we first have
\begin{equation}
\inf_{j} \mathrm{dist}(\theta^\star_i,{\mc{I}_j})\geq \frac{\alpha}{T}-\mc{O}\left((1+\sigma)\left(\frac{p_{\rm tail}}{p_{\min}}\right)\frac{1}{T}\right)-\delta/2\geq \frac{\alpha}{2T}\,,
\end{equation}
where we use $\alpha=\Omega(\chi^2+(1+\sigma)^2(p_{\rm tail}/p_{\min})^2)$ and $\alpha=\mc{O}(\Delta\cdot T)$ in the last inequality. Using~\cref{lem:G_theta_property}~\cref{eqn:G_theta_bound}, we have
\begin{equation}
\begin{aligned}
  G(\theta^\star_i)=&~\underbrace{\Phi_{:,\mc{D}} \cdot \diag\left(\left\{\exp(-(\theta^\star_i-\lambda_m)^2T^2)\right\}_{m\in \mc{D}}\right) \cdot \left(\Psi_{:,\mc{D}}\right)^\dagger}_{\|\cdot\|_2\leq \sqrt{KLR}\exp(-\alpha^2/4)}\\
  +&\underbrace{\Phi_{:,\mc{D}^c} \cdot \diag\left(\left\{\exp(-(\theta^\star_i-\lambda_m)^2T^2)\right\}_{m\in \mc{D}^c}\right) \cdot \left(\Psi_{:,\mc{D}^c}\right)^\dagger+\mc{E}(\theta^\star_i)}_{\|\cdot\|_2\leq p_{\rm tail}}\,.
\end{aligned}
\end{equation}
Because $\alpha=\Omega(\log(KLR/p_{\rm tail}))$, we have
\begin{equation}
\left\|  G(\theta^\star_i)\right\|_2\leq 2p_{\rm tail}<\tau\,.
\end{equation}
for $i>I$. This implies that $m_i=0$ for $i>I$ and concludes the proof.
\end{proof}

\subsection{Useful lemmas}
\begin{lemma}\label{lem:truncation_error} Given $\beta>0$. Let $\{t_n\}$ be i.i.d. sampled from the truncated Gaussian $a_T(t)$ defined as~\cref{eqn:a_T}. Define
\begin{equation}
\mc{E}\left(\theta\right):=\frac{1}{N}\sum^N_{n=1}Z_n\exp(\i \theta t_n)-{\Phi} \cdot \diag\left(\{\exp(-(\theta-\lambda_m)^2T^2{)}\}_{m\in [M]}\right) \cdot {\Psi^\dagger}\in \C^{L\times R}\,.
\end{equation}
Let $\Theta:=\left\{\theta_j\right\}^J_{j=1}\cup \left\{\lambda_m\right\}_{m\in\mathcal{D}}$.

If $\sigma=\Omega\left(\log^{1/2}\left(\sqrt{LR}/\beta\right)\right)$ as in~\cref{eqn:a_T} and $N=\Omega\left(\frac{LR}{\beta^2}\log\left(\left(J+|\mathcal{D}|\right)\frac{LR}{\eta}\right)\right)$,  we have
\begin{equation}\label{eqn:random_error}
\Pr\left[\max_{\theta\in\Theta}\|\mc{E}(\theta)\|_F\leq \beta\right]\geq 1-\eta\,.
\end{equation}
and
\begin{equation}\label{eqn:random_error_2}
\Pr\left[\bigcap_{\theta,\theta'\in\Theta}\|\mc{E}(\theta)-\mc{E}(\theta')\|_F\leq \beta\sigma T|\theta-\theta'|\right]\geq 1-\eta\,.
\end{equation}
\end{lemma}
\begin{proof} For each $1\leq i\leq L$ and $1\leq j\leq R$, by \cref{lem:bound_E_entrywise} with $\wt{\beta}=\frac{\beta}{\sqrt{LR}}$ and $\wt{\eta}=\frac{\eta}{LR}$, we have
\begin{equation}
\Pr\left[\max_{\theta\in \Theta}|\mc{E}_{i,j}(\theta)|\leq \frac{\beta}{\sqrt{LR}}\right]\geq 1-\frac{\eta}{LR}\,,
\end{equation}
and
\begin{equation}
\Pr\left[\bigcap_{\theta,\theta'\in\Theta}\left\{|\mc{E}_{i,j}(\theta)-\mc{E}_{i,j}(\theta')|\leq \frac{\beta}{\sqrt{LR}}\sigma T|\theta-\theta'|\right\}\right]\geq 1-\frac{\eta}{LR}\,.
\end{equation}
By union bound, the first equation implies that
\begin{equation}
\Pr\left[\max_{\theta\in \Theta}\|\mc{E}(\theta)\|_F=\max_{\theta\in \Theta}\sqrt{\sum_{i,j}|\mc{E}_{i,j}(\theta)|^2}\leq \beta\right]\geq \Pr\left[\bigcap_{i,j}\left\{\max_{\theta\in \Theta}|\mc{E}_{i,j}(\theta)|\leq \frac{\beta}{\sqrt{LR}}\right\}\right]\geq 1-\eta\,.
\end{equation}
And the second equation implies
\begin{equation}
\begin{aligned}
&\Pr\left[\bigcap_{\theta,\theta'\in \Theta}\|\mc{E}(\theta)-\mc{E}(\theta')\|_F=\sqrt{\sum_{i,j}|\mc{E}_{i,j}(\theta)-\mc{E}_{i,j}(\theta')|^2}\leq \beta \sigma T|\theta-\theta'|\right]\\
\geq &~\Pr\left[\bigcap_{i,j}\left\{\bigcap_{\theta,\theta'\in \Theta}|\mc{E}_{i,j}(\theta)-\mc{E}_{i,j}(\theta')|\leq \frac{\beta}{\sqrt{LR}}\sigma T|\theta-\theta'|\right\}\right]\geq 1-\eta\,.
\end{aligned}
\end{equation}
This completes the proof of the lemma.
\end{proof}
\begin{lemma}\label{lem:G_theta_property} Given any matrix $\Phi\in\mathbb{C}^{L\times M}$, $\Psi\in\mathbb{C}^{R\times M}$, and $\{\lambda_m\}_{m\in[M]}$. Define
\begin{equation}
G(\theta)=\Phi \cdot \diag\left(\{\exp(-(\theta-\lambda_m)^2T^2)\}_{m\in [M]}\right) \cdot \Psi^\dagger.
\end{equation}
We have
\begin{equation}\label{eqn:G_theta_bound}
\|G(\theta)\|_F\leq \left(\max_m \exp\left(-(\theta-\lambda_m)^2T^2\right)\right)\sum_{m\in [M]}\sqrt{\left(\sum^L_{l=1} |\Phi_{l,m}|^2\right)\left(\sum^R_{r=1} |\Psi_{r,m}|^2\right)}
\end{equation}
and
\begin{equation}\label{eqn:G_theta_lipschitz}
\begin{aligned}
&\|G(\theta)-G(\theta')\|_F\leq T|\theta-\theta'|\sum_{m\in [M]}\left\|\Phi_{:,m}\Psi^\dagger_{:,m}\right\|_F\\
\leq &{\sqrt{2/e}}T|\theta-\theta'|\sum_{m\in [M]}\sqrt{\left(\sum^L_{l=1} |\Phi_{l,m}|^2\right)\left(\sum^R_{r=1} |\Psi_{r,m}|^2\right)}
\end{aligned}
\end{equation}
\end{lemma}
\begin{proof} First, we note that
\begin{equation}
\begin{aligned}
\left\|G(\theta)\right\|_F&\leq \sum_{m\in [M]}\left\|\Phi_{:,m}\exp(-(\theta-\lambda_m)^2T^2)\Psi^\dagger_{:,m}\right\|_F\\
&\leq \left(\max_m \exp\left(-(\theta-\lambda_m)^2T^2\right)\right)\sum_{m\in [M]}\left\|\Phi_{:,m}\Psi^\dagger_{:,m}\right\|_F\\
&=\left(\max_m \exp\left(-(\theta-\lambda_m)^2T^2\right)\right)\sum_{m\in [M]}\sqrt{\left(\sum^L_{l=1} |\Phi_{l,m}|^2\right)\left(\sum^R_{r=1} |\Psi_{r,m}|^2\right)}
\end{aligned}
\end{equation}
This concludes the proof of~\cref{eqn:G_theta_bound}.

Second, we have
\begin{equation}
\begin{aligned}
&\left\|G(\theta)-G(\theta')\right\|_F\\
\leq &~\sum_{m\in [M]}\left\|\Phi_{:,m}\left(\exp(-(\theta-\lambda_m)^2T^2)-\exp(-(\theta'-\lambda_m)^2T^2)\right)\Psi^\dagger_{:,m}\right\|_F\\
\leq &~\left(\max_m \left(\exp(-(\theta-\lambda_m)^2T^2)-\exp(-(\theta'-\lambda_m)^2T^2)\right)\right)\sum_{m\in [M]}\left\|\Phi_{:,m}\Psi^\dagger_{:,m}\right\|_F\\
\leq &{\sqrt{2/e}}T|\theta-\theta'|\sum_{m\in [M]}\left\|\Phi_{:,m}\Psi^\dagger_{:,m}\right\|_F\\
=&~{\sqrt{2/e}}T|\theta-\theta'|\sum_{m\in [M]}\sqrt{\left(\sum^L_{l=1} |\Phi_{l,m}|^2\right)\left(\sum^R_{r=1} |\Psi_{r,m}|^2\right)}\,,
\end{aligned}
\end{equation}
where the last step follows from the {Lipschitzness} of $\exp(-x^2)$.

This concludes the proof of~\cref{eqn:G_theta_lipschitz}.
\end{proof}

\begin{lemma}[{\cite[Lemma A.1]{ding2024quantum}}]\label{lem:bound_E_entrywise}
Given $\widetilde{\beta}>0$, $i\in [L]$, $j\in [R]$. Let $\{t_n\}$ to be i.i.d sampled from the truncated Gaussian $a_T^{\rm trunc}(t)$ defined as~\cref{eqn:a_T}. Define
\begin{align}
    \mc{E}_{i,j}(\theta):=\frac{1}{N}\sum^N_{n=1}(Z_n)_{i,j}\exp(\i \theta t_n)-\Phi_{i,:} \cdot \diag\left(\{\exp(-(\theta-\lambda_m)^2T^2{)}\}_{m\in [M]}\right) \cdot (\Psi^\dagger)_{:,j}\in \C\,.
\end{align}
Let $\Theta:=\left\{\theta_j\right\}^J_{j=1}\cup \left\{\lambda_m\right\}_{m\in\mathcal{D}}$. If $\sigma=\Omega\left(\log^{1/2}\left(1/\wt{\beta}\right)\right)$  and $N=\Omega\left(\frac{1}{\wt{\beta}^2}\log\left(\left(J+|\mathcal{D}|\right)\frac{1}{\wt{\eta}}\right)\right)$,  we have
\begin{align}
    \Pr\left[\max_{\theta\in \Theta}|\mc{E}_{i,j}(\theta)|>\wt{\beta}\right]\leq {\wt{\eta}}\,,
\end{align}
and\footnote{In fact,~\cite[Lemma A.1]{ding2024quantum} has extra $\beta^2$ on the RHS of the inequality inside the probability in the equation below, which is unnecessary.}
\begin{align}
    \Pr\left[\bigcap_{\theta,\theta'\in \Theta}|\mc{E}_{i,j}(\theta)-\mc{E}_{i,j}(\theta')|\leq \wt{\beta}\sigma T|\theta-\theta'|\right]\geq 1-\wt{\eta}\,.
\end{align}
\end{lemma}

\section{Analysis of the observable estimation}\label{sec:ana_ob}

In this section, we give the proof of~\cref{prop:observable}. In {addition} to the {assumptions} in~\cref{thm:non_orthogonal_location}, we assume $\delta=0,p_{\rm tail}=0$, $p_{\min}=\Omega(1)$, and $K,L,R=\mathcal{O}(1)$. Following the proof of~\cref{thm:non_orthogonal_location}, when $T=\widetilde{\Omega}(\Delta^{-1})$ and $N=\Omega(1)$, we first have
\begin{equation}
\sup_{1\leq j\leq I}|\theta^\star_j-\lambda^\star_{\pi(j)}|\leq \Delta/4\,.
\end{equation}
This implies $\theta^\star_j$ is at least $\Delta/2$ away from $\lambda^\star_i$ when $i\neq \pi(j)$.

Now, we improve the above bound and aim to show
\begin{equation}\label{eqn:error_for_ob}
\sup_{1\leq j\leq I}|\theta^\star_j-\lambda^\star_{\pi(j)}|=\mathcal{O}\left(\frac{1}{T}\left(\exp(-\Delta^2T^2/32)+N^{-1/2}\right)\right)\,.
\end{equation}
Without loss of generality, we consider $j=1$. Because $p_{\rm tail}=0$, we have
\begin{equation}\label{eqn:G_error}
\begin{aligned}
  G(\theta)=&\underbrace{\Phi_{:,\mc{D}_{\pi_1}} \cdot \diag\left(\left\{\exp(-(\theta-\lambda_m)^2T^2)\right\}_{m\in \mc{D}_{\pi_1}}\right) \cdot \left(\Psi_{:,\mc{D}_{\pi_1}}\right)^\dagger}_{:=A(\theta)}\\
  +&\underbrace{\Phi_{:,\mc{D}^c} \cdot \diag\left(\left\{\exp(-(\theta-\lambda_m)^2T^2)\right\}_{m\in \mc{D}^c}\right) \cdot \left(\Psi_{:,\mc{D}^c}\right)^\dagger}_{=0}\\
  +&\underbrace{\sum_{i\neq \pi_1}\Phi_{:,\mc{D}_{i}} \cdot \diag\left(\left\{\exp(-(\theta-\lambda_m)^2T^2)\right\}_{m\in \mc{D}_{i}}\right) \cdot \left(\Psi_{:,\mc{D}_{i}}\right)^\dagger}_{:=C(\theta)}+\mc{E}(\theta)\,.
\end{aligned}
\end{equation}
Because $K,L,R=\mathcal{O}(1)$ and $|\theta^\star_1-\lambda^\star_j|\geq \Delta/2$ for $j\neq \pi_1$, we first have $|\theta^\star_1-\lambda_m|\geq \Delta/4$ if the dominant index $m\notin \mathcal{D}_{\pi_1}$. This implies
\begin{equation}
\left\|C(\theta)\right\|_F=\mathcal{O}\left(\exp\left(-\Delta^2 T^2/16\right)\right)\,,
\end{equation}
and
\begin{equation}
\left\|G(\theta)-A(\theta)-\mc{E}(\theta)\right\|_F=\mathcal{O}\left(\exp\left(-\Delta^2 T^2/16\right)\right)\,.
\end{equation}
Using~\cref{lem:bound_E_entrywise}, with high probability, we also have
\begin{equation}
{\|\mc{E}(\theta)\|_F=\widetilde{\mathcal{O}}\left(N^{-1/2}\right)},\quad \|\mc{E}(\theta)-\mc{E}(\theta')\|_F=\mathcal{O}\left(\sigma T|\theta-\theta'|N^{-1/2}\right)\,.
\end{equation}
Let $\tilde{\theta}_1^\star$ be the closest energy grid point to the desired cluster center $\lambda_{\pi_1}^{\star}$. Similar to~\cref{eqn:A_B_C_inter} in the second step of the proof of~\cref{thm:non_orthogonal_location}, we have
\begin{equation}
    \begin{aligned}
    &2\mathrm{Re}\left(\mathrm{Tr}(A^\dagger(\theta^\star_1)\mc{E}(\theta^\star_1))\right)-2\mathrm{Re}\left(\mathrm{Tr}(A^\dagger(\widetilde{\theta}^\star_{1})\mc{E}(\widetilde{\theta}^\star_{1}))\right)\\
    =&~\mathcal{O}\left(T\left|\theta^\star_1-\widetilde{\theta}^\star_{1}\right|N^{-1/2}+ \sigma T \left|\theta^\star_1-\widetilde{\theta}^\star_{1}\right|N^{-1/2}\right)
    \\
    =&~\mathcal{O}\left((1+\sigma)T \left(\left|\theta^\star_1-\lambda^\star_{\pi_1}\right| +\frac{q}{T} \right) N^{-1/2}\right)\,.
    \end{aligned}
\end{equation}
Choosing $q=\widetilde{\mathcal{O}}(N^{-1/2})$, similar to~\cref{eqn:W_diff}, we have
\begin{equation}
\begin{aligned}
&\underbrace{\left(\exp\left(-2\left(\widetilde{\theta}^\star_1-\lambda^\star_{\pi_1}\right)^2T^2\right)-\exp\left(-2\left(\theta^\star_1-\lambda^\star_{\pi_1}\right)^2T^2\right)\right)}_{\leq 0}\\
=&~\mathcal{O}\left(T\left|\theta^\star_1-\lambda^\star_{\pi_1}\right|N^{-1/2}+\frac{1}{N}+\exp\left(-\Delta^2 T^2/16\right)\right)\,.
\end{aligned}
\end{equation}
Following a similar calculation as~\cref{eqn:delta_theta_inquality}, we have
\begin{equation}
T\left|\theta^\star_1-\lambda^\star_{\pi_1}\right|=\mathcal{O}\left(\frac{1}{\sqrt{N}}+\exp\left(-\Delta^2 T^2/32\right)\right)\,.
\end{equation}
We conclude {the} proof of~\cref{eqn:error_for_ob}.

Next, we consider the observation data matrix error. For a fixed $i$, we define the exact data matrix for the generalized eigenvalue problem:
\begin{equation}
G_{\rm exact}=\Phi_{:,{\mc{D}_i}}(\Psi_{:,{\mc{D}_i}})^\dagger,\quad G^O_{\rm exact}=\Phi_{:,{\mc{D}_i}}\cdot {O_{\mc{D}_i}}\cdot (\Psi_{:,{\mc{D}_i}})^\dagger\,,
\end{equation}
where {$O_{\mc{D}_i}$} is defined in~\cref{eqn:O_D_i}. According to~\cref{eqn:error_for_ob},~\cref{lem:truncation_error}, and following a similar calculation as~\cref{eqn:error_calculation}, we have
\begin{equation}\label{eqn:norm_error}
\begin{aligned}
&\left\|G(\theta^\star_i)-G_{\rm exact}\right\|_2=\mathcal{O}\left(\frac{1}{\sqrt{N}}+\exp\left(-\Delta^2 T^2/16\right)\right)\,,\\
&\left\|G^O(\theta^\star_i)-G^O_{\rm exact}\right\|_2=\mathcal{O}\left(\|O\|\left(\frac{1}{\sqrt{N}}+\exp\left(-\Delta^2 T^2/16\right)\right)\right)\,,
\end{aligned}
\end{equation}
where $G(\theta^\star_i)$ and $G^O(\theta^\star_i)$ are defined in~\cref{eqn:G,eqn:G_O}, respectively.

According to~\cref{asp:linear_dep}, $G_{\rm exact}$ has rank $m_i$ with smallest nonzero singular value lower bounded by $p_{\min}/(1+\chi)$. Combining this and~\cref{eqn:norm_error}, $|m_i|$-th singular value of $G(\theta^\star_i)$ can be lower bound by $p_{\min}/(2(1+\chi))$ when $T=\Omega(\Delta^{-1})$ and $N=\Omega(1)$. Furthermore, following the third step of the proof of~\cref{thm:non_orthogonal_location}, the $(m_i+1)$-th singular value of $G(\theta^\star_i)$ can be upper bound by $p_{\min}/(4(1+\chi))$ when $T=\Omega(\Delta^{-1})$ and $N=\Omega(1)$. Using singular value perturbation theory, these imply that
\begin{equation}
\begin{aligned}
&\left\|\widetilde{G}(\theta^\star_i)-\widetilde{G}_{\rm exact}\right\|_2=\mathcal{O}\left(\frac{1}{\sqrt{N}}+\exp\left(-\Delta^2 T^2/32\right)\right)\,,\\
&\left\|\widetilde{G}^O(\theta^\star_i)-\widetilde{G}^O_{\rm exact}\right\|_2=\mathcal{O}\left(\|O\|\left(\frac{1}{\sqrt{N}}+\exp\left(-\Delta^2 T^2/32\right)\right)\right)\,,
\end{aligned}
\end{equation}
where $\widetilde{G}_{\rm exact}$ and $\widetilde{G}^O_{\rm exact}$ are defined using the same forms as~\cref{eqn:truncated_matrix}. Here, both $\widetilde{G}(\theta^\star_i)$ and $\widetilde{G}_{\rm exact}$ are full rank matrices with smallest singular value lower bounded by  $p_{\min}/(2(1+\chi))$. According to generalized eigenvalue perturbation theory, we finally have
\begin{equation}
\left|\lambda^O_k-\lambda_{k,\rm exact}\right|=\mathcal{O}\left(\|O\|\left(\frac{1}{\sqrt{N}}+\exp\left(-\Delta^2 T^2/32\right)\right)\right)\,,
\end{equation}
where $\lambda_{k,\rm exact}$ is the $k$-th generalized eigenvalue {with respect to the matrix pencil $(\widetilde{G}^O_{\rm exact}, \widetilde{G}_{\rm exact})$.} Therefore to achieve $\epsilon_O$ precision in $\lambda_k^O$, it is required that $N = \tilde{\Omega}(\epsilon_O^{-2})$ and $T = \tilde{\mc{O}}(\Delta^{-1})$.
This concludes the proof.

\end{widetext}
\bibliography{ref}

\end{document}